\documentclass[sigconf]{acmart}

 \pdfoutput=1 
 
 \let\ACMmaketitle=\maketitle
 \renewcommand{\maketitle}{\begingroup\let\footnote=\thanks 
 \ACMmaketitle\endgroup}
 
\renewcommand\footnotetextcopyrightpermission[1]{} 

\usepackage{soul}

\usepackage{color, colortbl}

\usepackage{moresize}

\usepackage{rotating}
\usepackage{graphicx}
\usepackage{multirow}
\usepackage{geometry}
\usepackage{amsfonts}
\usepackage{amsmath}
\usepackage{amssymb}
\usepackage{mathtools}
\usepackage{wrapfig}
\usepackage{listings}
\usepackage{graphicx}
\usepackage{caption}
\usepackage{tabularx}
\usepackage{fontawesome}
\usepackage{pifont}
\usepackage{enumitem}
\usepackage[font={small}]{caption}
\usepackage{algpseudocode}
\usepackage[font={bf,sf,footnotesize}]{subfig} 
\usepackage{amsthm}
\usepackage{placeins}
\usepackage{bm}
\usepackage[framemethod=tikz]{mdframed}

\usepackage{algorithm}
\usepackage{algpseudocode}
\usepackage{algorithmicx}
\usepackage{verbatim}
\usepackage{xr}
\usepackage{float}
\usepackage{listings}

\newfloat{lstfloat}{htbp}{lop}
\floatname{lstfloat}{Listing}

\newcommand\fw{0.31}
\newcommand\fww{0.9}

\newcommand\blfootnote[1]{%
	\begingroup
	\renewcommand\thefootnote{}\footnote{#1}%
	\addtocounter{footnote}{-1}%
	\endgroup
}

\DeclareMathOperator*{\argmax}{arg\,max}

\definecolor{darkgrey}{RGB}{70,70,70}
\definecolor{lightgrey}{RGB}{200,200,200}

\usepackage[font={bf,sf,footnotesize}]{caption}

\usepackage{cleveref}
\usepackage[utf8]{inputenc}
\crefname{section}{§}{§§}
\Crefname{section}{§}{§§}

\newtheorem{thm}{Theorem}

\newtheorem{lma}{Lemma}

\newcommand{\macb}[1]{\textbf{\textsf{#1}}}
\newcommand*{\affmark}[1][*]{\textsuperscript{#1}}
\DeclareSymbolFont{matha}{OML}{txmi}{m}{it}
\DeclareMathSymbol{\varS}{\mathord}{matha}{83}

\usepackage{xcolor}
\colorlet{hlcolor}{yellow!20}

\mdfdefinestyle{graybox}{
	splittopskip=0,
	splitbottomskip=0,
	frametitleaboveskip=0,
	frametitlebelowskip=0,
	skipabove=0,
	skipbelow=0,
	leftmargin=0,
	rightmargin=0,
	innerrightmargin=0mm,
	innerleftmargin=0mm,
	innertopmargin=1mm,
	innerbottommargin=0mm,
	roundcorner=0mm,
	backgroundcolor=hlcolor}

\acmConference[Technical Report]{Technical Report}{2019}{}
\startPage{1}

\setcopyright{none}

\usepackage{booktabs} 
\usepackage{makecell}

\hypersetup{draft}

\usepackage{pifont}

\begin{document}
	
	\title[I/O Optimal Parallel Matrix 
	Multiplication]{
		Red-Blue Pebbling Revisited:\\ Near 
	Optimal 
	Parallel 
		Matrix-Matrix Multiplication}
	\subtitle{Technical Report
	}
	
\author{
       Grzegorz Kwasniewski\affmark[1], Marko Kabić\affmark[2,3], Maciej 
       Besta\affmark[1],\\Joost VandeVondele\affmark[2,3], Raffaele 
       Solcà\affmark[2,3], Torsten Hoefler\affmark[1]\\
       {\normalsize\affmark[1]Department of Computer Science, ETH Zurich}, 
       {\normalsize\affmark[2]ETH Zurich},
       {\normalsize\affmark[3]Swiss National Supercomputing Centre (CSCS)}
}

\renewcommand{\shortauthors}{G. Kwasniewski et al.}
	
	\begin{abstract}
		We propose COSMA: a parallel matrix-matrix multiplication
		algorithm that is near communication-optimal for all 
		combinations 
		of matrix dimensions, processor counts, and memory sizes. The key 
		idea behind COSMA is to derive an optimal 
		(up to a factor of 0.03\% for 10MB of fast memory) 
		sequential schedule and then 
		parallelize it, preserving I/O optimality. To achieve this, we use the 
		red-blue pebble game to precisely model MMM dependencies and derive a 
		constructive and tight 
		sequential and parallel 
		I/O lower bound proofs. Compared to 2D or 3D algorithms, which fix 
		processor decomposition upfront and then map it to the matrix 
		dimensions, it reduces communication volume by up to $\sqrt{3}$ times. COSMA 
		outperforms the established ScaLAPACK, CARMA, and CTF algorithms in all 
		scenarios 
		up to 12.8x (2.2x on average), achieving up to 88\% of Piz Daint's peak 
		performance. Our work does not require any hand tuning and is 
		maintained as an open source implementation.
		
	\end{abstract}
	
	\maketitle
	
	\section{Introduction}
	\label{sec:intro}
	Matrix-matrix\blfootnote{This is an extended version of the SC'19 
	publication 
		(DOI 10.1145/3295500.3356181)}		
	\blfootnote{Changes in the original submission are listed in the Appendix}
	 multiplication (MMM) is one of the most fundamental building
	blocks in scientific computing, used in linear algebra algorithms~\cite{meyer2000matrix, chatelin2012eigenvalues, linearAlgebraLAPACK},
 (Cholesky 
	and
	LU decomposition~\cite{meyer2000matrix}, 
  eigenvalue
	factorization~\cite{chatelin2012eigenvalues}, triangular
	solvers~\cite{linearAlgebraLAPACK}), 
machine
	learning~\cite{dlsurvey}, graph
	processing~\cite{cormen2009introduction, azad2015parallel,
		kepner2016mathematical, ng2002spectral, slimsell, maciejBC}, 
		computational 
	chemistry~\cite{joost}, and others. Thus, accelerating MMM routines is 
	of great significance for many domains. 
  In this work, we focus on minimizing the amount of transferred data in MMM,
  both across the memory 
	hierarchy (\emph{vertical I/O}) and between processors
	(\emph{horizontal I/O}, aka ``communication'')\footnote{We also focus only on ``classical'' MMM algorithms which 
	perform 
	$n^3$ multiplications and additions. We do not analyze Strassen-like 
	routines~\cite{Strassen}, as in practice they are 
	usually slower~\cite{strassenVsClassic}.}.
	
	\begin{table*}
			\setlength{\tabcolsep}{4pt}
			\renewcommand{\arraystretch}{0.6}
		\centering
		\fontsize{0.26cm}{0.4cm}\selectfont
		\sf
		\begin{tabular}{lllll}
			\toprule
			& \textbf{2D~\cite{summa}} & \textbf{2.5D~\cite{25d}} & 
			\textbf{recursive~\cite{CARMA}} & \textbf{COSMA (this work)} \\
			\midrule
			Input:
			&
			User--specified grid
			&
			\makecell[l] {
				Available
				memory
			}
			&
			\makecell[l] {
				Available memory, 
				matrix dimensions}
			& 
			\makecell[l] {
				Available memory, 
				matrix dimensions}
			\\
			\textbf{Step 1}
			&
			Split $m$ and $n$
			&
			Split $m$, $n$, $k$
			& 
			\makecell[l] {
				Split recursively the largest dimension
			}
			& 
			\makecell[l] {
				Find the optimal
				sequential schedule
			}
			\\
			\textbf{Step 2}
			&
			\makecell[l] {
				Map matrices 
				to processor grid
			}
			&
			\makecell[l] {
				Map matrices 
				to processor grid
			}
			&
			\makecell[l] {
				Map matrices 
				to recursion tree
			}
			& 
			\makecell[l] {
				Map sequential 
				domain to matrices
			}
			\\
			\midrule
			&
			\makecell[l] {
				\faThumbsDown ~Requires manual tuning  \\
				\faThumbsDown ~Asymptotically
				more comm.  
			}
			&
			\makecell[l] {
				\faThumbsOUp ~Optimal for $m=n$  \\
				\faThumbsDown ~Inefficient for 
				$m \ll n$ 
				or $n \ll m$  \\
				\faThumbsDown ~Inefficient for some $p$
			}
			&
			\makecell[l] {
				\faThumbsOUp ~Asymptotically 
				optimal 
				for all  $m,n,k,p$  \\
				\faThumbsDown ~Up to $\sqrt{3}$ times higher comm. cost\\
				\faThumbsDown ~$p$ must be 
				a power of 2  
			}
			& 
			\makecell[l] {
				\faThumbsOUp ~ ~Optimal for 
				all $m,n,k$  \\
				\faThumbsOUp ~ ~Optimal for
				all $p$  \\
				\faThumbsOUp \faThumbsOUp ~Best time-to-solution
			}
			\\
			\bottomrule
		\end{tabular}
		\caption{
			\textmd{
				Intuitive comparison between the COSMA algorithm and the 
			state-of-the-art 2D, 2.5D, and recursive decompositions. $C=AB, A 
			\in 
			\mathbb{R}^{m \times k}$, $ B \in 
			\mathbb{R}^{k \times n}$
		}}
		\label{tab:intro}
	\end{table*}
	
The path to I/O optimality of MMM algorithms is at least 50 years old.
The first
parallel MMM algorithm is by Cannon~\cite{Cannon}, which
works for square matrices and square processor 
decompositions. 
 Subsequent works~\cite{mmm1,mmm2}
generalized the MMM algorithm to rectangular matrices, different processor
decompositions, and communication patterns. PUMMA~\cite{pumma} package
generalized previous works to transposed matrices and different data layouts. 
  SUMMA algorithm~\cite{summa} further extended
it by 
optimizing the communication, introducing pipelining and
communication--computation overlap. This is now a 
state-of-the-art so-called 2D
algorithm (it decomposes processors in a 2D grid) used e.g., in
ScaLAPACK library~\cite{scalapack}.

Agarwal et al.~\cite{summa3d} showed that in a presence of extra memory, one can
do better and introduces a 3D processor decomposition.
The 2.5D algorithm by Solomonik
and Demmel~\cite{25d} effectively interpolates
between those two results, depending on the available
memory. However, Demmel et al. showed
that algorithms optimized for square matrices often perform poorly
when matrix dimensions vary significantly~\cite{CARMA}.
Such matrices are common in many
relevant areas, for example in machine learning
~\cite{rectangularML, 
kmeansTallSkinny} or computational 
chemistry~\cite{rectangularChemistry, 
quantumTallSkinny}.
They introduced CARMA~\cite{CARMA}, 
a recursive algorithm that achieves asymptotic lower bounds for all
configurations of dimensions and memory sizes. This evolution for chosen steps 
is depicted 
symbolically in Figure~\ref{fig:timeline}.

		\begin{figure}[t]
		\includegraphics[width=0.9\columnwidth]{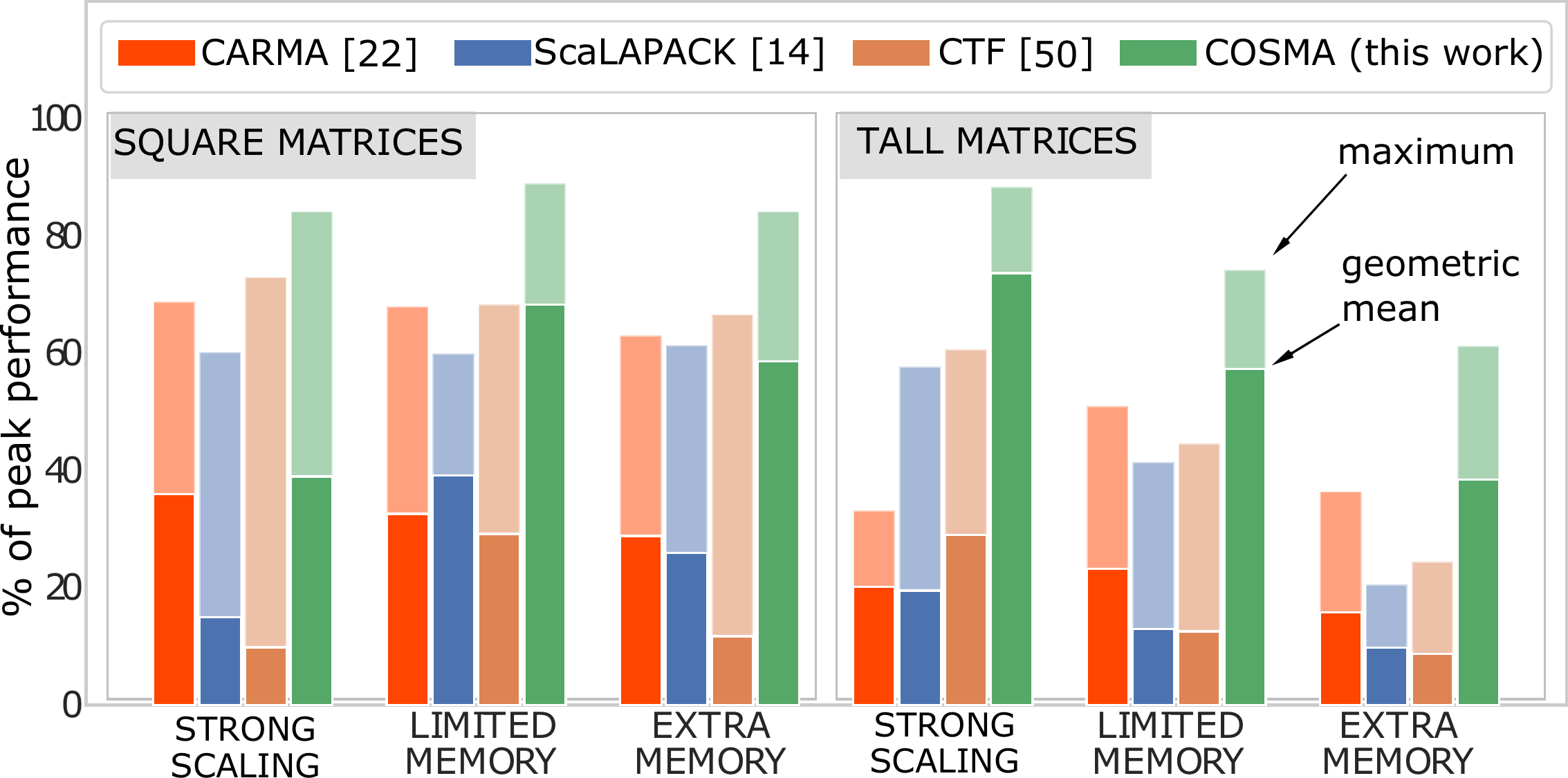}
		\caption{
			\textmd{Percentage of peak flop/s across the 
			experiments 
		ranging 
				from 
				109 to 18,432 cores achieved by COSMA and the state-of-the-art 
				libraries (Sec. \ref{sec:results}).}}
		\label{fig:introPlot}
	\end{figure}
	
\begin{figure}
		\includegraphics[width=0.92\columnwidth]{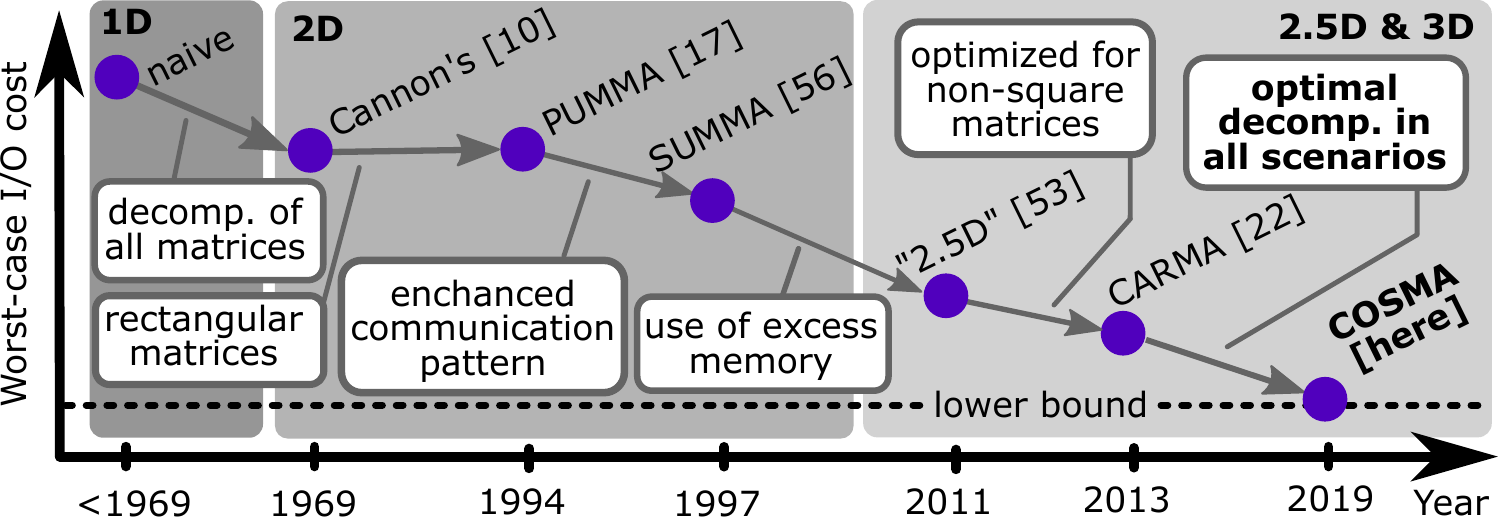}
\vspace{-1.0em}
\caption{\textmd{Illustratory evolution of MMM algorithms reaching the I/O lower 
	bound.}}
\label{fig:timeline}
\vspace{-1.0em}
\end{figure}

	Unfortunately, we observe several limitations with state-of-the art
	algorithms.
	ScaLAPACK~\cite{scalapack} (an implementation of
	SUMMA) supports
	only the 2D~decomposition, which is communication--inefficient in the 
	presence of 
	extra memory. Also, it requires a 
	user to
	fine-tune parameters such as block sizes or a processor
	grid size. CARMA supports only scenarios when the 
	number of
	processors is a power of two~\cite{CARMA}, a serious limitation, as the 
	number 
	of processors
	is usually determined by the available
	hardware resources.
	Cyclops Tensor Framework (CTF)~\cite{cyclops} (an implementation of the 
	2.5D decomposition) can utilize any number of processors, but its 
	decompositions 
	may be far from 
	optimal~(\cref{sec:results}). We also emphasize that 
	\emph{asymptotic 
		complexity is an insufficient measure of practical performance}. 
	We later (\cref{sec:parStrategies}) identify that CARMA performs up to 
	$\sqrt{3}$ more communication.
	Our 
	observations are summarized in Table~\ref{tab:intro}. Their practical 
	implications are shown in Figure~\ref{fig:introPlot}, where we see that all 
	existing algorithms perform poorly for some configurations. 
	
	In this work, we present COSMA (Communication Optimal S-partition-based 
	Matrix 
	multiplication Algorithm): 
	an algorithm that takes a new approach to multiplying
	matrices and alleviates the issues above. COSMA is I/O optimal
	for \emph{all combinations of parameters}  (up to the 
	factor
	 of $\sqrt{S}/(\sqrt{S+1}-1)$,
	  where $S$ is the size of the 
	 fast
	 memory\footnote{Throughout this paper we use the original notation from 
	 Hong
	 	and Kung to denote the memory size $S$. In literature, it is also 
	 	common to use
	 	the symbol $M$~\cite{externalMem,IronyMMM, parallelExMem}.}). The 
		driving idea is to develop a 
	general method of deriving I/O optimal schedules by explicitly modeling 
	data reuse in the red-blue pebble game. We then parallelize the sequential 
	schedule, minimizing the I/O between processors, and derive an optimal 
	domain 
	decomposition.
	This is in contrast with the other discussed algorithms, which fix the 
	processor grid upfront and then map it to a sequential schedule for each 
	processor.   We outline the
	algorithm in~\cref{sec:commDescr}.
	To prove its optimality,  we first provide a new 
	constructive proof of a sequential I/O lower 
	bound~(\cref{sec:seqScheduling}),  
	then we
	derive the communication cost of  
	parallelizing the sequential
	schedule~(\cref{sec:parStrategies}), and finally we construct an I/O 
	optimal 
	parallel 
	schedule~(\cref{sec:parScheduling}). 
	The detailed communication analysis of COSMA, 2D, 2.5D, and 
	recursive decompositions is
	presented in Table~\ref{tab:summary}. Our algorithm reduces the data 
	movement volume by a factor of up to $\sqrt{3} \approx 1.73$x compared to 
	the 
	asymptotically optimal recursive decomposition and up to 
	$\max\{m,n,k\}$ times compared to the 2D
	algorithms, like Cannon's~\cite{generalCannon} or SUMMA~\cite{summa}.
	
	Our implementation
	enables transparent integration with the ScaLAPACK data
	format~\cite{scalapackLayout} and delivers near-optimal computation 
	throughput.
	We later (\cref{sec:implementation}) show that the schedule 
	naturally expresses
	communication--computation overlap, 
	enabling even higher speedups
	using Remote Direct Memory Access (RDMA).
	Finally, our I/O-optimal approach is
	 generalizable to other linear algebra kernels. 
	We provide the following contributions:
	
	\begin{itemize}[leftmargin=1em]
		\item We propose COSMA: a distributed MMM algorithm that is 
		nearly-optimal (up to the factor of $\sqrt{S}/(\sqrt{S+1}-1)$)
		for 
		\emph{any combination of input parameters} (\cref{sec:commDescr}). 
		\item Based on the red-blue pebble game 
		abstraction~\cite{redblue}, we provide a new method of deriving I/O 
		lower 
		bounds 
		(Lemma~\ref{lma:comp_intesity}), which may be used to generate optimal 
		schedules (\cref{sec:introIO}).
		\item Using Lemma~\ref{lma:comp_intesity}, we provide a new 
		constructive proof 
		of
		the sequential MMM I/O lower bound. The proof delivers 
		constant factors tight up to $\sqrt{S}/(\sqrt{S+} - 
		1)$(\cref{sec:seqOpt}).
    \item We extend the sequential proof to parallel 
		machines and provide I/O optimal parallel MMM schedule 
		(\cref{sec:parScheduling}).
		\item We reduce memory footprint for communication buffers and guarantee
		minimal local data reshuffling by using a blocked 
		data layout
		(\cref{sec:datalayout}) and a static buffer pre-allocation
		(\cref{sec:bufferReuse}), providing compatibility with the
		ScaLAPACK format.
		\item We evaluate the performance of COSMA, ScaLAPACK, CARMA, and CTF 
		on the CSCS Piz Daint 
		supercomputer 
		 for an
		extensive selection of problem dimensions, memory sizes, and numbers of
		processors,
		 showing significant I/O reduction and the speedup of up to 8.3 times 
		 over the 
		second-fastest 
		algorithm
		(\cref{sec:results}).
	\end{itemize}
	
	\section{Background}
	
		We first describe our machine 
	model (\cref{sec:machineModel}) and computation model (\cref{sec:compModel}).
  We then define 
	our 
	optimization 
	goal: ~\emph{the I/O cost}~(\cref{sec:optGoals}). 

	\subsection{Machine Model}
	\label{sec:machineModel}
	
	We model a parallel machine with $p$
	processors, each with local memory of size $S$ words.
	A processor can send and receive from any other processor up to $S$ words 
	at 
	a time.
	To perform any computation, all operands must reside in  
	processor' local 
	memory.
	If shared memory is present, then it is assumed that it has infinite 
	capacity. A cost of transferring a word from the shared to the local memory 
	is equal to the cost of transfer between two local memories.
	
	\subsection{Computation Model}
	\label{sec:compModel}
	
	We now briefly specify a model of \emph{general} 
	computation; we use this
	model to derive the theoretical I/O cost in both the sequential and 
	parallel 
	setting.  
	An execution of an algorithm is modeled with the \emph{computational 
	directed acyclic 
	graph} 
	(CDAG)
	$G=(V,E)$~\cite{completeRegisterProblems,pebblegameregister,
		registerpebblecolor}. A vertex $v \in V$ represents one
	elementary operation in the given computation.
	  An edge $(u,v) \in E$
	indicates that an operation $v$ depends on the result of $u$. A set of all
	immediate predecessors (or successors) of a vertex are its \emph{parents} (or
	\emph{children}).  Two selected subsets $I, O \subset V$ are \emph{inputs} 
	and
	\emph{outputs}, that is, sets of vertices that have no parents (or no 
	children,
	respectively).
	
  \noindent
	\macb{Red-Blue Pebble Game} Hong and Kung's red-blue pebble 
	game \cite{redblue} 
	models
	an execution of an algorithm in a two-level memory structure with a
	small-and-fast as well as large-and-slow memory. A red (or a blue) pebble
	placed on a vertex of a CDAG denotes that the result of the corresponding 
	elementary computation is inside
	the fast (or slow) memory. 
	In the initial (or terminal) configuration, only inputs (or outputs) of the 
	CDAG have
	blue pebbles.
	There can be at most $S$ red pebbles used at any given time. A \emph{complete 
	CDAG calculation} is a
	sequence of moves that lead from the initial to the terminal
	configuration.
	One is allowed to:  place a red pebble on any
	vertex with a blue pebble (load),  place a blue pebble on any
	vertex with a red pebble (store),  place a red pebble on a vertex
	whose parents all have red pebbles (compute), remove any pebble, 
	red or blue, from any vertex (free memory).
		An \emph{I/O optimal} complete CDAG calculation corresponds to a sequence 
		of 
	moves (called 
	\emph{pebbling} of a graph) which minimizes loads and stores.
	In the MMM context, it is an order in which all $n^3$
	multiplications are performed. 
	
	\subsection{Optimization Goals}
	\label{sec:optGoals}
	Throughout this paper we focus on the \emph{input/output (I/O) cost}
	of an algorithm.
	The 
	I/O cost~$Q$ is the total number of words transferred 
	during the execution 
	of a schedule. On a sequential or shared memory machine equipped 
	with small-and-fast and slow-and-big memories, these transfers are load 
	and store operations from and to the slow memory (also called the 
	\emph{vertical 
		I/O}). For a distributed machine with a limited memory per node, 
	the 
	transfers are communication operations between the nodes (also called the
	\emph{horizontal I/O}). A schedule is \emph{I/O optimal} if it minimizes 
	the 
	I/O cost among all schedules of a given CDAG. 
	We also model a \emph{latency cost} $L$, which is a 
	maximum number of
	messages sent by any processor.  
	
	\subsection{State-of-the-Art MMM Algorithms}
	\label{sec:state-of-the-artAlgs}
	
	Here we briefly describe strategies of the existing MMM algorithms.
	Throughout the whole paper, we consider matrix multiplication $C = 
	AB$, where $A \in \mathbb{R}^{m \times k}, B \in \mathbb{R}^{k \times n},  
	C 
	\in \mathbb{R}^{m \times 
		n}$, where $m$, $n$, and $k$ are matrix dimensions. Furthermore, we 
		assume 
	that the size of each matrix element is one word, and 
	that $S < 
	\min\{mn, mk, nk\}$, that is, none of the matrices 
	fits into single processor's fast memory. 
	
	We compare our algorithm with the 2D, 2.5D, and recursive 
	decompositions (we select parameters for 2.5D to also cover 3D). We assume 
	a square processor grid $[\sqrt{p}, \sqrt{p}, 1]$ 
	for the 2D variant, analogously 
	to Cannon's algorithm~\cite{Cannon}, and a cubic grid $[\sqrt{p/c}, 
	\sqrt{p/c}, c]$ for the 2.5D 
	variant~\cite{25d}, where $c$ is the amount of the ``extra'' memory $c = 
	pS/(mk 
	+ nk)$. For the recursive decomposition, we assume that in each recursion 
	level 
	we split the largest dimension $m,n,$ or $k$ in half, until the 
	domain 
	per 
	processor fits into memory.
	The detailed complexity analysis of these decompositions is 
	in 
	Table~\ref{tab:summary}.
	We note that ScaLAPACK or CTF can handle non-square 
	decompositions, however they create different problems, as discussed 
	in~\cref{sec:intro}. 
	Moreover, 
	in~\cref{sec:results} we compare their 
	performance with COSMA 
	and measure significant speedup in \emph{all} scenarios. 
	
\enlargethispage{\baselineskip}

	\section{COSMA: High-Level Description}
	\label{sec:commDescr}
	
	\begin{figure}[!tbp]
		\centering
		\subfloat[3D domain
		decomposition]{\includegraphics[width=0.23\textwidth]
			{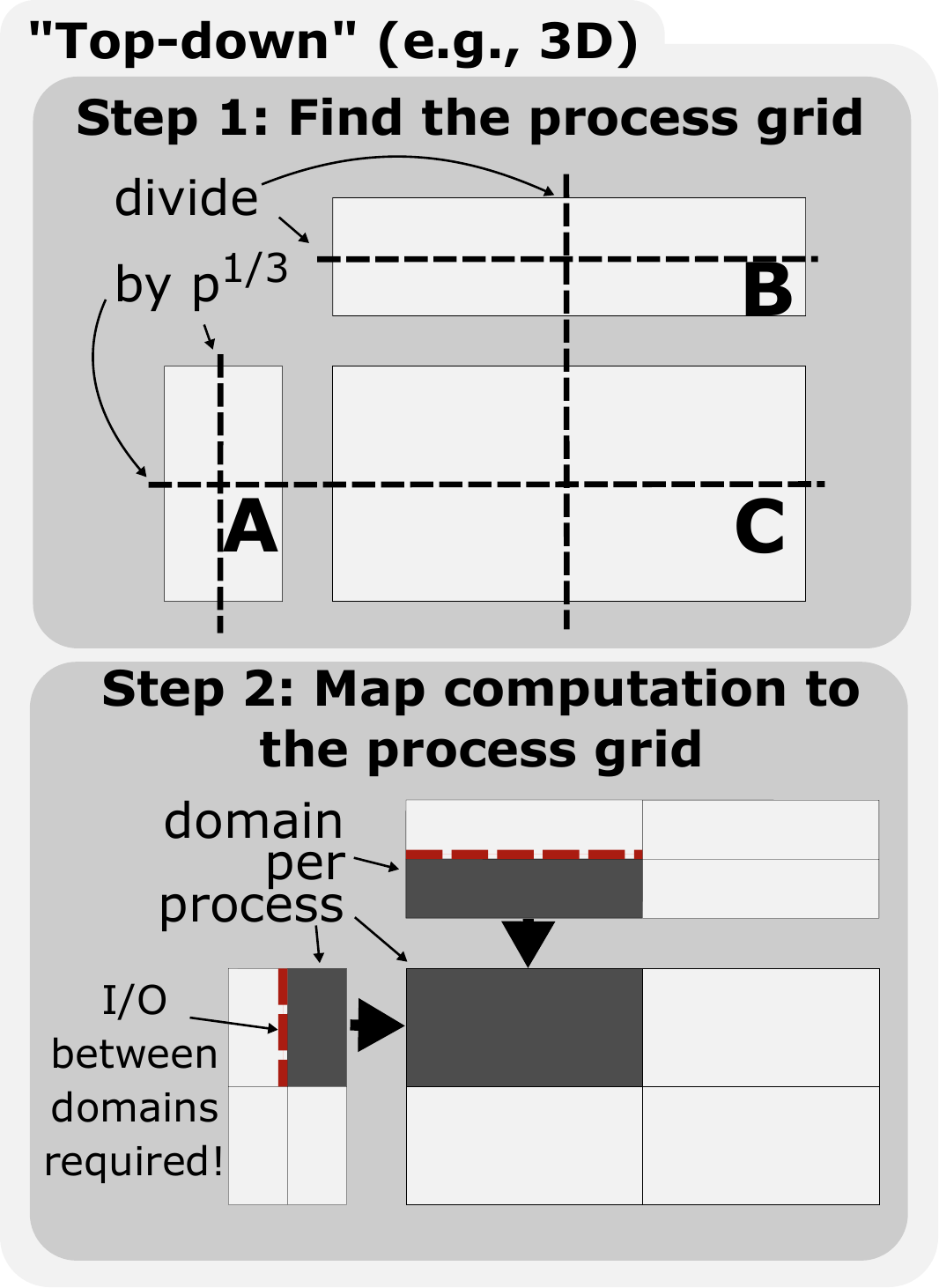}\label{fig:f1}}
		\hfill
		\subfloat[COSMA decomposition]{\includegraphics[width=0.23\textwidth]
			{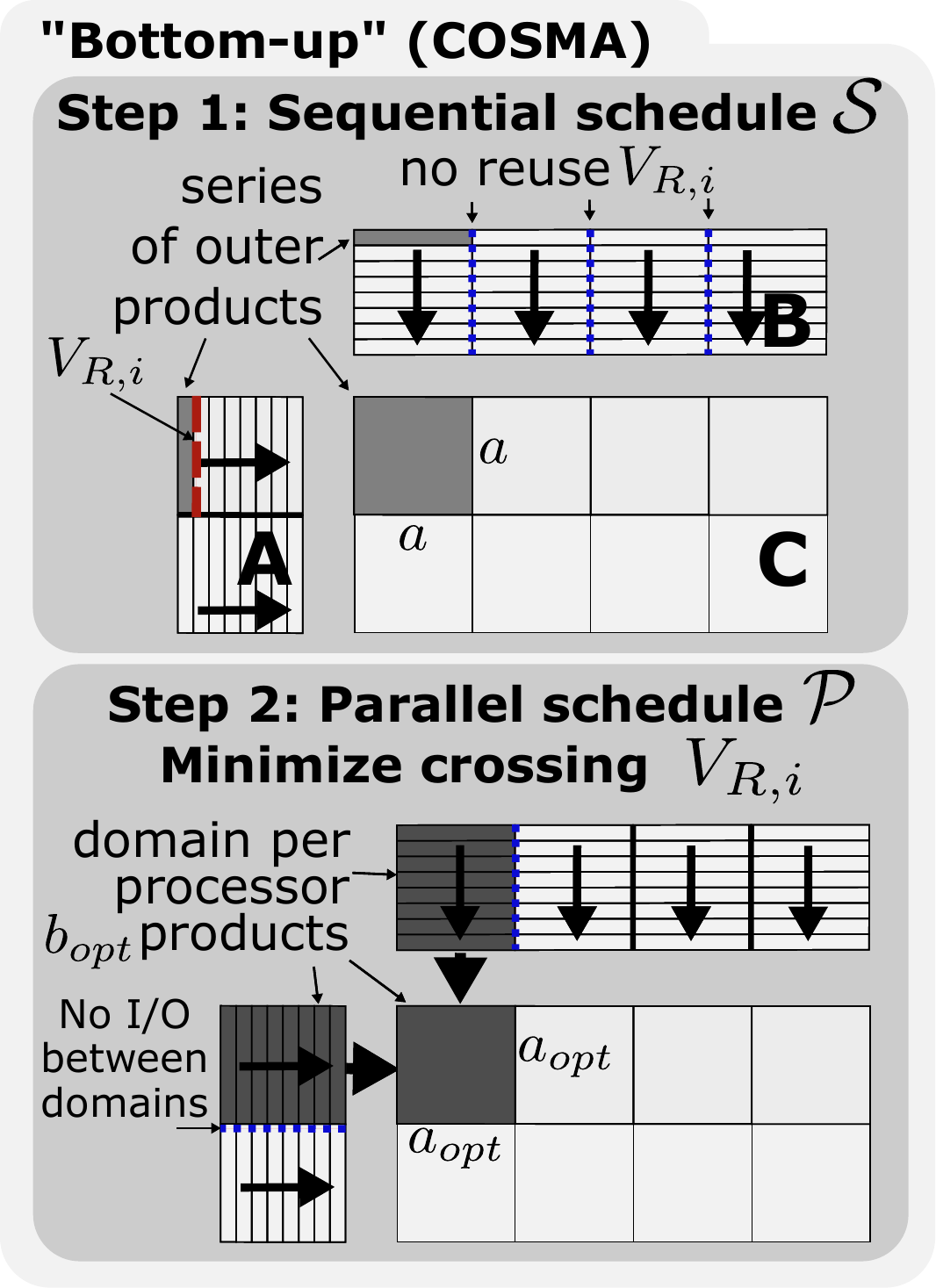}\label{fig:f2}}
		\caption{
			\textmd{Domain decomposition
			using $p=8$  processors. In scenario (a), a 
			straightforward 3D 
			decomposition
			divides every dimension in $p^{1/3}=2$. In 
			scenario (b), COSMA 
			starts by
			finding a near optimal sequential schedule and then parallelizes it 
			minimizing
			crossing data reuse $V_{R,i}$~(\cref{sec:seqOpt}). The total 
			communication volume is reduced by 17\%
			compared to the former strategy.}
		}
		\label{fig:topdown-vs-bottomup}
	\end{figure}
	
	COSMA  decomposes processors by parallelizing the near optimal sequential 
	schedule
	under constraints: (1) equal work distribution and 
	(2) equal memory size per 
	processor.
	Such a local sequential schedule is independent of matrix dimensions.  Thus,
	intuitively, instead of dividing  a global domain among $p$ processors 
	(the
	\emph{top-down} approach), we start from deriving a near I/O optimal
	\emph{sequential} schedule. We then parallelize it, minimizing the I/O and
	latency costs $Q$, $L$ (the \emph{bottom-up} approach);
	Figure~\ref{fig:topdown-vs-bottomup} presents more details.
	
	\hspace{\parindent}COSMA is sketched in Algorithm~\ref{alg:COSMA}.
	In Line~\ref{alg:line:seqSched} we derive a sequential schedule, which 
	consists of series of $a \times a$ outer products. 
	(Figure~\ref{fig:mmmParallelization} b). In 
	Line~\ref{alg:line:findOptDomain}, each processor is assigned to compute 
	$b$ of these products, forming a 
	\emph{local domain} $\mathcal{D}$
	(Figure~\ref{fig:mmmParallelization} c), that is 
	each 
	$\mathcal{D}$ 
	contains $a 
	\times a 
	\times b$ vertices (multiplications to be performed - the derivation 
	of $a$ and $b$ is presented in~\cref{sec:parScheduling}).
	In 
	Line~\ref{alg:line:fitranks}, we 
	find a processor
	grid~$\mathcal{G}$ that evenly distributes this domain by the matrix 
	dimensions 
	$m,n$,
	and $k$. If the dimensions are not divisible by $a$ or $b$, this
	function also evaluates new values of $a_{opt}$ and $b_{opt}$ by fitting 
	the best 
	matching
	decomposition, possibly not utilizing some 
	processors~(\cref{sec:decompArbitrary}, 
	Figure~\ref{fig:mmmParallelization} 
	d-f). The maximal number of idle processors is a tunable parameter 
	$\delta$.
	In Line~\ref{alg:line:datadecomp}, 
	we determine the initial decomposition of matrices $A, B$, and $C$ to the 
	submatrices $A_l, B_l, C_l$ 
	that 
	are 
	local for 
	each processor.
	COSMA may handle any initial data layout, however, 
	an optimal 
	block-recursive
	one~(\cref{sec:datalayout}) may be achieved in a preprocessing phase.
	In Line~\ref{alg:line:stepsize}, we compute the size of the 
	communication step, that is, how many of $b_{opt}$ outer products assigned 
	to each processor are computed in a single round, minimizing the 
	latency~(\cref{sec:parScheduling}).
	In Line~\ref{alg:line:steps} we
	compute the number of
	sequential steps 
	(Lines~\ref{alg:line:innerloopStart}--\ref{alg:line:innerLoopEnd}) in which 
	every processor: (1) distributes and updates its
	local data $A_l$ and $B_l$ among the grid $\mathcal{G}$ 
	(Line~\ref{alg:line:distrData}),
	and (2) multiplies $A_l$ and $B_l$
	(Line~\ref{alg:line:compute}). Finally, the partial results $C_l$ are
	reduced over $\mathcal{G}$ (Line~\ref{alg:line:reduce}).

	\macb{I/O Complexity of COSMA}
	Lines~\ref{alg:line:findOptDomain}--\ref{alg:line:steps} require no 
	communication 
	(assuming that the 
	parameters $m, n, k, p, S$ are already distributed).
	The loop in Lines~\ref{alg:line:innerloopStart}-\ref{alg:line:innerLoopEnd} 
	executes $\left\lceil{2ab/(S-a^2)}\right \rceil$ times. In 
	Line~\ref{alg:line:distrData}, each processor receives $|A_l| + |B_l|$ 
	elements.
	Sending the partial results in 
	Line~\ref{alg:line:reduce} adds $a^2$ communicated elements. 
	In~\cref{sec:parScheduling} we derive the optimal values for $a$ and  
	$b$, which yield a total of 
	$\min \Big\{S + 2 \cdot \frac{mnk}{p\sqrt{S}}, 3 
	\left(\frac{mnk}{P}\right)^{2/3} \Big\}$ elements communicated.
	\begin{algorithm}
		\footnotesize
		\caption{COSMA} \label{alg:COSMA}
		\begin{algorithmic}[1]
			\Require $\text{matrices } A \in \mathbb{R}^{m \times k}, B \in 
			\mathbb{R}^{k \times n}$,
			\Statex {number of processors:  $p$, memory size: $S$, 
			computation-I/O tradeoff ratio $\rho$}
			\Ensure $\text{matrix } C = AB \in \mathbb{R}^{m \times n}$
			\State $a \gets FindSeqSchedule(S,m,n,k,p)$
			\Comment{sequential I/O optimality~(\cref{sec:seqOpt})}
			\label{alg:line:seqSched}
			\State $b \gets ParallelizeSched(a,m,n,k,p)$ 
			\Comment{parallel I/O optimality~(\cref{sec:parOptimality})}
			\label{alg:line:findOptDomain}
			\State $(\mathcal{G}, a_{opt}, b_{opt}) \gets 
			FitRanks(m,n,k,a,b,p, \delta)$ 
			\label{alg:line:fitranks}
			\ForAll{$p_i \in \left\{1 \dots p \right\}$} \textbf{in parallel}
			\label{alg:line:outerloopStart}
			\State $(A_l, B_l, C_l) \gets GetDataDecomp(A,B, \mathcal{G}, p_i)$ 
			\label{alg:line:datadecomp}
			\State $s \gets \left \lfloor{\frac{S - a_{opt}^2}{2a_{opt}}}\right 
			\rfloor$
			\Comment{latency-minimizing size of a 
			step~(\ref{sec:parScheduling})}
			\label{alg:line:stepsize}
			\State $t \gets \left \lceil{\frac{b_{opt}}{s}}\right \rceil$ 
			\Comment{number of steps}
			\label{alg:line:steps}
			\For{$j \in \{1 \dots t\}$} 
			\label{alg:line:innerloopStart}
			\State $(A_l, B_l) \gets DistrData(A_l,B_l,\mathcal{G}, j, p_i)$ 
			\label{alg:line:distrData}
			\State $C_l \gets Multiply(A_l, B_l,j)$ 
			\Comment{compute locally}
			\label{alg:line:compute}
			\EndFor
			\label{alg:line:innerLoopEnd}
			\State $C \gets Reduce(C_l,\mathcal{G})$ 
			\Comment{reduce the partial 
			results}
			\label{alg:line:reduce}
			\EndFor
			\label{alg:line:outerLoopEnd}
		\end{algorithmic}
	\end{algorithm}

	\section{Arbitrary CDAGs: Lower Bounds}
	\label{sec:introIO}

	We now present a mathematical machinery for 
	deriving I/O lower bounds for general CDAGs. We extend 
	the main lemma by 
	Hong and 
	Kung~\cite{redblue}, which provides a method to find an I/O lower bound for 
	a 
	given CDAG. That lemma, however, does not give a tight bound, 
	as it 
	overestimates a \emph{reuse set} size (cf. 
	Lemma~\ref{lma:reuse}). Our key result here, 
	Lemma~\ref{lma:comp_intesity},
	allows us to derive a constructive proof of a tighter I/O lower bound for a 
	sequential execution of the MMM CDAG~(\cref{sec:seqOpt}). 

	The driving idea of both Hong and Kung's and our approach is to show that 
	some properties of an 
	optimal pebbling of a 
	CDAG (a problem which is PSPACE-complete~\cite{redblueHard_}) can be 
	translated to the properties of a specific partition of the CDAG (a collection of subsets $V_i$ of the CDAG; these subsets form
  subcomputations, see~\cref{sec:compModel}). 
	One can use the properties of this partition to bound the 
	number of I/O operations of the corresponding pebbling.
Hong and Kung use a specific variant of this partition, denoted as 
$S$-partition~\cite{redblue}.

	We first introduce our generalization of  
	$S$-partition,
	called $X$-partition, that is the base of our analysis.
	We describe 
	symbols used in our analysis in Table~\ref{tab:symbols}. 
	\begin{table}[h]
		\vspace{-0.5em}
		\setlength{\tabcolsep}{2pt}
		\renewcommand{\arraystretch}{0.7}
		\centering
		\scriptsize
		\sf
		\begin{tabular}{@{}l|ll@{}}
			\toprule
			\multirow{5}{*}{\begin{turn}{90}\textbf{MMM}\end{turn}}
			& $m, n, k$& Matrix dimensions \\
			& $A, B$& Input matrices $A \in \mathbb{R}^{m \times k}$ and $ B 
			\in 
			\mathbb{R}^{k \times n}$ \\
			& $C = AB$& Output matrix $C \in \mathbb{R}^{m \times n}$ \\
			& $p$& The number of processors \\
			\midrule
			\multirow{5}{*}{\begin{turn}{90}\textbf{graphs}\end{turn}}
			& $G$&A directed acyclic graph $G=(V,E)$\\
			& $Pred(v)$& A set of immediate predecessors of a vertex $v$:\\
			& & $Pred(v) = 
			\{u : (u,v) 
			\in E\}$ \\
			& $Succ(v)$& A set of immediate successors of a vertex $v$: \\
			& &  $Succ(v) = \{u : 
			(v,u) 
			\in E\}$ \\
			\midrule
			\multirow{13}{*}{\begin{turn}{90}\textbf{I/O complexity}\end{turn}}
			& $S$ & The number of red pebbles (size of the fast memory)\\
			& $V_i$ & An $i$-th subcomputation of an $S$-partition \\
			& $Dom(V_i), Min(V_i)$ & Dominator and minimum sets of 
			subcomputation 
			$V_i$\\
			& $V_{R,i}$ & \makecell[l]{The \emph{reuse set}: a set of vertices 
				containing red pebbles\\(just before $V_i$ starts) and used by 
				$V_i$} \\
			& $H(S)$ & The smallest cardinality of a valid $S$-partition \\
			& $R(S)$ & The maximum size of the reuse set \\
			& $Q$ & The I/O cost of a schedule (a number of I/O operations) \\
			& $\rho_i$ & The computational intensity of $V_i$\\
			& $\rho = \max_i\{\rho_i\}$ & The maximum computational intensity\\
			\midrule
			\multirow{4}{*}{\begin{turn}{90}
					\textbf{Schedules}
			\end{turn}} 
			& $\mathcal{S} = \{V_1, \dots, V_h\}$ & The sequential schedule (an 
			ordered set 
			of 
			$V_i$) \\ 
			& $\mathcal{P} = \{\mathcal{S}_1, \dots, \mathcal{S}_p\}$ & The 
			parallel schedule (a set 
			of 
			sequential schedules $\mathcal{S}_j$) \\
			& $\mathcal{D}_j = \bigcup_{V_i \in \mathcal{S}_j}V_i $ & The local 
			domain 
			(a set of vertices 
			in $\mathcal{S}_j$ \\
			& $a,b$ &  Sizes of a local domain: $|\mathcal{D}_j| = a^2b$\\
			
			\bottomrule
		\end{tabular}
		\caption{
			\textmd{The most important symbols used in the paper.}
		}
		\label{tab:symbols}
	\end{table}
	
	\vspace{-1.5em}
	\macb{$X$-Partitions}
  Before we define $X$-partitions, we first need to define two sets,
  the \emph{dominator set} and the \emph{minimum set}.
	Given a subset $V_i \in V$, define a \emph{dominator set} $Dom(V_i)$ as a 
	set 
	of 
	vertices in $V$, such that every path from any input of a 
	CDAG to any vertex in $V_i$ must contain at least one vertex in $Dom(V_i)$. 
	Define 
	also the \emph{minimum set} $Min(V_i)$ as the set of all vertices in $V_i$ 
	that 
	do 
	not 
	have any children in $V_i$.

	Now, given a CDAG $G = (V,E)$, let $V_1, V_2, \dots V_h \in V$ be a series of 
	subcomputations that (1) are 
	pairwise disjoint ($\forall_{i,j, i \ne j}V_i \cap V_j = \emptyset )$, (2)
	cover the whole CDAG ($\bigcup_i V_i = V$), (3)
  have no cyclic dependencies between them,
	and (4) their dominator and minimum sets are 
	at most of size $X$ ($\forall_i (|Dom(V_i)| \le X \land |Min(V_i)| \le X)$).
	These subcomputations $V_i$ correspond to some execution order (a schedule) 
	of 
	the CDAG, such that at step $i$, only vertices in $V_i$ are pebbled. We 
	call 
	this 
	series an \emph{$X$-partition} or a 
	\emph{schedule} of the 
	CDAG and denote this schedule with $\mathcal{S}(X) = \{V_1, \dots, V_h\}$.
	
	\subsection{Existing General I/O Lower Bound}
	\label{sec:spartProof}
	
	Here we need to briefly bring back the original lemma by Hong 
	and Kung, together with an intuition of its proof, as we use 
	a similar method for our Lemma~\ref{lma:reuse}.
	
	\macb{Intuition}
	The key notion in the existing bound is to use\linebreak 
	$X=2S$-partitions for a 
	given 
	fast
	memory size~$S$.
	For any subcomputation $V_i$, if $|Dom(V_i)| = 2S$, then at 
	most $S$
	of them could contain a red pebble before $V_i$ begins.  Thus, at least $S$
	additional pebbles need to be loaded from the memory.  The similar argument
	goes for $Min(V_i)$. Therefore, knowing the lower bound on the number of 
	sets
	$V_i$ in a valid \emph{$2S$-partition}, together with the observation that 
	each
	$V_i$ performs at least $S$ I/O operations, we phrase the 
	lemma by Hong and Kung: 

	\begin{lma}[~\cite{redblue}]
	
	\emph{The minimal number $Q$ of I/O operations for any 
		valid execution of a CDAG 
		of
		any I/O computation is bounded by}
	
	\begin{equation}
	\label{eq:redbluebound}
	Q \ge S \cdot (H(2S) - 1)
	\end{equation}
	\end{lma}
	
	\begin{proof}
	Assume that we know the optimal \emph{complete calculation} of the CDAG, 
	where a calculation
	is a sequence of allowed moves in the red-blue pebble game~\cite{redblue}. 
	Divide the complete calculation 
	into $h$ consecutive subcomputations $V_1, V_2, ..., V_h$, such that during 
	the
	execution of $V_i$, $i < h$, there are exactly $S$ I/O operations, and in 
	$V_h$
	there are at most $S$ operations. Now, for each $V_i$, we define two 
	subsets of
	$V$, $V_{R,i}$ and $V_{B,i}$.
	$V_{R,i}$ contains vertices that have red pebbles placed on them just before
	$V_i$ begins.
	$V_{B,i}$ contains vertices that have blue pebbles placed on 
	them just before
	$V_i$ begins, and have red pebbles placed on them during $V_i$.
	Using these definitions, we have: \ding{182} $V_{R,i} \cup 
	V_{B,i} =
	Dom(V_i)$, \ding{183} $|V_{R,i}| \le S$, \ding{184} 
	$|V_{B,i}| \le S$, and
	\ding{185} $|V_{R,i} \cup V_{B,i}| \le |V_{R,i}| + |V_{B,i}| 
	\le 2S$.
	We define similar subsets $W_{B,i}$ and $W_{R,i}$ for the 
	minimum set 
	$Min(V_i)$.  $W_{B,i}$ contains all vertices in $V_i$ that 
	have a blue 
	pebble placed on them during $V_i$, and  $W_{R,i}$ contains all vertices in 
	$V_i$ that have a red pebble at the end of $V_i$. By the definition of 
	$V_i$, 
	$|W_{B,i}| \le S$, by the constraint on the red pebbles, we 
	have $|W_{R,i}| \le 
	S$, and by te definition of the minimum set,$Min(V_i) \subset W_{R,i} \cup 
	W_{B,i}$.
	Finally, by the definition of $S$-partition, $V_1, V_2, ..., 
	V_h$ form a valid
	$2S$-partition of the CDAG. 
\end{proof}

	\subsection{Generalized I/O Lower Bounds}
	\label{sec:seq-proof}
	
	\subsubsection{Data Reuse}
	\label{sec:datareuse}
	
	A more careful look at sets $V_{R,i}, V_{B,i}, W_{R,i}$, 
	and $W_{B,i}$
	allows us to refine the bound on the number of I/O operations on a CDAG.  By
	definition, $V_{B,i}$ is a set of vertices on which we place a red pebble 
	using
	the load rule; We call
    $V_{B,i}$ a \emph{load set} of $V_i$. Furthermore, $W_{B,i}$
	contains all the vertices on which we place a blue pebble during the 
	pebbling
	of $V_i$; We call
    $W_{B,i}$ a \emph{store set} of $V_i$.
	However,
	we impose more strict $V_{R,i}$ and $W_{R,i}$ definitions: $V_{R,i}$ 
	contains vertices 
	that
	have red pebbles placed on them just before $V_i$ begins \emph{and -- for 
	each
		such vertex $v \in V_{R,i}$ -- at least one child of $v$ is pebbled 
		during the
		pebbling of $V_i$ using the compute rule of the red-blue pebble game}.
We call
	$V_{R,i}$ a \emph{reuse set} of $V_i$.   
	Similarly, $W_{R,i}$ contains vertices 
	that
	have red pebbles placed on them after $V_i$ ends and were pebbled during 
	$V_i$ \emph{and -- for 
		each
		such vertex $v \in W_{R,i}$ -- at least one child of $v$ is pebbled 
		during the
		pebbling of $V_{i+1}$ using the compute rule of the red-blue pebble 
		game}.
We call $W_{R,i}$ a \emph{cache set} of $V_i$.
Therefore, if $Q_i$ is the number 
	of I/O operations during the
	subcomputation $V_i$, then $Q_i \ge |V_{B,i}| +  |W_{B,i}|$.
	
	We first observe that, given the optimal complete calculation, one can 
	divide
	this calculation into subcomputations such that each subcomputation $V_i$
	performs an arbitrary number of $Y$ I/O operations. 
	We still have $|V_{R,i}| \le S$, $|W_{R,i}| \le S$, $0 \le 
	|W_{B,i}|$ (by the
	definition of the red-blue pebble game rules).
	Moreover, observe that, because we perform exactly $Y$ I/O operations in 
	each
	subcomputation, and all the vertices in $V_{B,i}$ by definition have to be
	loaded, $|V_{B,i}| \le Y$. A similar argument gives $0 \le |W_{B,i}| \le Y$.
	
	Denote an upper bound on $|V_{R,i}|$ and $|W_{R,i}|$ as $R(S)$ \linebreak 
	($\forall_i 
	\max\{|V_{R,i}|,|W_{R,i}|\} \le 
	R(S) \le S$).  Further, denote a lower bound on  $|V_{B,i}|$ and 
	$|W_{B,i}|$ as 
	$T(S)$ 
	($\forall_i  0 \le T(S) \le \min\{|V_{B,i}|,|W_{B,i}|\}$). We can use 
	$R(S)$ and  
	$T(S)$ to  
	tighten
	the bound on $Q$.
	We call $R(S)$ a \emph{maximum reuse} and $T(S)$ a \emph{minimum I/O} of
	a CDAG.

	\subsubsection{Reuse-Based Lemma}
	
	We now use the above definitions and observations to \emph{\textbf{generalize the result of Hong and 
	Kung~\cite{redblue}}}.

\begin{lma}
	\label{lma:reuseCalculation}
	An optimal complete calculation of a CDAG $G = (V,E)$, which performs $q$ 
	I/O operations, is associated with an $X$-partition of $G$ such that 
	
	$$q \ge (X - R(S) + T(S)) \cdot (h - 1)$$
	
	for any value of $X \ge S$, where $h$ is the number of subcomputations in 
	the $X$-partition, $R(S)$ is the maximum reuse set size, and $T(S)$ is the 
	minimum I/O in the given $X$-partition.
\end{lma}	

	\begin{proof}
		We use analogous 
		reasoning as in the original lemma. 
		We associate 
		the optimal pebbling with $h$ consecutive subcomputations 
		$V_1, \dots V_h$ with the difference that each 
		subcomputation $V_i$ performs $Y=X-R(S)+T(S)$ I/O operations.
		Within those $Y$ operations, we consider separately $q_{i,s}$ 
		store and $q_{i,l}$ load 
		operations.  For each $V_i$ we have $q_{i,s} + q_{i,l} = Y$, $q_{i,s} 
		\ge T(S)$, and
		$q_{i,l} \le Y - T(S) = X -R(S)$. 
		
		\vspace{-0.5em}
		\begin{gather}
		\nonumber
		\forall_{i}: |V_{B,i}| \le q_{l,i} \le Y-T(S) \\
		\nonumber
		\forall_{i}: |V_{R,i}| \le q_{s,i} \le R(S) \le S
		\end{gather}
		\vspace{-0.5em}
		
		Since $V_{R,i} \cup V_{B,i} = Dom(V_i)$:
		
		\vspace{-0.5em}
		\begin{gather}
		\label{eq:proof}
		\nonumber
		|Dom(V_i)| \le |V_{R,i}| + |V_{B,i}| \\
		\nonumber
		|Dom(V_i)| \le R(S) + Y- T(R) =  X
		\end{gather}
		\vspace{-0.5em}
		
		By an analogous construction for store 
		operations, we show that $|Min(V_i)| \le X$. To show that $S(X) = \{V_1 
		\dots 
		V_{h}\}$ meets the remaining properties of a valid 
		$X$-partition $\mathcal{S}(X)$, we use the same reasoning as originally 
		done~\cite{redblue}. 
		
		Therefore, a complete calculation 
		performing $q > (X - R(S) + T(S)) \cdot (h - 1)$ I/O operations has 
		an 
		associated $\mathcal{S}(X)$, such that $|\mathcal{S}(X)| 
		= h$  (if $q = (X - R(S) + T(S)) \cdot (h - 
		1)$, then $|\mathcal{S}(X)| = h-1$).
	\end{proof}

\noindent
From the previous lemma, we obtain a tighter I/O lower bound.

\begin{lma}
	\label{lma:reuse}
	Denote $H(X)$ as the minimum number of subcomputations
	in any valid $X$-partition of a CDAG $G = (V, E)$, for any $X \ge S$.
	The minimal number $Q$ of I/O operations for any valid execution of a 
	CDAG 
	$\ G=(V,E)$ is bounded by  
	
	\vspace{-0.5em}
	\begin{equation}
	Q \ge (X - R(S) + T(S)) \cdot (H(X) - 1)
	\label{eq:reusebound} \end{equation}
	\vspace{-0.5em}
	
	\noindent
	where $R(S)$ is the maximum reuse set size and $T(S)$ is the 
	minimum I/O 
	set size.
	Moreover, we have
	
	\vspace{-0.5em}
	\begin{equation}\label{eq:reusebound-pmax}
	H(X) \ge \frac{|V|}{|V_{max}|}
	\end{equation}
	\vspace{-0.5em}
	
	\noindent
	where $V_{max} = \argmax_{V_i \in \mathcal{S}(X)}|V_i|$ is 
	the largest
	subset of vertices in the CDAG schedule $\mathcal{S}(X) = \{V_1, \dots, 
	V_h\}$.
\end{lma}

\begin{proof}
	By definition, $H(X) = 
	\min_{\mathcal{S}(X)}|\mathcal{S}(X)| \le h$, so $Q \ge (X
	- R(S) + T(S)) \cdot (H(X) - 1)$ immediately follows from 
	Lemma~\ref{lma:reuseCalculation}.
	
	To prove Eq.~(\ref{eq:reusebound-pmax}), observe that $V_{max}$ by 
	definition
	is the largest subset in the optimal $X$-partition. As the subsets are
	disjoint, any other subset covers fewer remaining vertices to be 
	pebbled 
	than
	$V_{max}$. Because there are no cyclic dependencies 
	between subsets, we can
	order them topologically as $V_1, V_2, ...V_{H(X)}$. To 
	ensure that the indices are correct,
	we also define $V_0 \equiv \emptyset$. Now, define $W_i$ to be the set
	of vertices not included in any subset from $1$ to $i$, that is $W_i = 
	V -
	\bigcup_{j=1}^{i} V_j$. Clearly, $W_0 = V$ and $W_{H(X)} = \emptyset$. 
	Then, we
	have
	
	\vspace{-1em}
	\begin{alignat}{2}
	\nonumber
	\forall_{i}\quad |V_i| & \le |V_{max}| \\
	\nonumber
	|W_i| = |W_{i-1}| - |V_i| & \ge |W_{i-1}| - |V_{max}| \ge |V| - 
	i|V_{max}| 
	\\
	\nonumber
	|W_{H(X)}| = 0 & \ge |V| - H(X) \cdot |V_{max}| 
	\end{alignat}
	\vspace{-1em}
	that is, after $H(X)$ steps, we have $H(X) |V_{max}| \ge |V|$.
\end{proof}

	From this lemma, we derive the following lemma that we use to prove a 
tight 
I/O lower bound for MMM (Theorem~\ref{thm:seqlowbounds}):

\begin{lma}
	\label{lma:comp_intesity}
	Define the number of computations performed by $V_i$ for one loaded 
	element as the \emph{computational 
		intensity} $\rho_i = \frac{|V_i|}{X - |V_{R,i}| + 
		|W_{B,i}|}$ of the subcomputation 
	$V_i$.
	Denote $\rho = \max_i(\rho_i) \le 
	\frac{|V_{max}|}{X-R(S)+T(S)}$ to be 
	the \emph{maximal computational intensity}.
	Then, the number of I/O operations $Q$ is bounded by $Q \ge 
	{|V|}/{\rho}$.
\end{lma}

\begin{proof}
	Note that the term $H(X) - 1$ in Equation~\ref{eq:reusebound} emerges from 
	a fact that the last subcomputation may execute less than $Y - R(S) + T(S)$ 
	I/O operations, since $|V_{H(X)}| \le |V_{max}|$. However, because $\rho$ 
	is defined as maximal computational intensity, then performing $|V_{H(S)}|$ 
	computations requires at least $Q_{H(S)} \ge |V_{H(S)}| / \rho$. The total 
	number of I/O operations therefore is:
	
	$$Q = \sum_{i = 1}^{H(X)}Q_{i} \ge \sum_{i = 1}^{H(X)}\frac{|V_{i}|}{\rho} 
	= \frac{|V|}{\rho}$$
\end{proof}

\vspace{-1em}
	\section{Tight I/O Lower Bounds for MMM}
  \label{sec:seqOpt}

	In this section, we present our main theoretical contribution: a 
	constructive proof of a tight I/O lower bound for classical 
	matrix-matrix multiplication. 
	In~\cref{sec:parOptimality}, we extend it to 
	the parallel setup (Theorem~\ref{thm:parSchedule}). This result is tight 
	(up to diminishing factor $\sqrt{S}/(\sqrt{S+1}-1)$), and therefore may be 
	seen as the last step in the long sequence of improved bounds.
	Hong and Kung~\cite{redblue} derived an asymptotic bound 
		$\Omega\left({n^3}/{\sqrt{S}}\right)$ for the sequential case.
	Irony et al.~\cite{IronyMMM} extended the 
	lower bound result to a parallel machine with $p$ processes,
	each having a fast private memory of size~$S$, proving the
	$\frac{n^3}{4\sqrt{2}p\sqrt{S}} - S$ 
	lower bound on the communication volume per 
	process. Recently, Smith and van de
	Gein~\cite{tightMMM} proved a tight sequential lower bound (up to an 
	additive
	term) of $2mnk/\sqrt{S} - 2S$. Our proof improves the additive term and 
	extends it to a parallel schedule.
	
	\begin{thm}[Sequential Matrix Multiplication I/O lower bound] 
		Any pebbling of MMM CDAG which multiplies matrices of sizes $m \times 
		k$ and $k \times 
		n $ by performing $mnk$ multiplications requires a 
		minimum 
		number of $\frac{2mnk}{\sqrt{S}} + mn$ I/O operations.
		\label{thm:seqlowbounds}
	\end{thm}
	
	The proof of Theorem~\ref{thm:seqlowbounds} requires 
	Lemmas~\ref{lma:greedy} 
	and~\ref{lma:non_greedy}, 
	which in turn, require several definitions.

	\noindent
	\begin{small}
	\emph{Intuition: 	Restricting the analysis to greedy schedules provides  
	explicit information of a 
	state 
	of memory (sets $V_r$, $V_{R,r}$, $W_{B,r}$), and to a 
	corresponding 
	CDAG pebbling. Additional constraints~(\cref{sec:seqScheduling}) guarantee 
	feasibility of a derived schedule (and therefore, lower bound tightness).}
	\end{small} 
	
		\vspace{-0.5em}
	\subsection{Definitions}
	
	\subsubsection{Vertices, Projections, and Edges in the MMM CDAG} 
	The set of vertices of MMM CDAG $G = (V,E)$ consists of three subsets $V = 
	\mathcal{A} \cup \mathcal{B} \cup \mathcal{C}$, which correspond to 
	elements in 
	matrices  
	$A$, $B$, and $mnk$ partial sums of $C$.
	Each vertex $v$ is defined uniquely by a pair $(M,T)$,  where $M \in 
	\{a,b,c\}$ determines to which subset $\mathcal{A}$, $\mathcal{B}$, 
	$\mathcal{C}$ vertex $v$ belongs to, and $T \in \mathbb{N}^d$ is a vector 
	of coordinates, $d = 2$ for $M = a \lor b$ and $d=3$ for $M=c$.  E.g., $v = 
	(a, (1,5)) \in \mathcal{A}$ is a vertex associated with element $(1,5)$ in 
	matrix $A$, and 
	$v = (c,(3,6,8)) \in \mathcal{C}$ is associated with 8th partial sum of 
	element $(3,6)$ of 
	matrix $C$.	
	
	For every $t_3$th partial update of element $(t_1,t_2)$ in matrix $C$, and 
	an 
	associated point $v = (c, (t_1,t_2,t_3)) \in 
	\mathcal{C}$ we define $\phi_{c}(v) = (t_1,t_2)$ to be a 
	\emph{projection} of this
	point to matrix $C$,
	$\phi_{a}(v) = (a,(t_1,t_3)) \in \mathcal{A}$ is its
	projection to matrix $A$, and
	$\phi_{b}(v) = (b,(t_3,t_2)) \in \mathcal{B}$ is its projection
	to matrix $B$.
	Note that while $\phi_{a}(v), \phi_{b}(v) \in V$, projection $\phi_{c}(v) 
	\notin V$ 
	has not any associated point in $V$.
	Instead, vertices associated with all $k$ partial updates of an element of 
	$C$ have the same projection $\phi_{c}(v)$:
	
	\vspace{-1em}
	\begin{multline}
	\label{eq:verticalDeps}
	\forall_{v = (c,(p_1,p_2,p_3)), w = (c,(q_1,q_2,q_3)) \in \mathcal{C}}:  
	(p_1 = 
	q_1) \land (p_2 
	= 
	q_2) \\
	\iff \phi_{c}(p) = \phi_{c}(q)
	\end{multline}
		\vspace{-0.5em}
	
		As a consequence, 
	$\phi_{c}((c,(t_1,t_2,t_3))) = \phi_{c}((c,(t_1,t_2,t_3 -1)))$.
	
	A $t_3$th update of $(t_1,t_2)$ element in matrix $C$ of a classical MMM is 
	formulated as 	
	$C(t_1,t_2,t_3) 
	= C(t_1,t_2,t_3 - 1) + A(t_1,t_3) \cdot B(t_3,t_2)$. Therefore for each $v 
	= (c,(t_1,t_2,t_3)) \in \mathcal{C}, t_3 > 1$, we have following edges in 
	the CDAG: $(\phi_a(v), v)$, $(\phi_b(v), v)$, $(c,(t_1,t_2,t_3 -1)), v) \in 
	E$. 
	
	\subsubsection{$\bm{\alpha, \beta, \gamma, \Gamma}$}
	For a given subcomputation $V_r 
	\subseteq \mathcal{C}$, we denote its projection to matrix $A$ as $\alpha_r 
	= \phi_{a}(V_r) = \{v : v = \phi_{a}(c), c \in 
	V_r\}$, 
	its projection to 
	matrix $B$ as $\beta_r = \phi_{b}(V_r) $, and its 
	projection to matrix $C$ as $\gamma_r = \phi_{c}(V_r)$. We further 
	define 
	$\Gamma_r \subset \mathcal{C}$ as a set of all vertices in $\mathcal{C}$ 
	that 
	have a child in $V_r$.
	The sets $\alpha, \beta, \Gamma$ therefore correspond to the 
	inputs of $V_r$ that belong to matrices $A$, $B$, and 
	previous partial results of $C$, respectively. These inputs 
	form a minimal dominator set of $V_r$:
	
	\vspace{-0.5em}
	\begin{equation}
	\label{eq:dom}
	Dom(V_r) = \alpha_r \cup \beta_r \cup \Gamma_r
	\end{equation}
		\vspace{-0.5em}
	
	Because $Min(V_r) \subset \mathcal{C}$, and each vertex $v \in \mathcal{C}$ 
	has at most one child $w$ with $\phi_c(v) = \phi_c(w)$ 
	(Equation~\ref{eq:verticalDeps}), the projection
	$\phi_c(Min(V_r))$ is also equal to $\gamma_r$ :
	
		\vspace{-0.5em}
	\begin{equation}
	\label{eq:eq_projections}
	\phi_{c}(V_r) = \phi_{c}(\Gamma_r) = \phi_{c}(Min(V_r)) = 
	\gamma_r
	\end{equation}
		\vspace{-1em}
	
	\subsubsection{$\bm{Red()}$}
	Define $Red(r)$ as the set of all vertices that have red pebbles just 
	before 
	subcomputation $V_r$ starts, with $Red(1) = \emptyset$. We further have 
	$Red(P), P \subset V$ is the set of all vertices in some subset $P$ that 
	have 
	red pebbles and $Red(\phi_{c}(P))$ is a set of unique pairs of first two 
	coordinates of vertices in $P$ that have red pebbles.
	
	\subsubsection{Greedy schedule}
	We call a schedule $\mathcal{S} = \{V_1, \dots, V_h\}$ \emph{greedy} if
	during every 
	subcomputation $V_r$ every vertex $u$ that will hold a red pebble either 
	has a 
	child in $V_r$ or belongs to $V_r$:
	\begin{equation}
	\label{eq:greedy}
	\forall_{r}: Red(r) \subset \alpha_{r-1} \cup \beta_{r-1} \cup V_{r-1}
	\end{equation}
	
	\subsection{I/O Optimality of Greedy Schedules}
	
	\begin{lma}
		\label{lma:greedy}
		Any greedy schedule that multiplies matrices of sizes $m 
		\times k$ and $k 
		\times n$ using $mnk$ multiplications requires a 
		minimum 
		number of $\frac{2mnk}{\sqrt{S}} + mn$ I/O operations.
	\end{lma}
	
	\begin{proof}
		We start by creating an $X$-partition for an MMM CDAG (the values of 
		$Y$ and $R(S)$
		are parameters that we determine in the course of the proof).
		The proof is divided into the following 6 steps 
		(Sections~\ref{sec:stepI} to~\ref{sec:greedy_result}).
		
		\subsubsection{Red Pebbles During and After Subcomputation}
		\label{sec:stepI}
		
		Observe that each vertex in $c = (t_1,t_2,t_3) \in 
		\mathcal{C}, t_1 = 1 \dots m, t_2 = 1 \dots n, t_3 = 1 \dots k -1$ has 
		only one 
		child $c = (t_1, t_2, t_3 + 1)$. Therefore, we can assume that in an 
		optimal 
		schedule there are no two vertices $(t_1,t_2,t_3), (t_1,t_2,t_3 + f) 
		\in 
		\mathcal{C}, f \in \mathbb{N}_+$ that simultaneously hold a red vertex, 
		as when 
		the vertex $(t_1,t_2,t_3 +1)$ is pebbled, a red pebble can be 
		immediately 
		removed from $(t_1,t_2,t_3)$:
		
		\begin{equation}
		\label{eq:red_proj}
		|Red(V_r)| = |\phi_{c}(Red(V_r))|
		\end{equation}
		
		On the other hand, for every vertex $v$, if all its predecessors 
		$Pred(v)$ 
		have red 
		pebbles, then vertex $v$ may be immediately computed, freeing a red 
		pebble from its 
		predecessor $w \in \mathcal{C}$, due to the fact, that $v$ is the only 
		child of $w$:
		\begin{equation}
		\label{eq:immediate_compute}
		\forall_{v \in V} \forall_r : Pred(v) \subset Dom(V_r) \cup V_r 
		\implies 
		\exists_{t \le r} v \in V_t
		\end{equation}
		
		Furthermore, after subcomputation $V_r$, all vertices 
		in 
		$V_r$ that have red pebbles are in its minimum set: 
		\begin{equation}
		\label{eq:eq_red_projs}
		Red(r+1) \cap V_r = 
		Red(r+1) \cap Min(V_r)
		\end{equation}
		
		Combining this result with the definition of a greedy schedule 
		(Equation~\ref{eq:greedy}), we have
		
		\begin{equation}
		\label{eq:red_state}
		Red(r+1) \subseteq \alpha_r \cup \beta_r \cup Min(V_r)
		\end{equation}
		
		\subsubsection{Surface and volume of subcomputations}
		\label{sec:subcomp_vol}
		By the definition of $X$-partition,	the  
		computation is divided into
		$H(X)$ subcomputations $V_r \subset \mathcal{C}, r \in 
		\{1,\dots H(X)\}$, such that $Dom(V_r), 
		Min(V_r) \le 
		X$. 
		
		Inserting
		Equations~\ref{eq:dom}, \ref{eq:eq_projections}, 
		and~\ref{eq:red_proj}, we have:
		
		\begin{gather}
		\label{eq:dm}
		|Dom(V_r)| = |\alpha_r| + |\beta_r| + |\gamma_r| \le X \\
		\nonumber
		|Min(V_r)| = |\gamma_r| \le X
		\end{gather}
		
			On the other hand, the Loomis-Whitney
		inequality~\cite{loomis1949inequality}
		bounds the 
		volume of $V_r$: 
		
		\begin{equation}
		\label{eq:lm}
		V_r \le \sqrt{|\alpha_r| |\beta_r| |\gamma_r|}
		\end{equation}
		
		Consider sets of 
		all different indices accessed by projections $\alpha_r$, $\beta_r$, 
		$\gamma_r$:
		
		\begin{gather}
		\nonumber
		T_1 = \{t_{1,1}, \dots, t_{1,a}\}, |T_1| = a \\
		\nonumber
		T_2 = \{t_{2,1}, \dots, t_{2,b}\}, |T_2| = b \\
		\nonumber
		T_3 = \{t_{3,1}, \dots, t_{3,c}\}, |T_3| = c \\
		\label{eq:alphasubset}
		\alpha_r \subseteq \{(t_1, t_3) : t_1 \in T_1, t_3 \in T_3 \} \\
		\label{eq:betasubset}
		\beta_r \subseteq \{(t_3, t_2) : t_3 \in T_3, t_2 \in T_2 \} \\
		\label{eq:gammasubset}
		\gamma_r \subseteq \{(t_1, t_2) : t_1 \in T_1, t_2 \in T_2 \}\\
				\label{eq:Vsubset}
		V_r \subseteq \{(t_1, t_2,t_3) : t_1 \in T_1, t_2 \in T_2, t_3 \in T_3 
		\}
				\end{gather}				
		For fixed sizes of the projections $|\alpha_r|$, 
		$|\beta_r|$, $|\gamma_r|$, then the volume $|V_r|$ is maximized when 
		left and right side of Inequalities~\ref{eq:alphasubset} 
		to~\ref{eq:gammasubset} are equal. Using~\ref{eq:dom} and 
		~\ref{eq:immediate_compute} we have that~\ref{eq:Vsubset} is an 
		equality too, and:
		
		\begin{equation}
		\label{eq:rect_projs}
		|\alpha_r| = ac, |\beta_r| = bc, |\gamma_r| = ab, |V_r| = abc,
		\end{equation}
		achieving the upper bound (Equation~$\ref{eq:lm}$).
		
		\subsubsection{Reuse set $\bm{V_{R,r}}$ and store set $\bm{W_{B,r}}$}
		\label{sec:reuse_store_set}
		
		Consider two subsequent computations, $V_r$ and 
		$V_{r+1}$.
		After $V_r$, $\alpha_r$, $\beta_r$, and $V_r$ may 
		have red 
		pebbles (Equation~\ref{eq:greedy}). On the other hand, for 
		the dominator 
		set of 
		$V_{r+1}$ we have $|Dom(V_{r+1})| = |\alpha_{r+1}| + 
		|\beta_{r+1}| + 
		|\gamma_{r+1}|$.
		Then, the reuse set $V_{R,i+1}$ is an intersection of 
		those sets. Since $\alpha_r \cap \beta_r = \alpha_r 
		\cap \gamma_r = \beta_r \cap \gamma_r = \emptyset$, we 
		have (confront Equation~\ref{eq:red_state}):
		
		\begin{gather}	
		\nonumber
		V_{R,r+1}
		\subseteq (\alpha_r \cap \alpha_{r+1}) \cup
		(\beta_r \cap \beta_{r+1}) \cup 
		(Min(V_r) \cap \Gamma_{r+1}) \\
		\label{eq:vri1}
		|V_{R,r+1}|
		\le |\alpha_r \cap \alpha_{r+1}| +
		|\beta_r \cap \beta_{r+1}| + 
		|\gamma_{r} \cap \gamma_{r+1}|
		\end{gather}	
		
		Note that vertices in $\alpha_r$ and $\beta_r$ are inputs of 
		the 
		computation:
		therefore,
		by the definition of the red-blue pebble game, they start 
		in the slow memory (they 
		already have blue 
		pebbles). $Min(V_r)$, on the other hand, may have only red 
		pebbles
		placed on them.
		Furthermore, by the definition of the $S$-partition, 
		these 
		vertices have children that have not been pebbled yet. 
		They either have to be reused forming the reuse set 
		$V_{R,r+1}$, or
		stored back, forming $W_{B,r}$ and requiring the 
		placement of the blue 
		pebbles. Because $Min(V_r) \in \mathcal{C}$ and $\mathcal{C} 
		\cap \mathcal{A} = \mathcal{C} \cap \mathcal{B} = \emptyset$, 
		we have:
		\begin{gather}
		\nonumber
		W_{B,r} \subseteq  Min(V_r) \setminus
		\Gamma_{r+1} \\
		\label{eq:wbr}
		|W_{B,r}| \le |\gamma_{r} \setminus \gamma_{r+1}| 
		\end{gather}
		
		\subsubsection{Overlapping computations}
		\label{sec:overlap_comp}
		Consider two subcomputations $V_r$ and $V_{r+1}$. Denote 
		shared parts of their projections as  $\alpha_s = \alpha_r 
		\cap \alpha_{r+1}$, $\beta_s = \beta_r \cap \beta_{r+1}$, and 
		$\gamma_s = \gamma_r \cap \gamma_{r+1}$. Then, there are two 
		possibilities:
		\begin{enumerate}
			\item $V_r$ and $V_{r+1}$ are not cubic, resulting in their volume 
			smaller 
			than the upper bound $|V_{r+1}| < \sqrt{|\alpha_{r+1}| 
				|\beta_{r+1}| 
				|\gamma_{r+1}|}$ (Equation~\ref{eq:lm}),
			\item $V_r$ and $V_{r+1}$ are cubic. If all overlapping 
			projections are not empty, then they generate an overlapping 
			computation, that is, there exist vertices $v$, such that 
			$\phi_{ik}(v) \in \alpha_s, \phi_{kj}(v) \in \beta_s, 
			\phi_{ij}(v) \in \gamma_s$. Because we consider greedy 
			schedules, those vertices cannot belong to computation 
			$V_{r+1}$ (Equation~\ref{eq:immediate_compute}). Therefore, again 
			$|V_{r+1}| <$ \linebreak
			$\sqrt{|\alpha_{r+1}| |\beta_{r+1}| |\gamma_{r+1}|}$.
			Now consider sets of 
			all different indices accessed by those rectangular projections 
			(Section~\ref{sec:subcomp_vol}, Inequalities~\ref{eq:alphasubset} 
			to~\ref{eq:gammasubset}).
			Fixing two non-empty projections we define all three sets $T_1, 
			T_2, T_3$, 
			which in turn, generate the third (non-empty) 
			projection, resulting again in overlapping computations which 
			reduce the size 
			of $|V_{r+1}|$. Therefore, for cubic subcomputations, their volume 
			is maximized 
			$|V_{r+1}| =\sqrt{|\alpha_{r+1}| |\beta_{r+1}| |\gamma_{r+1}|}$ if 
			at most one 
			of the 
			overlapping projections is non-empty (and therefore, there is no 
			overlapping 
			computation).
		\end{enumerate}
		
		\subsubsection{Maximizing computational intensity}
		\label{sec:max_comp_intensity}
		Computational intensity $\rho_r$ of a subcomputation $V_r$ is an upper 
		bound on ratio 
		between 
		its size 
		$|V_r|$ and the number of I/O operations required. The number of I/O 
		operations is minimized when $\rho$ is maximized 
		(Lemma~\ref{lma:comp_intesity}):
		
		\begin{gather}
		\nonumber
		\text{maximize } \rho_r = 
		\frac{|V_{r}|}{X - R(S) + T(S)} \ge \frac{|V_r|}{Dom(V_r) - |V_{R,r}| + 
		|W_{B,r}|} 
		\\
		\nonumber
		\text{subject to:} \\
		\nonumber
		|Dom(V_{r})| \le X\\
		\nonumber
		|V_{R,r}| \le S
		\end{gather}
		
		To maximize the computational intensity, for a fixed number of I/O 
		operations, 
		the subcomputation size $|V_r|$ is maximized. Based on 
		Observation~\ref{sec:overlap_comp}, it is maximized only if at most one 
		of the 
		overlapping projections $\alpha_r \cap \alpha_{r+1}, \beta_r \cap 
		\beta_{r+1}, 
		\gamma_r \cap \gamma_{r+1}$ is not empty.
		Inserting Equations~\ref{eq:lm}, ~\ref{eq:dm},~\ref{eq:vri1}, 
		and~\ref{eq:wbr}, 
		we have the following three equations for the computational intensity, 
		depending on the non-empty projection:
		
		\begin{gather}
		\nonumber
		\alpha_r \cap \alpha_{r+1} \ne \emptyset: \\
		\rho_r = \frac{\sqrt{|\alpha_r| |\beta_r| |\gamma_r|}}{|\alpha_r| + 
			|\beta_r| + 
			|\gamma_r| - |\alpha_r \cap \alpha_{r+1}| + |\gamma_r|}\\
		\nonumber
		\beta_r \cap \beta_{r+1} \ne \emptyset: \\
		\rho_r = \frac{\sqrt{|\alpha_r| |\beta_r| |\gamma_r|}}{|\alpha_r| + 
			|\beta_r| + 
			|\gamma_r| - |\beta_r \cap \beta_{r+1}| + |\gamma_r|}\\
		\nonumber
		\gamma_r \cap \gamma_{r+1} \ne \emptyset: \\
		\label{eq:ne_gamma}
		\rho_r = \frac{\sqrt{|\alpha_r| |\beta_r| |\gamma_r|}}{|\alpha_r| + 
			|\beta_r| + 
			|\gamma_r| - |\gamma_r \cap \gamma_{r+1}| + |\gamma_r \setminus 
			\gamma_{r+1}|}
		\end{gather}
		
		$\rho_r$ is maximized when $\gamma_r = \gamma_{r+1}, \gamma_r \cap 
		\gamma_{r+1}  \ne \emptyset, \gamma_r \setminus \gamma_{r+1} 
		= \emptyset$ (Equation~\ref{eq:ne_gamma}).
		
		Then, inserting Equations~\ref{eq:rect_projs}, we have:
		\begin{gather}
		\nonumber
		\text{maximize } \rho_r = \frac{abc}{ac + cb} 
		\\
		\nonumber
		\text{subject to:} \\
		\nonumber
		ab + ac + cb \le X \\
		\nonumber
		ab \le S \\
		\nonumber		
		a, b, c \in \mathbb{N}_{+},
		\end{gather}
		
		where $X$ is a free variable.
		Simple optimization technique using Lagrange multipliers yields the 
		result:
		
		\begin{gather}
		a = b =  \lfloor\sqrt{S} \rfloor, c = 1, \\
		\nonumber
		|\alpha_r| = |\beta_r|=  \lfloor \sqrt{S}  \rfloor, 
		|\gamma_r| =  \lfloor \sqrt{S}  \rfloor^2, \\
		\nonumber
		|V_r| =  \lfloor \sqrt{S}  \rfloor^2, X =  \lfloor 
		\sqrt{S}  \rfloor^2 + 2 
		\lfloor \sqrt{S}  \rfloor \\
				\label{eq:seqSolution}
		\rho_r = \frac{\lfloor \sqrt{S} \rfloor}{2}
		\end{gather}
		
		From now on, to keep the calculations simpler, we use assume that 
		$\sqrt{S} \in \mathbb{N}+$. 
		
		\subsubsection{MMM I/O complexity of greedy schedules}
		\label{sec:greedy_result}
		By the computational intensity corollary (cf. page 4 in the main 
		paper): 
		$$Q \ge
		\frac{|V|}{\rho} = \frac{2mnk}{\sqrt{S}}$$ 
		This is the I/O cost of putting a red pebble at least once on every 
		vertex 
		in 
		$\mathcal{C}$. Note however, that we did not put any
		blue pebbles on the outputs yet (all vertices in $\mathcal{C}$ had only 
		red
		pebbles placed on them during the execution). By
		the definition of the red-blue pebble game,
		we need to place blue pebbles on $mn$ 
		output
		vertices, corresponding to the output matrix $C$, resulting in 
		additional 
		$mn$
		I/O operations, yielding final bound
		$$Q \ge \frac{2mnk}{\sqrt{S}} + mn$$ 
	\end{proof}

	\subsubsection{Attainability of the Lower Bound}
	\label{sec:seqScheduling}
	
	Restricting the analysis to greedy schedules provides  
	explicit information of a 
	state 
	of memory (sets $V_r$, $V_{R,r}$, $W_{B,r}$), and therefore, to a 
	corresponding 
	CDAG pebbling. 	
	In Section~\ref{sec:max_comp_intensity},
	it is proven that an optimal greedy 
	schedule is composed of $\frac{mnk}{R(S)}$ outer product calculations, 
	while 
	loading $\sqrt{R(S)}$ elements of each of matrices $A$ and $B$. While the 
	lower 
	bound is achieved for $R(S)=S$, such a schedule is infeasible, as at least 
	some 
	additional red pebbles, except the ones placed on the reuse set $V_{R,r}$, have to be placed on $2\sqrt{R(S)}$ vertices of $A$ and 
	$B$.
	
	A direct way to obtain a feasible greedy schedule is to set $X=S$, 
	ensuring that the dominator set can fit into the 
	memory. 
	Then each subcomputation is an outer-product of column-vector of matrix $A$ 
	and row-vector of $B$, both holding $\sqrt{S+1}-1$ values.
	Such a schedule performs $\frac{2mnk}{\sqrt{S+1}-1} +mn$ I/O operations, a 
	factor of $\frac{\sqrt{S}}{\sqrt{S+1}-1}$ more than a lower bound, which 
	quickly approach 1 for large $S$. Listing~\ref{lst:pseudocode} provides a 
	pseudocode of this algorithm, which is a well-known rank-1 update 
	formulation of MMM.
	However, we can do better.
	
	Let's consider a generalized case of such subcomputation $V_r$. Assume, 
	that in 
	each step:
	\begin{enumerate}
		\item 	 $a$ elements of $A$ (forming $\alpha_r$) are loaded,
		\item  $b$ elements of $B$ (forming $\beta_r$) are loaded,
		\item $ab$ partial results of $C$ are kept in the fast memory (forming 
		$\Gamma_r$)
		\item $ab$ values of $C$ are updated  (forming $V_r$),
		\item no store operations are performed.
	\end{enumerate}
	 Each vertex in $\alpha_r$ has $b$ 
	children in $V_r$ (each of which has also a parent in $\beta_r$). 
	Similarly, each vertex in $\beta_r$ has $a$ children in $V_r$, each of 
	which has also a parent in $\alpha_r$.
	We first note, that $ab < S$ (otherwise, we cannot do any computation while 
	keeping all $ab$ partial results in fast memory). Any red vertex placed on 
	$\alpha_r$ should not be removed from it until all $b$ children are 
	pebbled, requiring red-pebbling of corresponding $b$ vertices from 
	$\beta_r$. But, in turn, any red pebble placed on a vertex in $\beta_r$ 
	should not be removed until all $a$ children are red pebbled.
	
	Therefore, either all $a$ vertices in $\alpha_r$, or all $b$ vertices in 
	$\beta_r$ have to be hold red pebbles at the same time, while at least one 
	additional red pebble is needed on $\beta_r$ (or $\alpha_r$). W.l.o.g., 
	assume we keep red pebbles on all vertices of $\alpha_r$.	
	We then have:
			\begin{gather}
	\nonumber
	\text{maximize } \rho_r = \frac{ab}{a + b} 
	\\
	\nonumber
	\text{subject to:} \\
	\nonumber
	ab + a + 1 \le S \\
	a, b \in \mathbb{N}_{+},
	\end{gather}
	
	The solution to this problem is 
	\begin{gather}
	\label{eq:aopt}
	 a_{opt} = \left\lfloor \frac{\sqrt{{\left(S-1\right)}^3}-S+1}{S-2} \right 
	\rfloor < \sqrt{S} \\
	\label{eq:bopt}
 	b_{opt} = \left\lfloor 
 	-\frac{2\,S+\sqrt{{\left(S-1\right)}^3}-S^2-1}{\sqrt{{\left(S-1\right)}^3}-S+1}
 	 \right \rfloor < \sqrt{S}
	\end{gather}
	

	\begin{lstlisting}[float=h, caption={Pseudocode of near optimal 
	sequential MMM}, 
	label=lst:pseudocode]
for $i_1 = 1:\left\lceil \frac{m}{a_{opt}} \right\rceil$
  for $j_1 = 1:\left\lceil \frac{n}{b_{opt}} \right\rceil$
    for $r = 1:k$    
      for $i_2 = i_1 \cdot T: \min((i_1 + 1) \cdot a_{opt}, m)$
        for $j_2 = j_1 \cdot T: \min((j_1 + 1) \cdot b_{opt}, n)$
          $C(i_2,j_2) = C(i_2,j_2) + A(i_2,r) \cdot 
	B(r,j_2)$
	\end{lstlisting}
	
	\subsection{Greedy vs Non-greedy Schedules}
	
	In
	~\cref{sec:greedy_result},
	 it is shown that the I/O lower 
	bound for any greedy schedule is $Q \ge \frac{2mnk}{\sqrt{S}} + mn$. 
	Furthermore, Listing~\ref{lst:pseudocode} provide a schedule that 
	attains this lower bound (up to a $a_{opt}b_{opt}/S$ factor).
	To prove that this bound applies to any schedule, we need to show, 
	that any non-greedy cannot perform better (perform less I/O 
	operations) than the greedy schedule lower bound.
	
	\begin{lma}
		\label{lma:non_greedy}
		Any non-greedy schedule computing classical matrix 
		multiplication performs at least $\frac{2mnk}{\sqrt{S}} + mn$ I/O 
		operations.
	\end{lma}
	
	\begin{proof}
		Lemma~\ref{lma:reuse} applies to any schedule and for any value of $X$. 
		Clearly, for any general schedule we cannot directly model $V_{R,i}$, 
		$V_{B,i}$, $W_{R,i}$, and $W_{B,i}$, and therefore $T(S)$ and $R(S)$. 
		However, it is always true that $0 \le T(S)$ and $R(S) \le S$. Also, 
		the dominator set formed in Equation~\ref{eq:dom} applies for any 
		subcomputation, as well as a bound on $|V_r|$ 
		from 
		Inequality~\ref{eq:lm}. 		
		We can then rewrite the computational intensity maximization problem:
		
				\begin{gather}
		\nonumber
		\text{maximize } \rho_r = 
		\frac{|V_{r}|}{X - R(S) + T(S)} \le 
		\frac{\sqrt{|\alpha_r||\beta_r||\gamma_r|}}{|\alpha_r| + |\beta_r| 
		+|\gamma_r|- S} 
		\\
		\nonumber
		\text{subject to:} \\
		\nonumber
		S < |\alpha_r| + |\beta_r| + |\gamma_r| = X\\
		\end{gather}
		This is maximized for $|\alpha_r| = |\beta_r| = |\gamma_r| = X/3$, 
		yielding
		
		$$\rho_r = \frac{(X/3)^{3/2}}{X - S}$$
		
		Because $mnk/\rho_r$ is a valid lower bound for any $X > S$ 	
		(Lemma~\ref{lma:comp_intesity}), we want to find such 
		value $X_{opt}$ for which 
		$\rho_r$ is minimal, yielding the highest (tightest) lower bound on $Q$:
		
		\begin{gather}
		\nonumber
			\text{minimize } \rho_r = \frac{(X/3)^{3/2}}{X - S}
		\\
		\nonumber
		\text{subject to:} \\
		\nonumber
		X \ge S \\
		\end{gather}	
		which, in turn, is minimized for $X = 3S$. This again shows, that the 
		upper bound on maximum computational intensity for any schedule is 
		$\sqrt{S}/2$, which matches the bound for greedy schedules 
		(Equation~\ref{eq:seqSolution}).
	\end{proof}

We note that Smith and van de Gein~\cite{tightMMM} in their paper also bounded 
the number of computations (interpreted geometrically as a subset in a 3D 
space) by its surface and obtained an analogous result for this surface (here, 
a dominator and minimum set sizes). However, using computational intensity 
lemma, our bound is tighter by $2S$ ($+mn$, counting storing the final 
result).
	
	\noindent
	\textbf{Proof of Theorem~\ref{thm:seqlowbounds}:}
	
	\noindent
	Lemma~\ref{lma:greedy} establishes that the I/O lower bound 
	for any greedy schedule is $Q = 2mnk/\sqrt{S} + mn$. 
	Lemma~\ref{lma:non_greedy} establishes that no other schedule 
	can perform less I/O operations.
	
	\qed

\macb{Corollary}: 
The greedy schedule associated with an $X=S$\nobreakdash -partition performs at 
most 
$\frac{\sqrt{S}}{\sqrt{S+1}-1}$ more I/O operations than a lower bound. The 
optimal greedy schedule is associated with an $X=a_{opt}b_{opt} + a_{opt} 
	+ b_{opt}$-partition and performs 
	$\frac{\sqrt{S}(a_{opt} 
		+ b_{opt})}{a_{opt}b_{opt}}$ I/O operations.
	 
  \section{Optimal Parallel MMM}
	\label{sec:parOptimality}
	
	We now derive the schedule of COSMA from the results
	from~\cref{sec:seqScheduling}. The key notion is the data reuse, that
	determines not only the sequential execution, as discussed
	in
	~\cref{sec:seq-proof}
	, but also the parallel 
	scheduling. Specifically, if the data reuse set spans across multiple local 
	domains, then this set 
	has 
	to be communicated between these domains, increasing the I/O cost 
	(Figure~\ref{fig:topdown-vs-bottomup}). 
	We first introduce a formalism required to 
	parallelize the sequential schedule
	~(\cref{sec:seqpar}). 
	In~\cref{sec:parStrategies}, we generalize parallelization strategies used 
	by 
	the
	2D, 2.5D, and recursive decompositions, deriving their communication cost 
	and 
	showing that none of them is optimal in the whole range of parameters. We 
	finally derive the optimal decomposition (\emph{FindOptimalDomain} function 
	in 
	Algorithm~\ref{alg:COSMA}) by expressing it as an optimization 
	problem~(\cref{sec:parScheduling}), and analyzing its I/O and latency cost. 
	The 
	remaining steps in Algorithm~\ref{alg:COSMA}: \emph{FitRanks}, 
	\emph{GetDataDecomp}, as well as  \emph{DistrData} and \emph{Reduce} are 
	discussed 
	in~\cref{sec:decompArbitrary}, \cref{sec:datalayout}, and 
	\cref{sec:commPattern}, 
	respectively.	
	For a distributed machine, we assume that all matrices fit into 
		collective memories of all processors: $pS \ge mn + mk + nk$.
	For a shared memory setting, we assume that all inputs start in a common 
	slow memory.
	
	\subsection{Sequential and Parallel Schedules}
	\label{sec:seqpar}
	
	We now describe how a parallel schedule is formed from a sequential 
	one. The sequential
	schedule~$\mathcal{S}$ partitions the CDAG $G = (V,E)$ into $H(S)$
	subcomputations $V_i$. The parallel schedule $\mathcal{P}$ divides
	$\mathcal{S}$ among $p$ processors: $\mathcal{P} = \{\mathcal{D}_1, \dots
	\mathcal{D}_p\}, \bigcup_{j=1}^p \mathcal{D}_j = \mathcal{S}$. The set
	$\mathcal{D}_j$ of all $V_k$ assigned to processor $j$ forms a 
	\emph{local
		domain} of $j$ (Fig.~\ref{fig:mmmParallelization}c).
	
	If two local domains $\mathcal{D}_k$ and  $\mathcal{D}_l$ are dependent, 
	that is,
	\linebreak $\exists u, \exists v:  u \in \mathcal{D}_k \land v \in 
	\mathcal{D}_l \land
	(u,v) \in 
	E$,
	then $u$ has to be \emph{communicated} from processor $k$ to $l$. The total
	number of vertices communicated between all processors is the \emph{I/O 
	cost} 
	$Q$ of schedule $\mathcal{P}$.  
	We say that the parallel schedule
	$\mathcal{P}_{opt}$ is
	\emph{communication--optimal} if $Q(\mathcal{P}_{opt})$ is minimal among all
	possible parallel schedules. 
	
	The vertices of MMM CDAG may be arranged in an $[m \times n \times k]$ 3D 
	grid 
	called an \emph{iteration space}~\cite{tiling}.
	The orthonormal vectors \textbf{i}, \textbf{j}, \textbf{k} 
	correspond to the loops in Lines 1-3 in Listing~\ref{lst:pseudocode} 
	(Figure~3a).
	We call a schedule $\mathcal{P}$
	\emph{parallelized in dimension} \textbf{d} if we ``cut'' the CDAG 
	along dimension \textbf{d}. More formally, each local domain 
	$\mathcal{D}_j, j 
	=\nobreak 1 \dots p$ is a grid of size either $[m/p,
	n, k]$,  $[m,
	n/p, k]$, or  $[m,
	n, k/p]$. 
	The schedule may also be parallelized in two dimensions
	($\mathbf{d}_1\mathbf{d}_2$) or three dimensions
	($\mathbf{d}_1\mathbf{d}_2\mathbf{d}_3$) with a local domain size $[m / p_m,
	n / p_n, k / p_k]$  for some $p_m, p_n, p_k$, such that
	$p_m p_n p_k = \nobreak p$. We call $\mathcal{G} = [p_m, p_n, p_k]$ the
	\emph{processor grid} of a schedule. E.g., Cannon's algorithm is 
	parallelized 
	in 
	dimensions \textbf{ij} , with the processor grid $[\sqrt{p}, \sqrt{p}, 1]$. 
	COSMA, 
	on the other hand, may use any of the possible parallelizations, depending 
	on 
	the problem parameters.
	
	\subsection{Parallelization Strategies for MMM}
	\label{sec:parStrategies}
	
	The 
	sequential schedule $\mathcal{S}$~(\cref{sec:seqOpt})
	consists of $mnk/S$ elementary outer product calculations, arranged in
	$\sqrt{S} \times \sqrt{S} \times k$ ``blocks'' 
	(Figure~\ref{fig:mmmParallelization}). The number $p_1 = mn/S$ of 
	dependency-free subcomputations 
	$V_i$
	(i.e., having no parents except for input vertices) in 
	$\mathcal{S}$ 
	determines the maximum 
	degree 
	of parallelism of $\mathcal{P}_{opt}$ for which no reuse set $V_{R,i}$ 
	crosses 
	two local 
	domains 
	$\mathcal{D}_j$, $\mathcal{D}_k$. The optimal schedule 
	is parallelized in dimensions \textbf{ij}. There is no communication 
	between 
	the 
	domains (except for inputs and outputs), and all I/O operations are 
	performed inside each $\mathcal{D}_j$ following the
	sequential 
	schedule. Each processor is  assigned to $p_1/p$ local domains
	$\mathcal{D}_j$ of size 
	$\big[\sqrt{S}, \sqrt{S}, k\big]$, each of which requires 
	$2\sqrt{S}k + S$ I/O 
	operations (Theorem~\ref{thm:seqlowbounds}), giving a total of $Q = 
	2mnk/(p\sqrt{S}) + mn/p$ I/O operations per processor.
	
	\begin{figure}
		\includegraphics[width=\columnwidth]{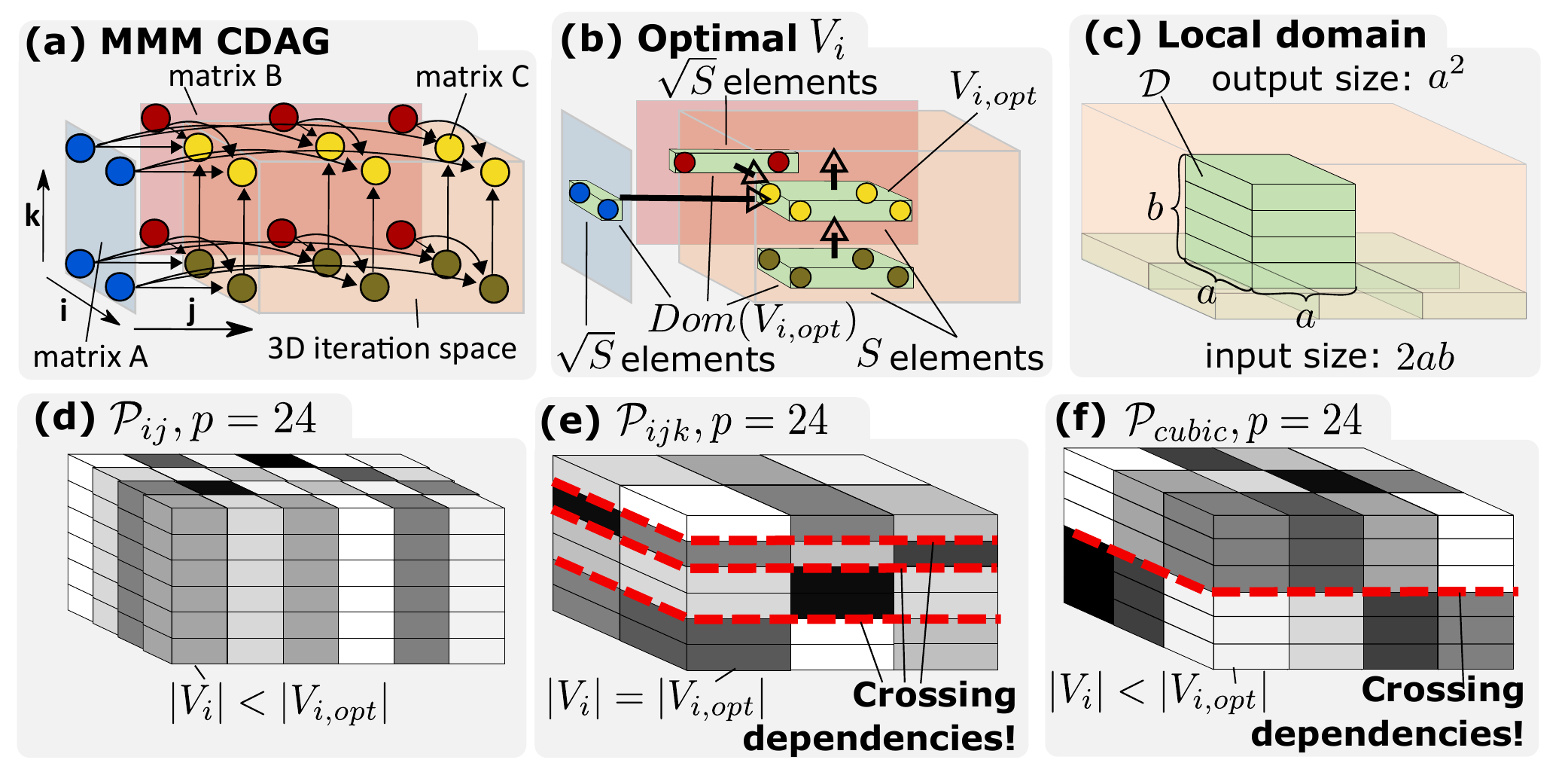}
		\caption{
			\textmd{
			(a) An MMM CDAG as a 3D grid (iteration 
			space).
			Each 
			vertex in it
			(except for the vertices in the bottom layer) has 
			three parents - blue (matrix $A$), 
			red (matrix $B$), 
			and yellow (partial result of matrix $C$) and one yellow child 
			(except for vertices 
			in the top layer). 
			(b) A union of inputs of all vertices in 
			$V_i$ form the 
			dominator 
			set $Dom(V_i)$ (two blue, two red and four dark yellow).  
			Using approximation $\sqrt{S+1}-1 \approx \sqrt{S}$, we have
			$|Dom(V_{i,opt})| = S$.
			(c) A local domain $\mathcal{D}$ consists of 
			$b$ 
			subcomputations $V_i$, each of a dominator size $|Dom(V_{i})| = 
			a^2 + 2a$.
			(d-f)
			Different parallelization schemes of near optimal sequential MMM 
			for 
			$p 
			= 24 > p_1 = 6$.} 
		}
		\label{fig:mmmParallelization}
	\end{figure}
	
	When $p > p_1$, the size of local domains $|\mathcal{D}_j|$ is smaller than 
	$\sqrt{S} 
	\times \sqrt{S} \times k $. Then, the schedule has to either be 
	parallelized in dimension~\textbf{k}, 
	or has to reduce the 
	size of the domain in \textbf{ij} plane.
	The former option
	creates
	dependencies between the local domains, which results in additional 
	communication (Figure~\ref{fig:mmmParallelization}e). 
	The latter does not utilize the whole available
	memory, making the sequential schedule not I/O optimal and decreasing the
	computational intensity $\rho$ (Figure~\ref{fig:mmmParallelization}d). We
	now analyze three possible parallelization
	strategies (Figure~\ref{fig:mmmParallelization}) which generalize 2D, 2.5D, 
	and 
	recursive decomposition strategies; see
	Table~\ref{tab:summary} for details.
	
	\macb{Schedule $\mathcal{P}_{ij}$} 
	The schedule is parallelized in dimensions \textbf{i} and \textbf{j}. The
	processor grid is $\mathcal{G}_{ij} = \big[\frac{m}{a} , \frac{n}{a}, 
	1\big]$, where $a = 
	\sqrt{\frac{mn}{p}}$.  
	Because 
	all
	dependencies are parallel to dimension \textbf{k}, there are no dependencies
	between $\mathcal{D}_j$ except for the inputs and the outputs.  Because $a <
	\sqrt{S}$, the corresponding sequential schedule 
	has a reduced
	computational intensity $\rho_{ij} < \sqrt{S}/2$.
	
	\macb{Schedule $\mathcal{P}_{ijk}$}
	The schedule is parallelized in all dimensions. The processor grid is
	$\mathcal{G}_{ijk} = \big[\frac{m}{\sqrt{S}}, \frac{n}{\sqrt{S}}, 
	\frac{pS}{mn}\big]$.  The
	computational intensity $\rho_{ijk} = \sqrt{S}/2$ is optimal. The 
	parallelization in
	\textbf{k} dimension creates dependencies between local domains, requiring
	communication and increasing the I/O cost.
	
	\macb{Schedule $\mathcal{P}_{cubic}$} 
	The schedule is parallelized in all dimensions. The grid is 
	$\big[\frac{m}{a_c} 
	,
	\frac{n}{a_c}, \frac{k}{a_c}\big]$, where $a_c = 
	\min\Big\{\big(\frac{mnk}{p}\big)^{1/3},$
	 $\sqrt{\frac{S}{3}}\Big\}$. Because $a_c < \sqrt{S}$, the corresponding
	computational intensity $\rho_{cubic}$ $< \sqrt{S}/2$ is not optimal.  
	The 
	parallelization in
	\textbf{k} dimension creates dependencies between local domains, 
	increasing 
	communication. 
	
	\macb{Schedules of the State-of-the-Art Decompositions} 
	If $m = n$, the $\mathcal{P}_{ij}$ scheme is reduced to the classical 2D
	decomposition (e.g., Cannon's algorithm~\cite{Cannon}), and 
	$\mathcal{P}_{ijk}$ 
	is reduced to the 2.5D decomposition~\cite{25d}.
	CARMA~\cite{CARMA} asymptotically reaches the $\mathcal{P}_{cubic}$ 
	scheme,
	guaranteeing that the longest dimension of a local cuboidal domain is at 
	most
	two times larger than the smallest one. We present a detailed complexity 
	analysis comparison for all algorithms in Table~\ref{tab:summary}.
	
	\begin{table*}
\sf
\footnotesize
		\begin{tabular}{lllll}
			\toprule
			\textbf{Decomposition} 
			& \textbf{2D~\cite{summa}} & 
			\textbf{2.5D~\cite{25d}} & 
			\textbf{recursive~\cite{CARMA}} & \textbf{COSMA (this paper)} \\
			\textbf{Parallel schedule $\mathcal{P}$} &
			$\mathcal{P}_{ij}$ for $m=n$ &
			$\mathcal{P}_{ijk}$ for $m=n$ &
			$\mathcal{P}_{cubic}$ &
			$\mathcal{P}_{opt}$ \\
			\midrule
			\makecell[l]{\textbf{grid} 
				$\left[p_m \times p_n \times p_k\right]$}
			&
			$\left[\sqrt{p} \times \sqrt{p} \times 1\right]$
			&
			\makecell[l]{$\left[\sqrt{p/c} \times \sqrt{p/c} \times c\right]$; 
				$c = \frac{pS}{mk + nk}$}
			& 
			\makecell[l]{$\left[{2^{a_1}} \times {2^{a_2}} \times 
				{2^{a_3}}\right]$; 
				$a_1 + a_2 + a_3 = \log_2(p)$}
			& 
			\makecell[l]{$\left[\frac{m}{a} \times \frac{n}{a} \times 
				\frac{k}{b}\right]$; 
				$a,b: $ Equation~\ref{eq:optTileShape}}
			\\
			\textbf{domain size}
			&
			$\left[\frac{m}{\sqrt{p}} \times \frac{n}{\sqrt{p}} \times 
			k\right]$ 
			&
			$\left[\frac{m}{\sqrt{p/c}} \times \frac{n}{\sqrt{p/c}} \times 
			\frac{k}{c}\right]$
			&
			$\left[\frac{m}{2^{a_1}} \times \frac{n}{2^{a_1}} \times 
			\frac{k}{2^{a_1}}\right]$
			& 
			$\left[a \times a \times b\right]$
			\\
			\midrule
			\multicolumn{5}{c}{\textbf{``General case'':}} \\
			\textbf{I/O cost} $Q$
			&
			$\frac{k}{\sqrt{p}} \left(m + n\right) + \frac{mn}{p}$
			&
			\makecell[l]{$\frac{\left(k(m+n)\right)^{3/2}}{p\sqrt{S}} + 
				\frac{mnS}{k(m+n)} $
			} 
			&
			\makecell[l]{$2\min \Big\{\sqrt{3} \frac{mnk}{p\sqrt{S}},
				\left(\frac{mnk}{p}\right)^{2/3} \Big\} $ 
				$ + \left(\frac{mnk}{p}\right)^{2/3}$}
			& 
			$\min \Big\{\frac{2mnk}{p\sqrt{S}} + S, 3 
			\left(\frac{mnk}{p}\right)^{2/3} \Big\}$
			\\
			\textbf{latency cost} $L$
			&
			$2k \log_2{(\sqrt{p})} $
			&
			$ \frac{\left(k(m+n)\right)^{5/2}}{pS^{3/2}(km + kn - mn)} + 
			3\log_2\left(\frac{pS}{mk + nk}\right)$
			&
			$ \left(3^{3/2}mnk\right)/\left(pS^{3/2}\right)+ 3\log_2(p)$
			& 
			$\frac{2ab}{S-a^2}\log_2\left(\frac{mn}{a^2}\right)$
			\\
			\midrule
			\multicolumn{5}{c}{\textbf{Square matrices, ``limited 
			memory'':} $m 
				= n = 
				k, S = 2{n^2}/{p}, p=2^{3n}$} \\
			\textbf{I/O cost} $Q$
			&
			$2n^2(\sqrt{p}+1)/p$
			&
			$2n^2(\sqrt{p}+1)/p$
			&
			$2n^2 \left(\sqrt{{3}/{2p}} + {1}/{2p^{2/3}} \right)$
			& 
			$2n^2(\sqrt{p}+1)/p$
			\\
			\textbf{latency cost} $L$
			&
			$2k \log_2{(\sqrt{p})}$
			&
			$\sqrt{p}$
			&
			$\left(\frac{3}{2}\right)^{3/2}\sqrt{p} \log_2(p)$
			& 
			$\sqrt{p}\log_2(p)$
			\\
			\midrule
			\multicolumn{5}{c}{\textbf{``Tall'' matrices, ``extra'' memory 
					available:} $m = 
				n = \sqrt{p}, k = {p^{3/2}}/{4}, S = 2{nk}/{p^{2/3}}, p=2^{3n + 
					1}$} \\
			\textbf{I/O cost}
			&
			${p^{3/2}}/{2}$
			&
			${p^{4/3}}/{2} + p^{1/3}$
			&
			${3p}/{4}$
			& 
			$p\left({3-2^{1/3}}\right)/{2^{4/3}} \approx 0.69 p$
			\\
			\textbf{latency cost} $L$
			&
			$p^{3/2} \log_2{(\sqrt{p})}/4$
			&
			$1$
			&
			$1$
			& 
			$1$
			\\
			\bottomrule
		\end{tabular}
		\caption{
			\textmd{The comparison of complexities of 2D, 2.5D, recursive, and 
		COSMA 
			algorithms. The 3D decomposition is a special case of 2.5D, and can 
			be 
			obtained 
			by
			instantiating $c=p^{1/3}$ in the 2.5D case.
			In 
			addition to the general analysis, we show two special cases. If the 
			matrices 
			are 
			square and there is no extra memory available, 2D, 2.5D and COSMA 
			achieves tight 
			communication lower bound $2n^2/\sqrt{p}$, whereas CARMA performs 
			$\sqrt{3}$ 
			times more 
			communication. If one dimension is much larger than the others and 
			there is 
			extra memory available, 2D, 2.5D and CARMA decompositions perform 
			$\mathcal{O}(p^{1/2})$, $\mathcal{O}(p^{1/3})$, and 8\%  more 
			communication 
			than COSMA, respectively. For simplicity, we assume that
			parameters are chosen such that all divisions have integer results.
		}
	}
\vspace{-3em}
		\label{tab:summary}
	\end{table*}
	
	\subsection{I/O Optimal Parallel Schedule}
	\label{sec:parScheduling}
	
	Observe that none of those schedules is optimal in the whole range of
	parameters. As discussed in~\cref{sec:seqOpt}, in sequential scheduling, 
	intermediate results of $C$ are not stored
	to the memory: they are consumed (reused) immediately by the next
	sequential step. Only the final result of $C$ in the local domain is sent.
	Therefore, the optimal parallel schedule $\mathcal{P}_{opt}$ minimizes the
	communication, that is, sum of the inputs' sizes plus the output size, under
	the sequential I/O constraint on subcomputations $\forall_{V_i \in 
		\mathcal{D}_j \in
		\mathcal{P}_{opt}} |Dom(V_i)| \le S \land |Min(V_i)| \le S$.
	
	The local domain $\mathcal{D}_j$ is a grid of size $[a \times a \times b]$, 
	containing $b$ outer products of vectors of length $a$. 
	The optimization problem of finding $\mathcal{P}_{opt}$
	using the 
	computational intensity~(Lemma \ref{lma:comp_intesity})
	is formulated as follows:
	
	\begin{gather}
	\label{eq:tileEq}
	\text{maximize } \rho = \frac{a^2b}{ab + ab + a^2}\\
	\nonumber
	\text{subject to: } \\
	\nonumber
	a^2 \le S \text{ (the I/O constraint)} \\
	\nonumber
	a^2b = \frac{mnk}{p} \text{ (the load balance constraint)} \\
	\nonumber
	pS \ge mn + mk + nk \text{ (matrices must fit into memory)}
	\end{gather}
	The I/O constraint $a^2 \le 
	S$ is binding (changes to equality) for $p \le \frac{mnk}{S^{3/2}}$. 
	Therefore, 
	the solution 
	to this problem is:
	
	\vspace{-1em}
	\begin{gather}
	a = \min\Big\{\sqrt{S}, \Big(\frac{mnk}{p}\Big)^{1/3} \Big\}, \text{ }
	b = \max\Big\{\frac{mnk}{pS}, \Big(\frac{mnk}{p}\Big)^{1/3} \Big\}
	\label{eq:optTileShape}
	\end{gather}
	\noindent
	The I/O complexity of this schedule is:
	\begin{gather}
		Q \ge \frac{a^2b}{\rho} = \min\Big\{\frac{2mnk}{p\sqrt{S}} + S, 
	3\Big(\frac{mnk}{p}\Big)^\frac{2}{3}\Big\}
	\end{gather} 
	\noindent
	This can be intuitively interpreted geometrically as follows: if we 
	imagine the 
	optimal local domain "growing" with the decreasing number of processors, 
	then it stays 
	cubic as long as it is still "small enough" (its side is 
	smaller than 
	$\sqrt{S}$). 
	After that point, its face in the \textbf{ij} plane stays constant 
	$\sqrt{S} \times 
	\sqrt{S}$ and it "grows" 
	only in the \textbf{k} dimension. This schedule effectively switches 
	from $\mathcal{P}_{ijk}$ to $\mathcal{P}_{cubic}$ once there is enough 
	memory ($S \ge ({mnk}/{p})^{2/3}$).
	
\vspace{+1em}
	\begin{thm}
		The I/O complexity of a classic Matrix Multiplication algorithm 
		executed on $p$ processors, each of local memory size $S \ge \frac{mn + 
		mk + nk}{p}$ is 
		$$	Q \ge \min\Big\{\frac{2mnk}{p\sqrt{S}} + S, 
		3\Big(\frac{mnk}{p}\Big)^\frac{2}{3}\Big\}$$
		\label{thm:parSchedule}
	\end{thm}
	
	\begin{proof}
		The theorem is a direct consequence of Lemma~\ref{lma:reuse} and 
		the computational intensity (Lemma~\ref{lma:comp_intesity}). The load 
		balance 
		constraint enforces a size of each local domain $|\mathcal{D}_j| = 
		mnk/p$. The 
		I/O cost is then bounded by $|\mathcal{D}_j|/\rho$. Schedule 
		$\mathcal{P}_{opt}$ maximizes $\rho$ by the formulation of the 
		optimization 
		problem~(Equation~\ref{eq:tileEq}).
	\end{proof}
\vspace{-0.5em}
		\macb{I/O-Latency Trade-off}
	As showed in this section, the local domain $\mathcal{D}$ of 
	the 
	near optimal 
	schedule $\mathcal{P}$ is a grid of size $[a \times a \times b]$, where 
	$a, 
	b$ are given by Equation~\eqref{eq:optTileShape}.
	 The corresponding sequential 
	schedule $\mathcal{S}$ is a sequence of $b$ outer products of vectors of 
	length 
	$a$. Denote the size of the communicated inputs in each step by $I_{step} 
	= 
	2a$. 
	This corresponds to $b$ steps of communication (the latency cost is $L = 
	b$).
	
	The number of steps (latency) is equal to the total communication volume 
	of 
	$\mathcal{D}$ divided by 
	the volume per step $L = Q/I_{step}$. To reduce the latency, one 
	either has to decrease $Q$ or increase $I_{step}$, under the memory 
	constraint that $I_{step} + a^2 \le S$ (otherwise we cannot fit both the 
	inputs 
	and 
	the outputs in the memory). Express $I_{step} = a \cdot h$, where $h$ is 
	the 
	number of sequential subcomputations $V_i$ we merge in one communication. 
	We 
	can express the I/O-latency trade-off:
	\begin{gather}
	\nonumber
	\min (Q, L) \\
	\nonumber
	\text{subject to:}\\
	\nonumber
	Q = 2ab + a^2, 
	L = \frac{b}{h} \\
	\nonumber
	a^2 + 2ah \le S \text{ (I/O constraint)} \\
	\nonumber
	a^2b = \frac{mnk}{p} \text{ (load balance constraint)}
	\end{gather}
	Solving this problem, we have $Q = \frac{2mnk}{pa} + a^2$ and $L 
	= 
	\frac{2mnk}{pa(S-a^2)}$, 
	where $a \le \sqrt{S}$. Increasing $a$ we \emph{reduce} the I/O cost 
	$Q$ and \emph{increase} the latency cost $L$. For minimal value of 
	$Q$ (Theorem~\ref{thm:parSchedule}),  $L = \left \lceil{\frac{2ab}{S - 
	a^2}}\right\rceil$, where $a = \min\{\sqrt{S}, (mnk/p)^{1/3} \}$ and 
	$b = 
	\max\{\frac{mnk}{pS}, (mnk/p)^{1/3} \}$. Based on our 
	experiments, we observe that the I/O cost is vastly greater than the 
	latency cost, therefore our schedule by default minimizes $Q$ and uses 
	extra memory (if any) to reduce $L$.
	
	\section{Implementation}
	\label{sec:implementation}

	We now present implementation optimizations that further 
	increase the performance of COSMA on top of the speedup due to our near I/O 
	optimal 
	schedule. The algorithm is designed to facilitate the overlap of 
	communication and computation~\cref{sec:compOverlap}. For 
	this, to leverage the RDMA 
	mechanisms of 
	current high-speed network interfaces, we use the MPI one-sided 
	interface~\cref{sec:rdma}. In addition, our implementation also offers 
	alternative efficient two-sided communication back end that uses MPI 
	collectives. 
	We also use a blocked data 
	layout~\cref{sec:datalayout}, a grid-fitting 
	technique~\cref{sec:decompArbitrary}, and an optimized binary broadcast 
	tree using static information about
	the communication pattern~(\cref{sec:commPattern}) together with the buffer 
	swapping~(\cref{sec:bufferReuse}). For the local matrix operations, we use 
	BLAS routines for highest performance. 
	Our code is publicly available at 
	https://github.com/eth-cscs/COSMA.
	
	\subsection{Processor Grid Optimization}
	\label{sec:decompArbitrary}
	
	\begin{figure}
		\centering
		\subfloat[$1 \times 5 \times 13$ 
		grid]{\includegraphics[width=0.17\textwidth]
			{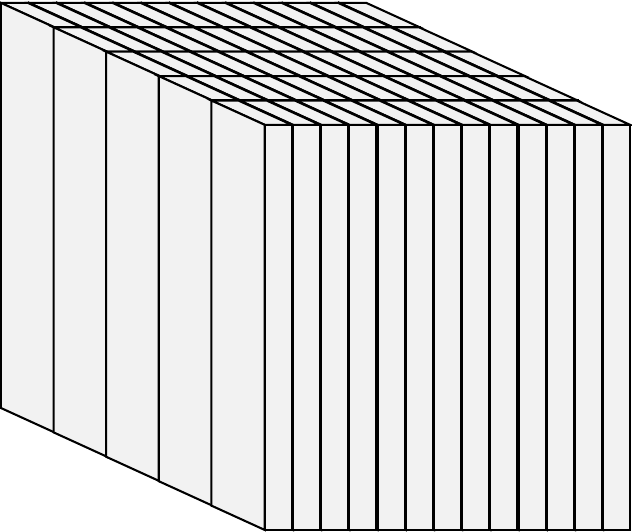}\label{fig:decomp_left}}
		\hfill
		\subfloat[$4 \times 4 \times 4$ grid with one idle 
		processor]{\includegraphics[width=0.26\textwidth]
			{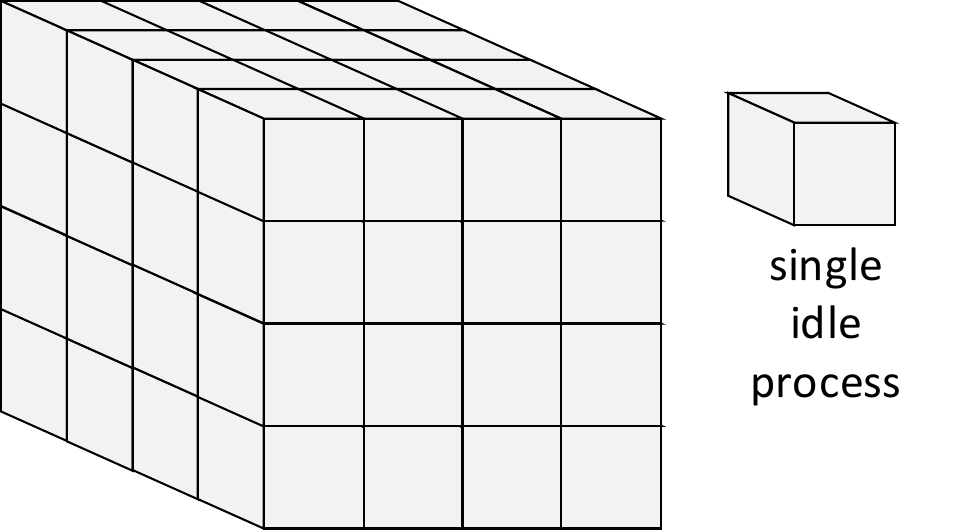}\label{fig:decomp_right}}
		\caption{\textmd{{Processor decomposition for square matrices 
		and 65 
			processors. (a) To utilize all resources, the local domain is 
			drastically stretched. (b) Dropping one 
			processor results in a symmetric grid  which 
			increases the computation per processor by 1.5\%, but reduces 
			the communication by 36\%.}}
				\vspace{-0.5em}}
		\label{fig:decompProblem}
	\end{figure}

	Throughout the paper, we assume all operations required to 
	assess the decomposition (divisions, roots) result in natural 
	numbers. We note that in practice it is rarely the case, as the parameters 
	usually emerge from external constraints, like a specification of a 
	performed 
	calculation 
	or 
	hardware resources~(\cref{sec:evaluation}). 
	If matrix dimensions are not divisible by the local domain sizes $a, b$ 
	(Equation~\ref{eq:optTileShape}), then a 
	straightforward option is to use the floor function, not utilizing the 
	``boundary'' processors whose local domains do not fit entirely in the 
	iteration 
	space, which result in more computation per processor. 
	The other option is to find factors of $p$ and then 
	construct 
	the processor grid by 
	matching the largest factors with largest matrix dimensions. 
	However, if the factors of $p$ do not match $m,n$, and $k$, this may 
	result 
	in a suboptimal decomposition.
	Our algorithm allows to not utilize some processors (increasing the 
	computation 
	volume per processor) to optimize the grid, which reduces 
	the 
	communication volume. 
	{Figure}~\ref{fig:decompProblem} {illustrates the comparison 
	between 
	these options.}
	We 
	balance this communication--computation
	trade-off by \linebreak"stretching" the local domain size 
	derived in ~\cref{sec:parScheduling} to fit the global 
	domain by adjusting its width, height, and length. The range of this 
	tuning (how many processors we drop to reduce communication) 
	depends on the hardware specification of the 
	machine (peak flop/s, memory and network bandwidth). For our 
	experiments on the Piz Daint machine we chose the maximal 
	number of unutilized cores to be~3\%, accounting for up to 
	2.4 times speedup for the square matrices using 2,198 
	cores~(\cref{sec:results}).
	
	
	
	
	\subsection{Enhanced Communication Pattern}
	\label{sec:commPattern}
	As shown in Algorithm~\ref{alg:COSMA}, COSMA by default executes in $t = 
	\frac{2ab}{S - a^2}$ 
	rounds. In each round, each processor receives $s = ab/t = (S - a^2)/2$ 
	elements of 
	$A$ 
	and 
	$B$. Thus, the input matrices are broadcast among the \textbf{i} and 
	\textbf{j} 
	dimensions of the processor grid. After the last round, the partial results 
	of 
	$C$ 
	are reduced among the \textbf{k} dimension. The communication pattern is 
	therefore similar to ScaLAPACK or CTF.
	
	To accelerate the collective communication, we implement our own binary 
	broadcast 
	tree, taking advantage of the known data layout, processor grid, and 
	communication 
	pattern. 
	Knowing the initial data 
	layout~\cref{sec:datalayout} and the processor 
	grid~\cref{sec:decompArbitrary}, we craft the binary reduction tree 
	in all three dimensions \textbf{i}, \textbf{j}, and \textbf{k} such that 
	the 
	distance in the grid between communicating processors is 
	minimized.
	Our implementation outperforms the standard MPI broadcast 
	from the 
	Cray-MPICH 
	3.1 library by approximately 10\%. 
	\subsection{Communication--Computation Overlap}
	\label{sec:compOverlap}
	The sequential rounds of the algorithm $t_i = 1, \dots, t$, 
	naturally express 
	communication--computation overlap. Using double 
	buffering, at each round 
	$t_i$ 
	we issue an asynchronous communication (using either 
	MPI\_Get or 
	MPI\_Isend / MPI\_Irecv~\cref{sec:rdma}) of the data required at round 
	$t_{i+1}$, 
	while locally processing the data received in a previous round. We note 
	that, 
	by the construction of the local domains 
	$\mathcal{D}_j$~\cref{sec:parScheduling}, the extra memory required for 
	double 
	buffering is rarely an issue. If we are constrained by the available 
	memory, 
	then the space required to hold the partial results of $C$, which is $a^2$, 
	is 
	much larger than the size of the receive buffers $s =(S - 
	a^2)/2$. If not, then there is extra memory available for the buffering.
	
	\macb{Number of rounds:} The minimum number of rounds, and therefore  
	latency, is $t= \frac{2ab}{S-a^2}$
	~(\cref{sec:parScheduling})
	. However, to 
	exploit more overlap, we can increase 
	the number of rounds $t_2 > t$. In this way, in one round we communicate 
	less 
	data $s_2 = ab/t_2 < s$, allowing the first round of computation to start 
	earlier.
	
	\subsection{One-Sided vs Two-Sided Communication}
	\label{sec:rdma}
	To reduce the latency~\cite{fompi} we implemented communication 
	using MPI RMA~\cite{mpi3-rma-overview}. This interface utilizes the 
	underlying features of Remote Direct Memory Access (RDMA) mechanism, 
	bypassing 
	the OS on the sender side and providing zero-copy 
	communication: data 
	sent is 
	not buffered in a temporary address, instead, it is written directly to its 
	location.
	
	All communication windows are 
	pre-allocated using \linebreak MPI\_Win\_allocate with the size of maximum 
	message in the 
	broadcast tree $2^{s-1} D$~(\cref{sec:commPattern}). Communication in each 
	step is performed using the MPI\_Get and MPI\_Accumulate 
	routines.
	
	For compatibility reasons, as well as for the performance comparison, we 
	also 
	implemented a communication back-end using MPI two-sided (the message 
	passing 
	abstraction). 
	
	\subsection{Communication Buffer Optimization}
	\label{sec:bufferReuse}
	The binary broadcast tree pattern is a generalization 
	of 
	the recursive structure of CARMA. However, CARMA in each recursive step 
	 dynamically allocates new buffers of the increasing size to match 
	the 
	message sizes $2^{s-1} D$, causing an additional runtime overhead.
	
	To alleviate this problem, we pre-allocate initial, send, 
	and receive buffers 
	for each of matrices A, B, and C of the maximum size of 
	the message 
	$ab/t$, 
	where $t = \frac{2ab}{S - a^2}$ is the number of steps in COSMA 
	(Algorithm~\ref{alg:COSMA}). Then, in each level $s$ of the communication 
	tree, 
	we  move the pointer in the receive buffer by $2^{s-1} D$ elements.
	
	\subsection{Blocked Data Layout}
	\label{sec:datalayout}
	
	COSMA's schedule induces the optimal initial data 
	layout, since for each $\mathcal{D}_j$ it determines its dominator set 
	$Dom(\mathcal{D}_j)$, that is, 
	elements accessed by processor $j$.
	Denote 
	$A_{l,j}$ and $B_{l,j}$  subsets of elements of matrices $A$ and $B$ that 
	initially reside in the local memory of processor $j$.
	The optimal data layout therefore requires that $A_{l,j}, B_{l,j} \subset 
	Dom(\mathcal{D}_j)$.
	However, the schedule does not specify exactly which elements of 
	$Dom(\mathcal{D}_j)$ 
	should be in $A_{l,j}$ and $B_{l,j}$.
	As a consequence of the communication pattern~\cref{sec:commPattern}, each 
	element of $A_{l,j}$ and $B_{l,j}$ is communicated to $g_m$, $g_n$ 
	processors, 
	respectively.
	To prevent data reshuffling, we therefore split each of 
	$Dom(\mathcal{D}_j)$
	into $g_m$ and $g_n$  smaller blocks, enforcing 
	that consecutive blocks are assigned to processors that communicate first. 
	This is unlike the 
	distributed CARMA implementation~\cite{CARMA}, which uses the cyclic 
	distribution among processors in the recursion base case and requires local 
	data 
	reshuffling after each communication round. Another advantage of our 
	blocked 
	data layout is a full compatibility with the block-cyclic one, which is 
	used in other linear-algebra libraries.

\section{Evaluation}
\label{sec:evaluation}

We evaluate COSMA's communication volume and performance against other 
state-of-the-art implementations with 
various 
combinations of matrix dimensions and memory requirements. These scenarios 
include both synthetic square matrices, in which all algorithms achieve 
their 
peak performance, 
{as well as ``flat'' (two large dimensions) and real-world 
	``tall-and-skinny'' (one large dimension)
	cases with 
	uneven number of processors}.

\vspace{0.5em}
\noindent
\macb{Comparison Targets} 

\noindent	
As a comparison, we use the widely used ScaLAPACK 
library
as provided by Intel MKL (version: 18.0.2.199)\footnote{the latest version 
	available on Piz Daint when benchmarks were 
	performed (August 2018). No improvements of P[S,D,C,Z]GEMM have been 
	reported in the MKL release notes since then.}, 
as well 
as Cyclops Tensor 
Framework\footnote{{https://github.com/cyclops-community/ctf, 
		commit ID  244561c on May 15, 2018}},
and 
the original CARMA 
implementation\footnote{{https://github.com/lipshitz/CAPS, commit ID 
		7589212 
		on July 19, 2013}}.
\emph{We manually tune ScaLAPACK parameters to achieve its maximum 
	performance.} Our experiments showed that on Piz Daint it achieves the 
highest performance when run with 4 MPI ranks per compute node, 9 cores 
per 
rank. Therefore, for each matrix sizes/node count configuration, we 
recompute the optimal rank decomposition for ScaLAPACK. Remaining 
implementations use default decomposition strategy and perform best 
utilizing 
36 ranks per 
node, 1 core per rank.

\vspace{0.5em}
\noindent
\macb{Infrastructure and Implementation Details}

\noindent
All implementations were compiled using the GCC 6.2.0 compiler. We use 
Cray-MPICH 3.1 implementation of MPI. 
{The parallelism within a rank of ScaLAPACK}\footnote{only ScaLAPACK 
	uses 
	multiple cores per ranks} is handled
internally by the MKL BLAS  (with GNU OpenMP threading) version 2017.4.196. 
To profile MPI communication 
volume, we use the mpiP version 3.4.1~\cite{mpip}.

\vspace{0.5em}
\noindent
\macb{Experimental Setup and Architectures}

\noindent
We run our experiments on the CPU partition of CSCS Piz Daint, which 
has 
1,813 XC40 nodes with
dual-socket Intel Xeon E5-2695 v4 processors ($2\cdot18$ 
cores, 3.30 GHz, 
45 MiB 
L3 shared 
cache, 
64 GiB DDR3 RAM),
interconnected by the Cray Aries network with a dragonfly network topology. 
{We set $p$ to a number of available cores}\footnote{for ScaLAPACK, 
	actual 
	number of MPI ranks is $p/9$} 
{and 
	$S$ to the main memory size per core}~(\cref{sec:machineModel}). 
{To 
	additionally capture cache size per core, the model can be extended to 
	a three-level memory hierarchy. However, cache-size tiling is already 
	handled internally by the MKL.}

\vspace{0.5em}
\noindent
\macb{Matrix Dimensions and Number of Cores}

\noindent
{We use square ($m = n = k$), ``largeK'' ($m = n \ll k$), ``largeM'' 
	($m 
	\gg n = k$), and ``flat'' ($m = n \gg k$) matrices.} The matrix dimensions 
and 
number 
of 
cores are (1) powers of 
two $m = 2^{r_1}, n = 2^{r_2}, m = 2^{r_3}$, (2) determined by the 
real-life 
simulations or hardware architecture (available nodes on a computer), (3) 
chosen adversarially, e.g, $n^3 + 1$. 
Tall matrix dimensions are taken from an application benchmark, namely 
the 
calculation of the random phase approximation (RPA) energy of water 
molecules~\cite{joost}. There, to simulate $w$ molecules, the sizes 
of the 
matrices are $m=n=136w$ and $k = 228w^2$. In the strong scaling scenario, 
we 
use 
$w=128$ as in the 
original paper, yielding $m=n=$ 17,408, $k = $ 3,735,552. For performance 
runs, we scale up to 512 nodes (18,432 cores).

\vspace{0.5em}
\noindent
\macb{Selection of Benchmarks}

\noindent
We perform both strong scaling and \emph{memory scaling} experiments. 
The memory scaling scenario fixes the input size per core 
($\frac{pS}{I}, I = mn + mk + nk$), as opposed to the work per 
core ($\frac{mnk}{p} \ne const$). We evaluate two cases: (1) "limited 
memory" ($\frac{pS}{I} = const$), and (2) "extra memory"  
($\frac{p^{2/3}S}{I} = 
const$).

{
	To provide more information about the impact of communication optimizations 
	on the total runtime, for each of the matrix shapes we also separately 
	measure time spent by COSMA on different parts of the code. for each matrix 
	shape we present two extreme cases of 
	strong scaling - with smallest number of processors (most compute-intense) 
	and with the largest (most communication-intense). To additionally increase 
	information provided, we perform these measurements with and without 
	communication--computation overlap. }

\vspace{0.5em}
\noindent
\macb{Programming Models}

\noindent
We use either the RMA or the Message Passing 
models. CTF also uses both models, whereas CARMA and ScaLAPACK use MPI 
two-sided (Message Passing).

\vspace{0.5em}
\noindent
\macb{Experimentation Methodology}

\noindent
For each combination of parameters, we perform 5 runs, each with different 
node allocation. 
As all the algorithms use BLAS routines for local matrix computations, for 
each run we execute the kernels three times and take the minimum to 
compensate for the BLAS setup overhead. 
We report median and 95\% confidence intervals of the runtimes.

\vspace{-1em}
	\section{Results}
	\label{sec:results}
		\begin{figure*}
		\vspace{-2.0em}
		\centering
		\subfloat[Strong scaling,  $n=m=k= $ 16,384
		]{\includegraphics[width=\fw \textwidth]
			{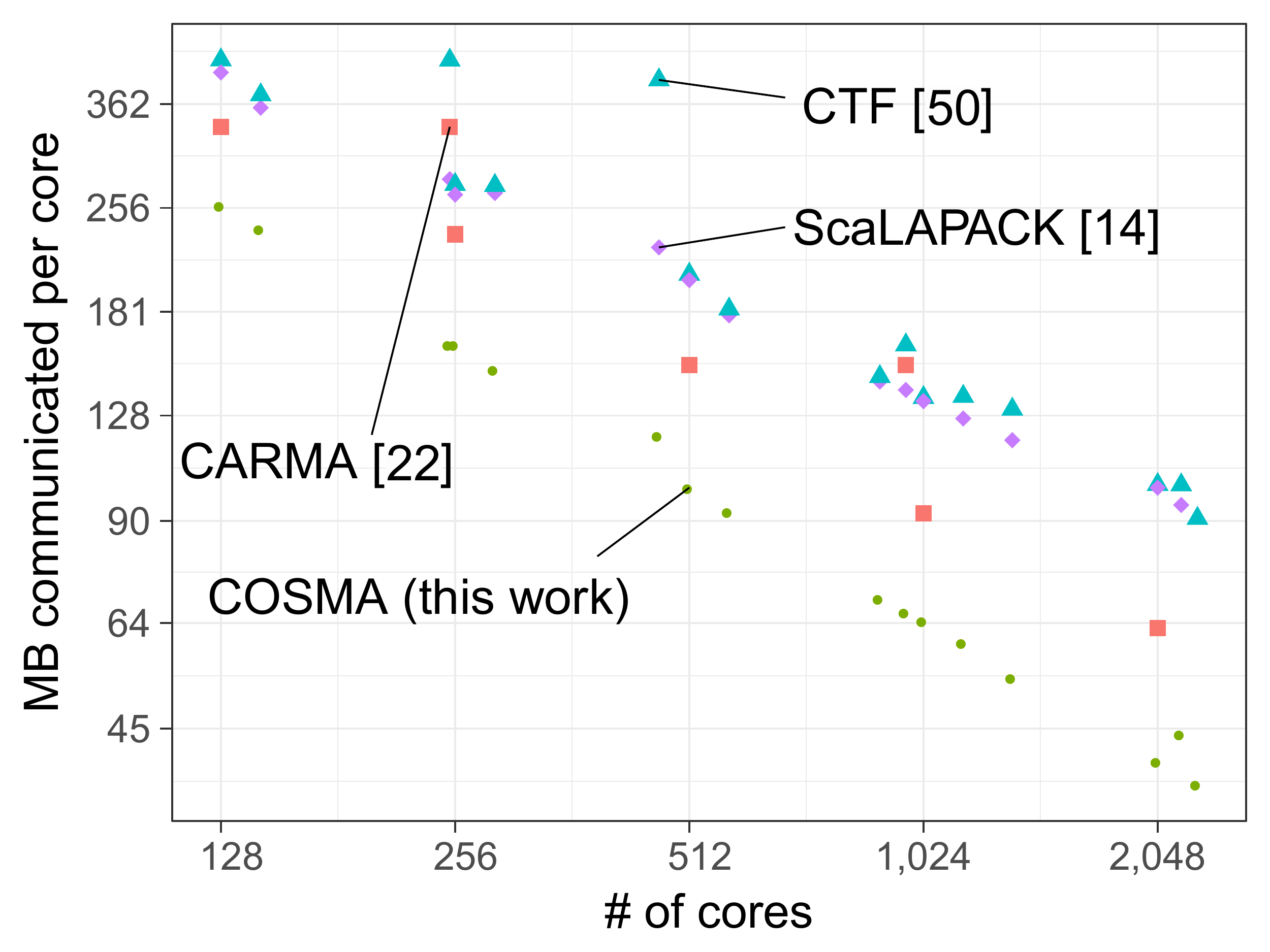}\label{fig:square_strong_comm}}
		\hfill	
		\subfloat[Limited memory, $n = m = k = 
		\sqrt{\frac{pS}{3}}$]{\includegraphics[width=\fw \textwidth]	
			{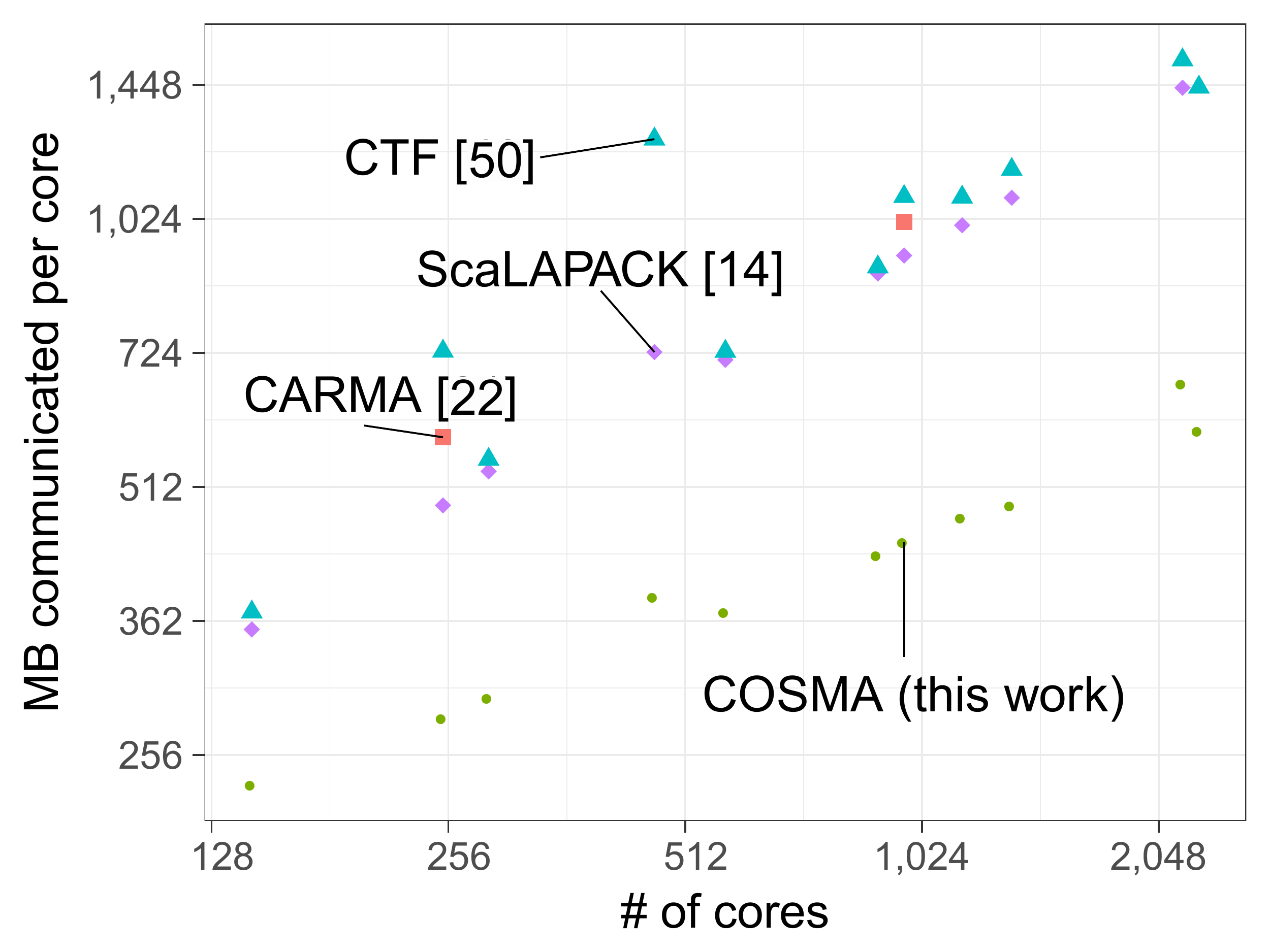}\label{fig:square_strong_comm}}
		%
		\hfill
		\subfloat[Extra memory,$n = m = k = 
		\sqrt{\frac{p^{2/3}S}{3}}$]{\includegraphics[width=\fw \textwidth]	
			{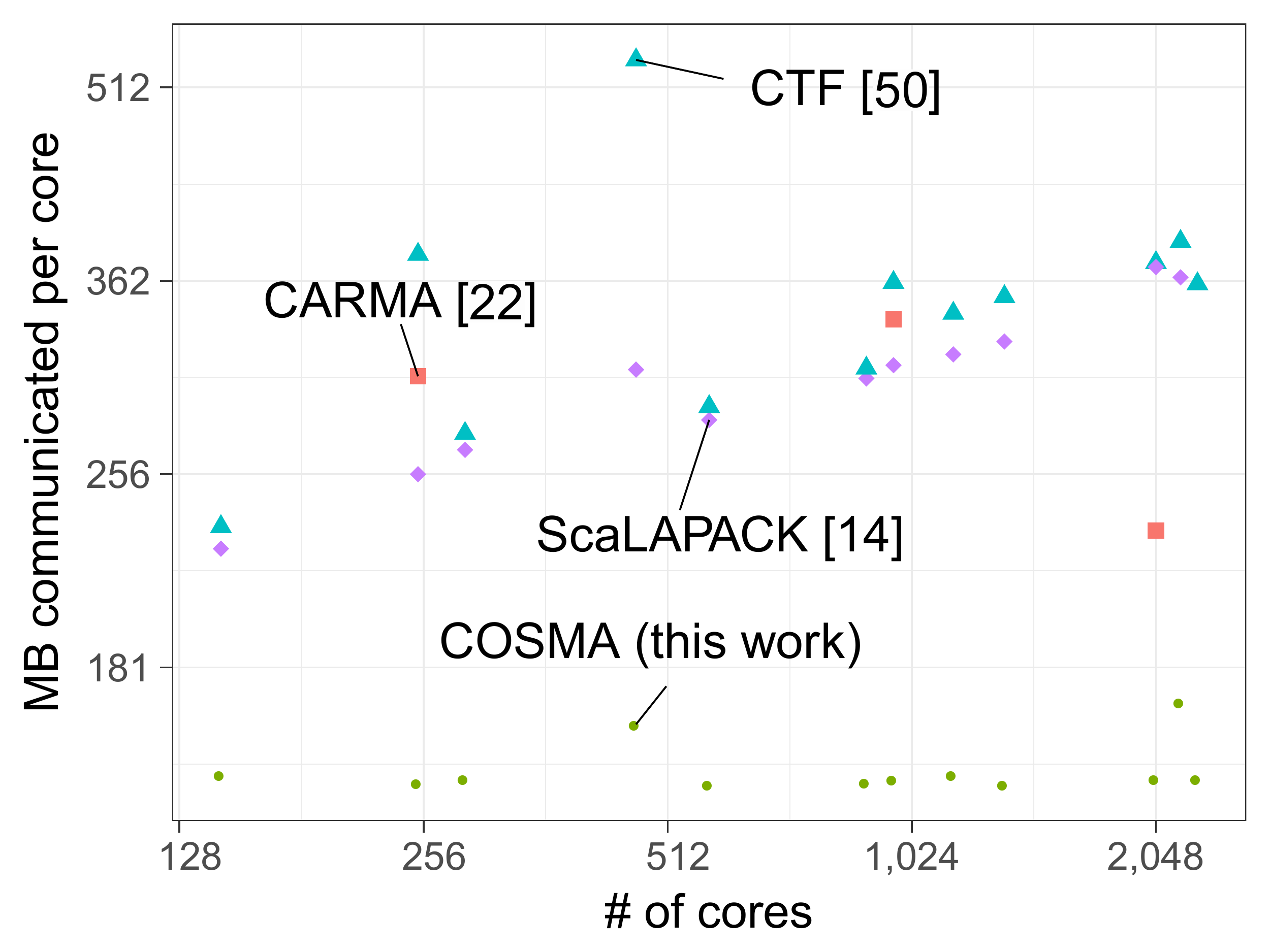}\label{fig:square_strong_comm}}
		\vspace{-1em}
		\caption{
			\textmd{{Total communication volume per core carried out by 
					COSMA, CTF, 
					ScaLAPACK 
					and 
					CARMA for square matrices, as measured by the mpiP 
					profiler.}} 
		}
		\vspace{-1.5em}
		\label{fig:performancePlotsCommSquare}	
	\end{figure*}
	\begin{figure*}
		\centering
		\subfloat[Strong scaling, $n=m= $17,408, $k= $ 
		3,735,552]{\includegraphics[width=\fw \textwidth]
		{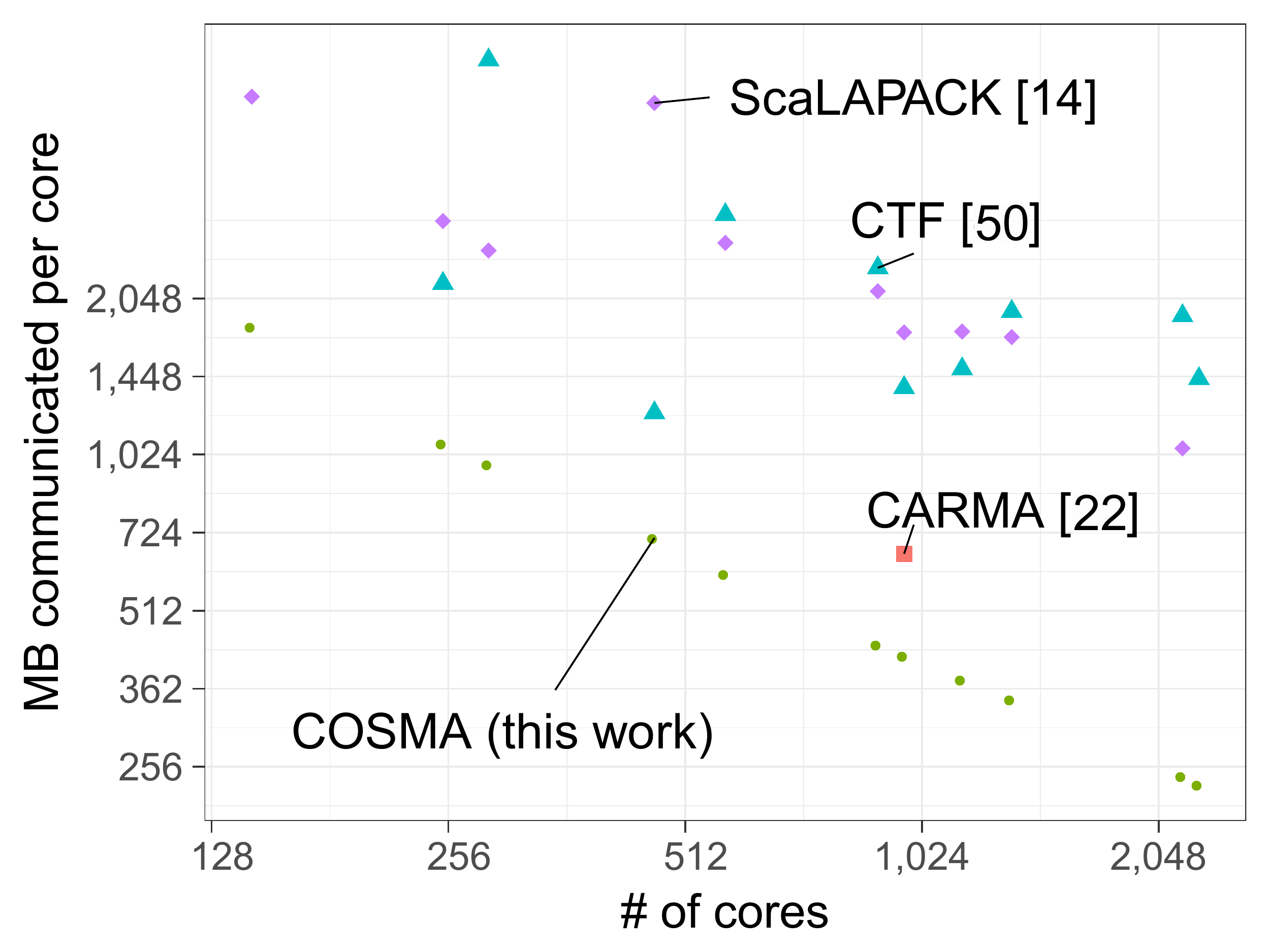}\label{fig:square_strong_comm}}
		\hfill	
		\subfloat[Limited memory,$m=n=979 p^\frac{1}{3}$, 
		$k=$1.184$p^\frac{2}{3}$]{\includegraphics[width=\fw \textwidth]	
	{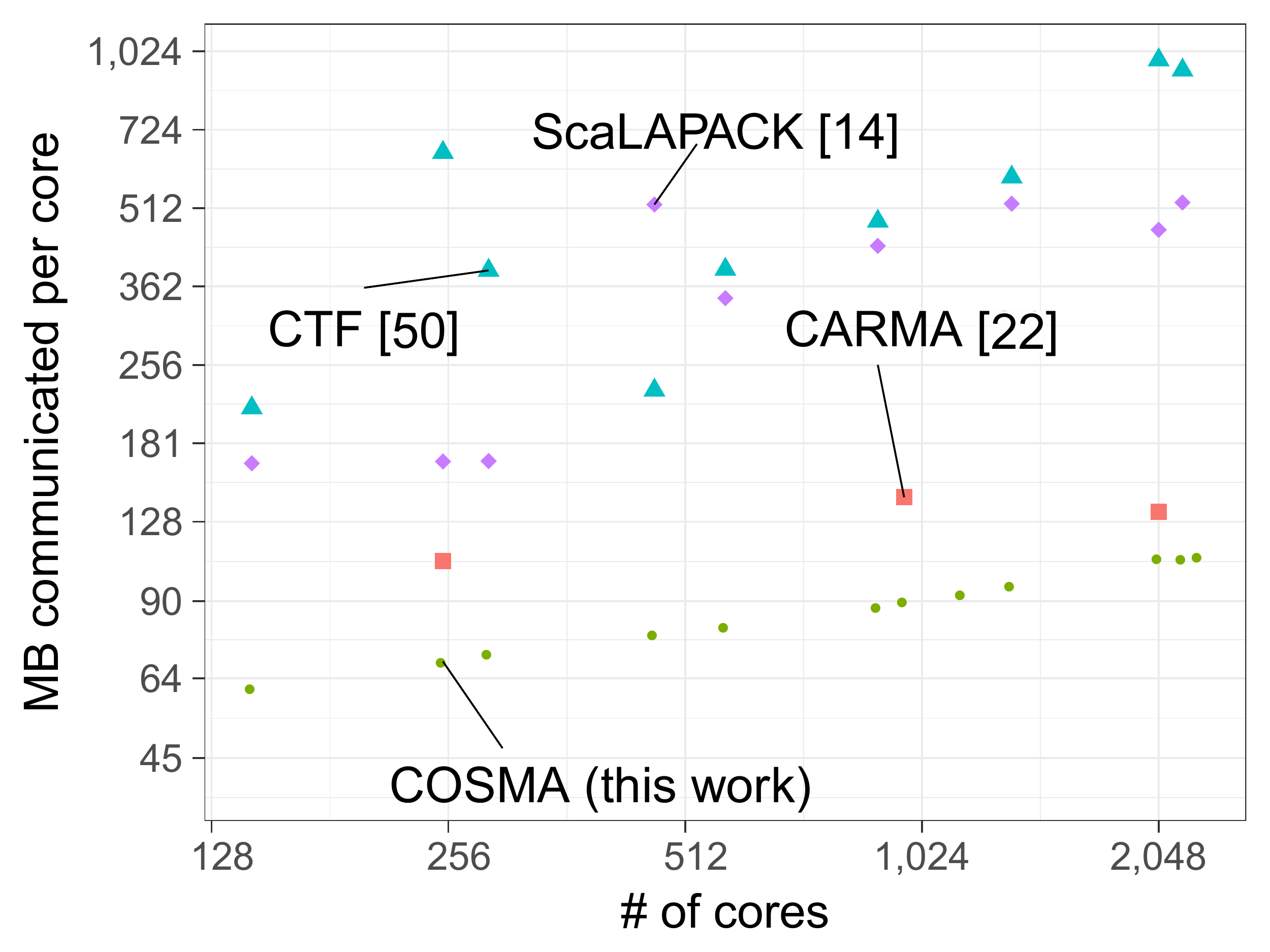}\label{fig:square_strong_comm}}
		%
		\hfill
		\subfloat[Extra memory,,$m=n=979 p^\frac{2}{9}$, 
		$k=$1.184$p^\frac{4}{9}$]{\includegraphics[width=\fw \textwidth]
	{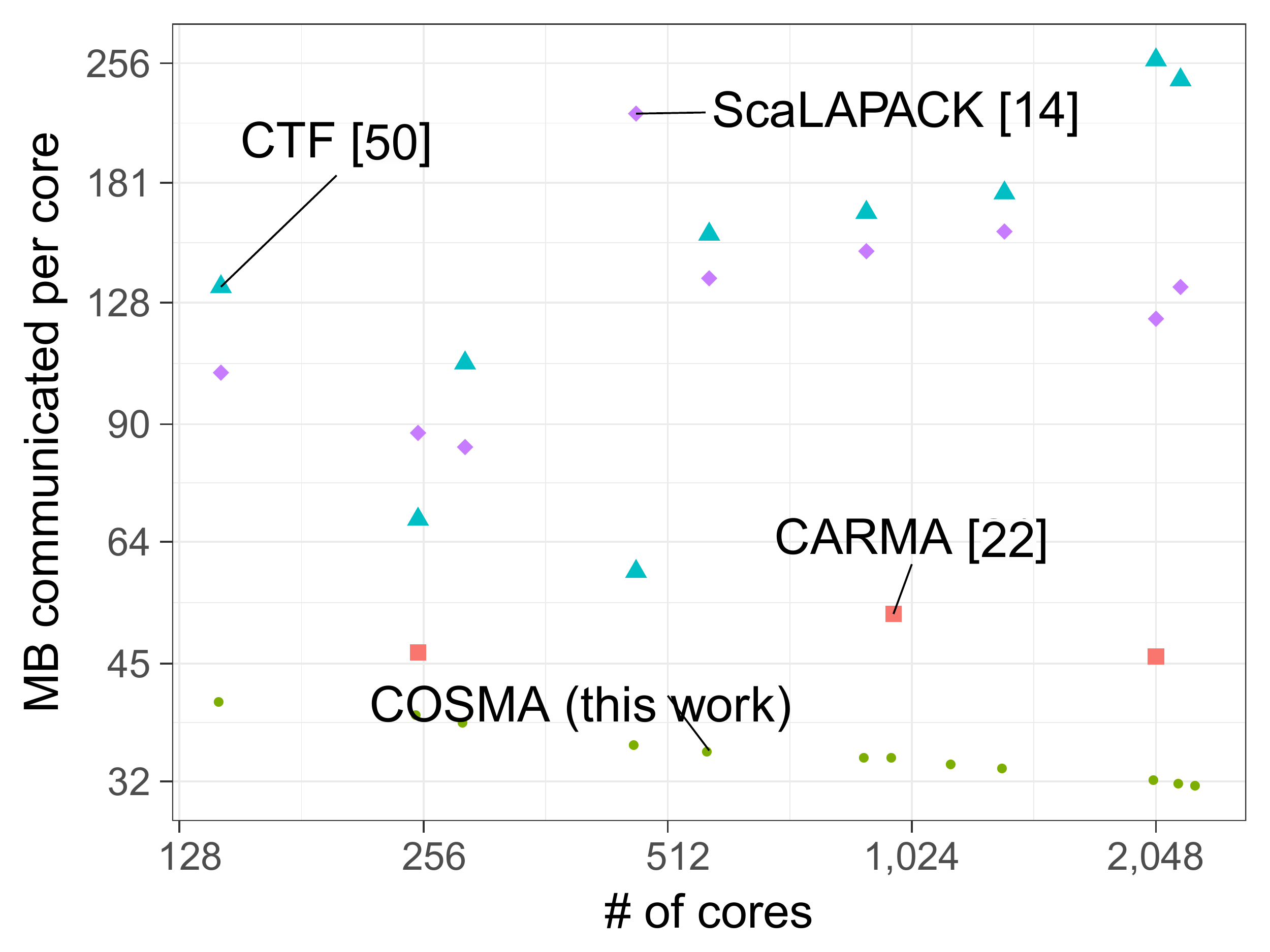}\label{fig:square_strong_comm}}
		\vspace{-1em}
		\caption{
			\textmd{{Total communication volume per core carried out by 
					COSMA, CTF, 
					ScaLAPACK 
					and 
					CARMA for ``largeK'' matrices, as measured by the mpiP 
					profiler.}} 
		}
		\vspace{-1.5em}
		\label{fig:performancePlotsCommLargeK}
	\end{figure*}
We now present the experimental results comparing COSMA with the existing 
algorithms. For both strong and memory scaling, we 
measure total communication volume and
runtime on both square and tall 
matrices. Our experiments show that 
COSMA 
always communicates least data and 
is  
the fastest in \emph{all} scenarios. 

\vspace{0.5em}
\noindent
\macb{Summary and Overall Speedups}

\noindent
As discussed in~\cref{sec:evaluation}, we evaluate three benchmarks -- 
strong scaling, 
``limited memory'' (no redundant copies of the input are possible), and 
``extra memory'' ($p^{1/3}$ extra copies of the input can fit into combined 
memory of all cores). Each 
of them we test for square, ``largeK'', ``largeM'', and , ``flat'' 
matrices, giving twelve 
cases in 
total.
In Table~\ref{tab:results}, we present arithmetic mean of total 
communication volume per MPI rank across all core counts. We also report 
the summary of minimum, geometric mean, and maximum speedups vs the 
second best-performing algorithm.

\vspace{0.5em}
\noindent
\macb{Communication Volume}

\noindent
As analyzed in~\cref{sec:seqOpt} and ~\cref{sec:parOptimality}, COSMA 
reaches I/O lower 
bound (up to the factor of 
$\sqrt{S}/(\sqrt{S+1}-1)$). Moreover, optimizations presented 
in~\cref{sec:implementation} secure further improvements 
compared to other state-of-the-art algorithms. In all cases, COSMA performs 
least communication. 
Total communication volume  for square and ``largeK'' 
scenarios is shown in 
Figures~\ref{fig:performancePlotsCommSquare} 
and~\ref{fig:performancePlotsLargeK}.

\vspace{0.5em}
\noindent
\macb{Square Matrices}

\noindent
Figure~\ref{fig:performancePlotsSquare} presents the \% of achieved peak 
hardware 
performance for square matrices in all three scenarios. As COSMA is based 
on the near optimal schedule, 
it  achieves the highest 
performance \emph{in all cases}. Moreover, its performance pattern is the 
most stable: when the number of cores is 
not a 
power of two, the performance does not vary much compared to all remaining 
three implementations. 
{
	We note that matrix dimensions in the strong scaling scenarios 
	($m=n=k=2^{14}$) are very small for distributed setting. Yet even in this 
	case COSMA maintains relatively high performance for large numbers of 
	cores: 
	using 4k cores it achieves 35\% of peak performance, compared to 
	<5\% of CTF and ScaLAPACK,
	showing excellent strong scaling characteristics.}

\vspace{0.5em}
\noindent
\macb{Tall and Skinny Matrices}

\noindent
Figure~\ref{fig:performancePlotsLargeK} 
presents the 
results for ``largeK'' matrices - due to space constraints, the symmetric 
``largeM'' 
case is 	
For strong scaling, the minimum number 
of cores 
is 2048 (otherwise, the matrices of size $m=n=$17,408, $k=$3,735,552 do 
not fit into memory).
Again, COSMA shows 
the most stable performance with a varying number of cores.

\vspace{0.5em}
\noindent
\macb{``Flat'' Matrices}

\noindent
Matrix dimensions for strong scaling are set to $m = n = 2^{17} = 
$131,072 and $k = 2^9 = $512. Our weak scaling scenario models the 
rank-k update kernel, with fixed $k = $256, and $m = n$ scaling 
accordingly for the ``limited'' and ``extra'' memory cases. Such kernels 
take most of the execution time in, e.g., matrix factorization 
algorithms, where updating Schur complements is performed as a rank-k gemm 
operation~\cite{MPLU}.

\vspace{0.5em}
\noindent
\macb{Unfavorable Number of Processors}

\noindent
Due to the processor grid 
optimization~(\cref{sec:decompArbitrary}), 
the performance is stable and does not suffer from unfavorable 
combinations 
of 
parameters. E.g., the runtime of COSMA for square matrices 
$m=n=k=$16,384 on $p_1=$9,216$=2^{10}\cdot3^2$ cores is 142~ms. 
Adding an 
extra core ($p_2=$9,217$=13 \cdot 709$), does not change COSMA's runtime, 
as the optimal 
decomposition does not utilize it. On 
the other hand, CTF for $p_1$ runs in 600~ms, while for 
$p_2$ the 
runtime \emph{increases} to 1613~ms due to a non-optimal 
processor decomposition.

\vspace{0.5em}
\noindent
\macb{Communication-Computation Breakdown}

\noindent
In
Figure~\ref{fig:performancePlotsBreakdown} we present the total 
runtime 
breakdown of COSMA into communication and computation routines. Combined 

with the comparison of communication volumes 
(Figures~\ref{fig:performancePlotsCommSquare} 
and~\ref{fig:performancePlotsCommLargeK}, Table~\ref{tab:results}) we see 
the importance of our I/O optimizations for distributed setting even for 
traditionally compute-bound MMM. E.g., for square or ``flat'' matrix and 
16k cores, COSMA communicates more than two times less than the 
second-best (CARMA). Assuming constant time-per-MB,  COSMA would be 
40\% slower if it communicated that much, being slower than CARMA by 
30\%. For ``largeK'', the 
situation is even more extreme, with COSMA suffering 2.3 times 
slowdown if communicating as much as the second-best 
algorithm, CTF, which communicates 10 times more.

\vspace{0.5em}
\noindent
\macb{Detailed Statistical Analysis}

\noindent
Figure~\ref{fig:performancePlotsDistr} provides a distribution of the 
achieved 
peak performance across all numbers of cores for all six scenarios. It can 
be seen that, for example, in the strong scaling scenario and square 
matrices, COSMA is comparable to the other implementations (especially 
CARMA). However, for tall-and-skinny matrices with limited memory 
available, \emph{COSMA lowest achieved performance is higher than the best 
	performance of CTF and ScaLAPACK.}
	
	\begin{table}[h]
\vspace{0.5em}
	\setlength{\tabcolsep}{2pt}
	\renewcommand{\arraystretch}{0.7}
	\centering
	\footnotesize
	\sf
	\vspace{+0.5em}
	\begin{tabular}{llrrrrrrr}
		\toprule
		& & 
		\multicolumn{4}{c}{total comm. volume per rank [MB]} &
		\multicolumn{3}{c}{speedup}  \\
		\cmidrule(lr){3-6} \cmidrule(lr){7-9}
		\textbf{shape} & \textbf{benchmark} 	
  & \scriptsize{\textbf{ScaLAPACK}}
  & \textbf{CTF}
    & \textbf{CARMA} 
  & \textbf{COSMA}
  & \textbf{min}  & 
  \textbf{mean} & \textbf{max}\\
		\midrule
		\multirow{3}{*}{
			\includegraphics[width=0.037 \textwidth]
			{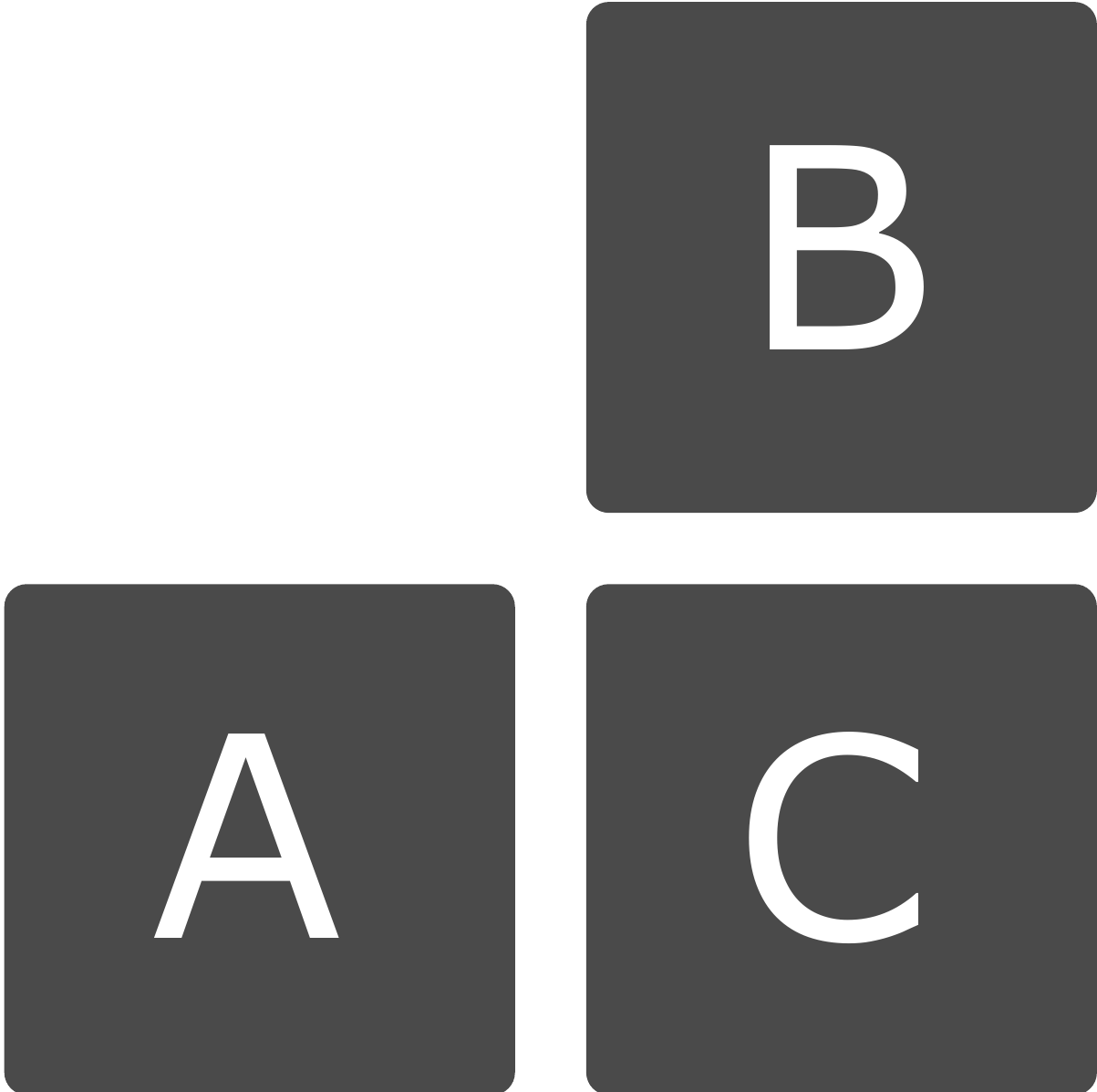}
		} & 
		strong scaling & 
		203 \hspace{0.14cm} & 222 & 195 \hspace{0.05cm} & 107 \hspace{0.1cm} &
		\textbf{1.07} & 1.94 & 4.81 \\
		& limited memory 
		& 816 \hspace{0.14cm} & 986 & 799 \hspace{0.05cm} & 424 \hspace{0.1cm}
		& 1.23 & 1.71 & 2.99 \\
		& extra memory 
		& 303 \hspace{0.14cm} & 350 & 291 \hspace{0.05cm} & 151 \hspace{0.1cm}
		&  1.14 & 2.03 & 4.73 \\
		\midrule
		\multirow{3}{*}{
		\includegraphics[width=0.037 \textwidth]
		{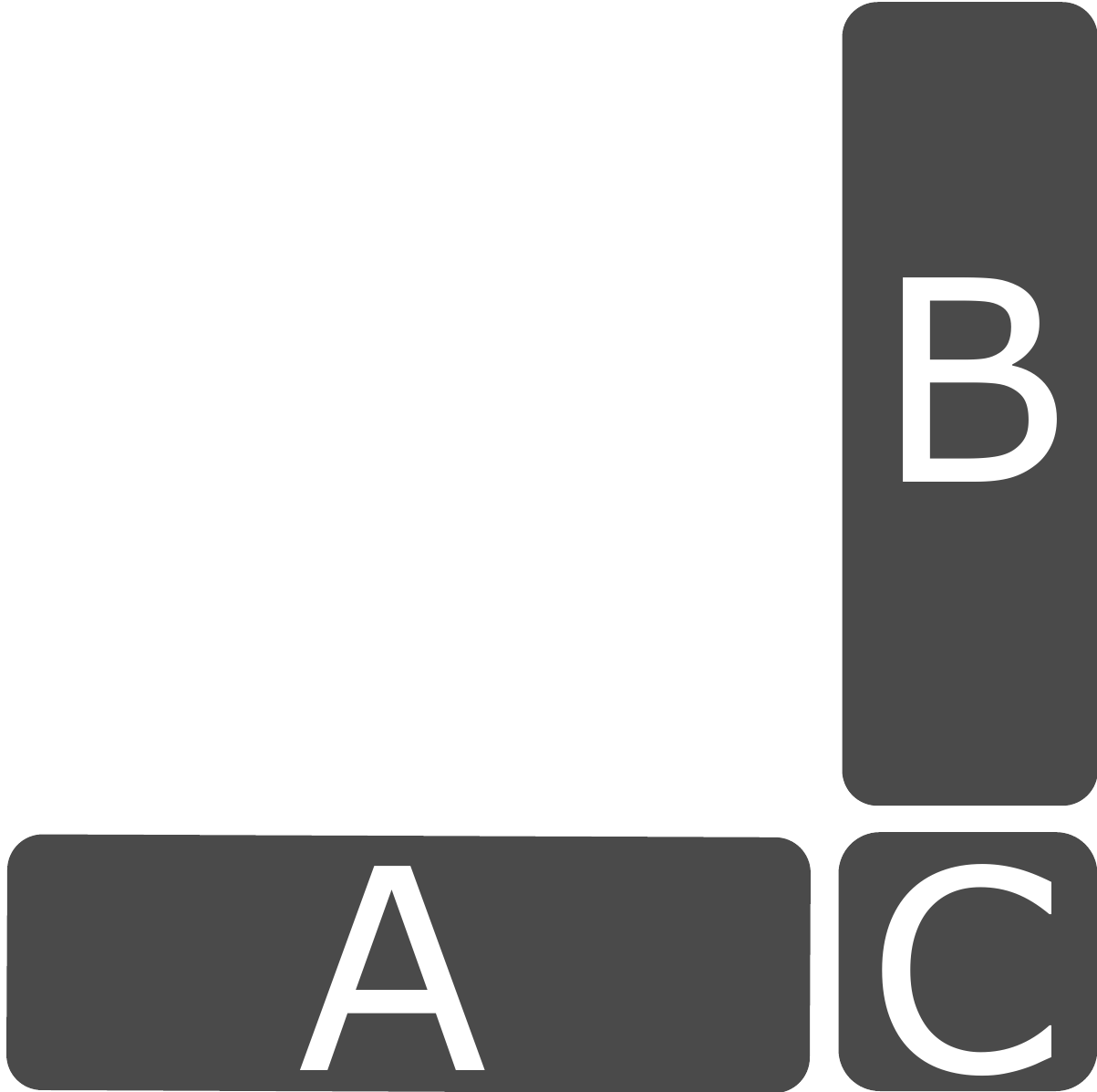}
	} & 
		strong scaling 
	& 2636 \hspace{0.14cm} & 2278 & 659 \hspace{0.05cm} & 545 \hspace{0.1cm}
		& 1.24 & 2.00 & 6.55 \\
		& limited memory 
		& 368 \hspace{0.14cm} & 541 & 128 \hspace{0.05cm} & 88 \hspace{0.1cm}
		& 1.30 & 2.61 & 8.26 \\
		& extra memory 
		& 133 \hspace{0.14cm} & 152 & 48 \hspace{0.05cm} & 35 \hspace{0.1cm}
		& 1.31 & 2.55 & 6.70 \\
			\midrule
		\multirow{3}{*}{
			\includegraphics[width=0.037 \textwidth]
			{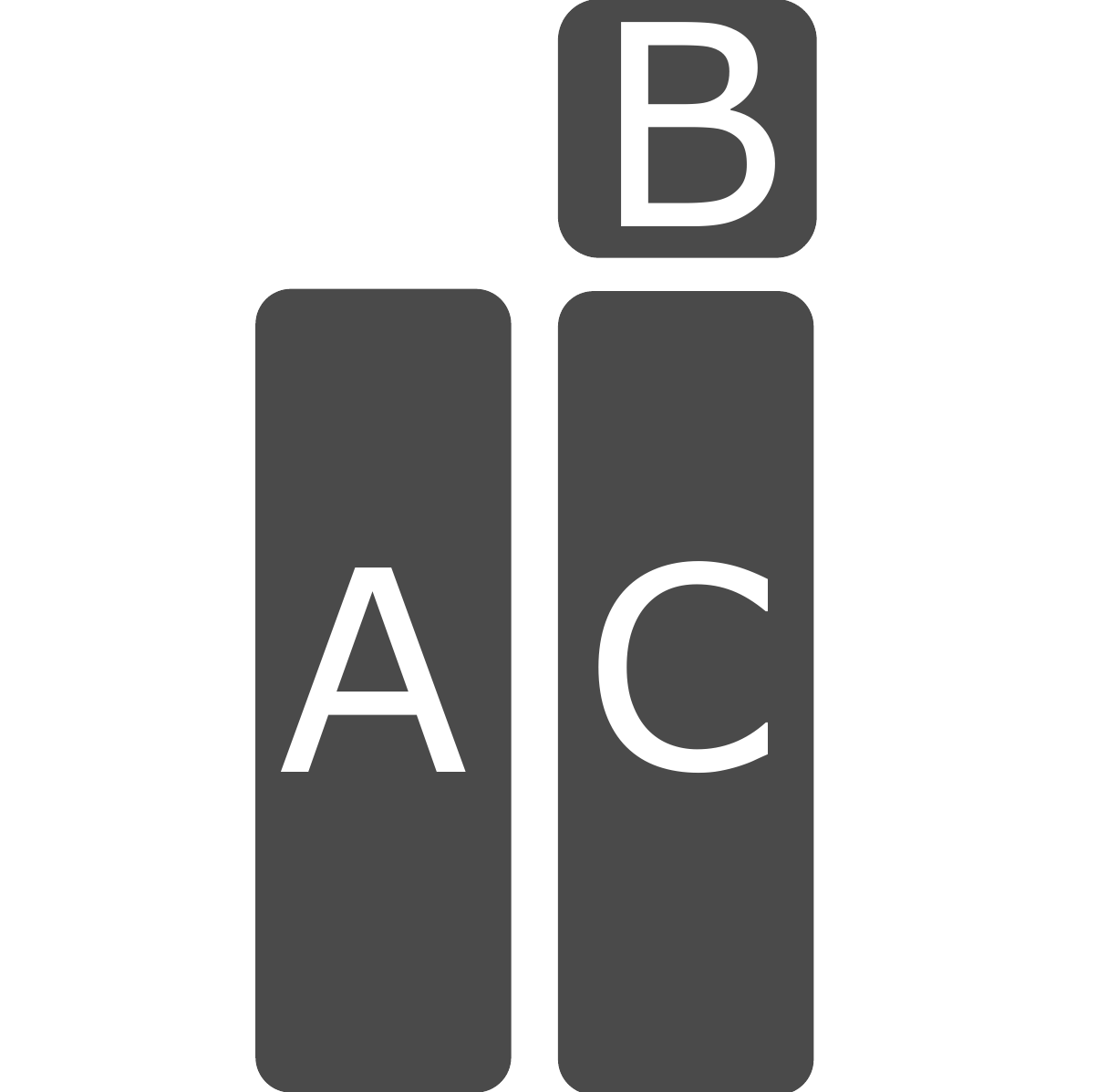}
		} & 
		strong scaling 
		& 3507 \hspace{0.14cm} & 2024 & 541 \hspace{0.05cm} & 410 \hspace{0.1cm}
		& 1.31 & 2.22 & 3.22 \\
		& limited memory 
		& 989 \hspace{0.14cm} & 672 & 399 \hspace{0.05cm} & 194 \hspace{0.1cm}
		& 1.42 & 1.7 & 2.27 \\
		& extra memory 
		& 122 \hspace{0.14cm} & 77 & 77 \hspace{0.05cm} & 29 \hspace{0.1cm}
		& 1.35 & 1.76 & 2.8 \\
			\midrule
		\multirow{3}{*}{
			\includegraphics[width=0.037 \textwidth]
			{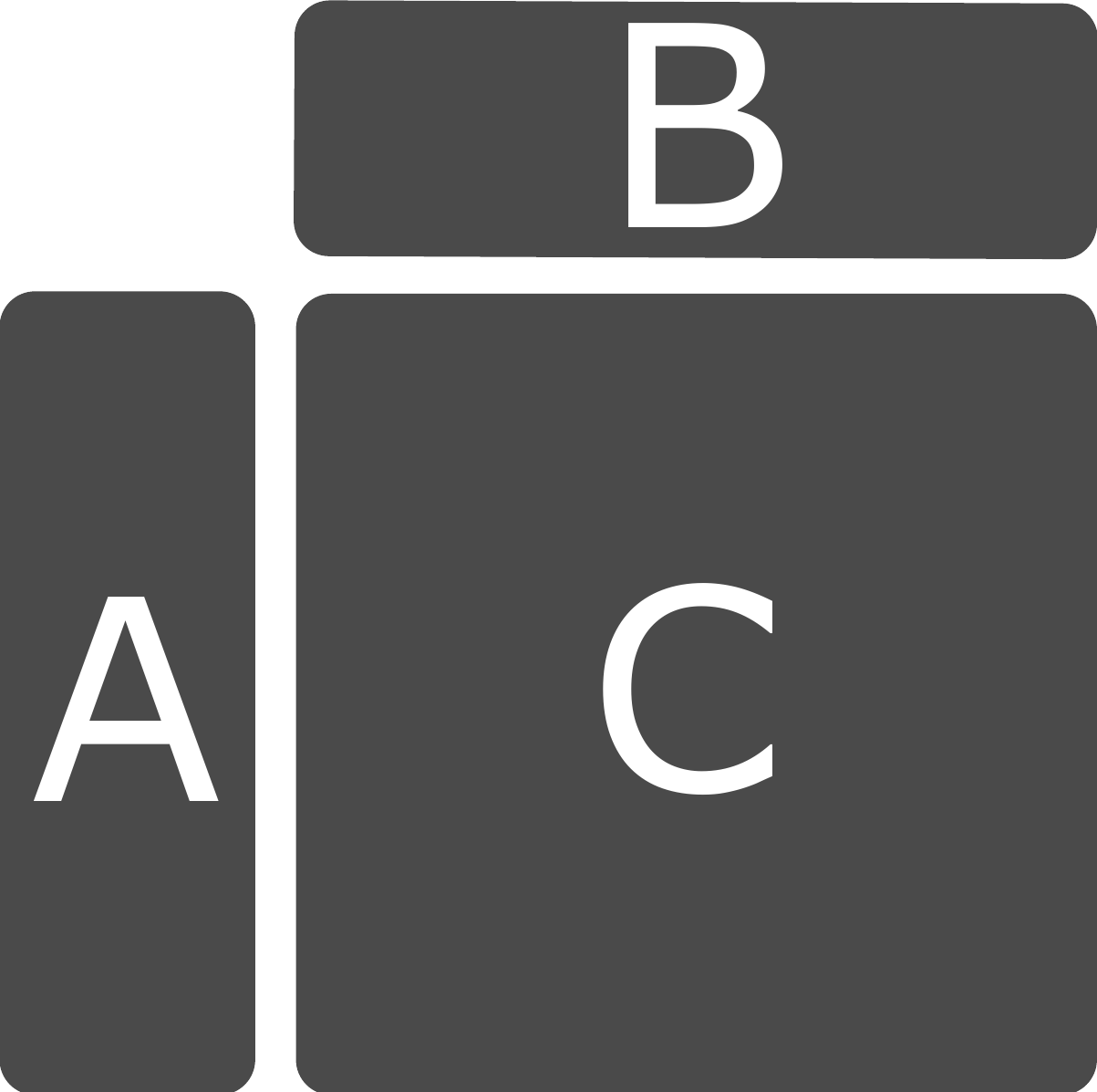}
		} & 
		strong scaling 
		& 134 \hspace{0.14cm} & 68 & 10 \hspace{0.05cm} & 7 \hspace{0.1cm}
		& 1.21 & 4.02 & \textbf{12.81} \\
		& limited memory 
		& 47 \hspace{0.14cm} & 101 & 26 \hspace{0.05cm} & 8 \hspace{0.1cm}
		& 1.31 & 2.07 & 3.41 \\
		& extra memory 
		& 15 \hspace{0.14cm} & 15 & 10 \hspace{0.05cm} & 3 \hspace{0.1cm}
		& 1.5 & 2.29 & 3.59 \\
		\midrule
		\multicolumn{2}{c}{\textbf{overall}} 			
		& 
	 &  &  & 
		& \textbf{1.07} & 
		\textbf{2.17} 
		& 
		\textbf{12.81} \\
		\bottomrule
	\end{tabular}
	\caption{
		\textmd{{Average communication volume per MPI rank and
			measured speedup of COSMA vs the 
			second-best algorithm
	 across all core 
		counts for each of the 
		scenarios.}}
	\vspace{-1.5em}
	}
\vspace{-1.5em}
	\label{tab:results}
	\end{table}

\begin{figure*}[t]
	\centering
	\subfloat[Strong scaling, $n=m=k= $ 16,384]
	{\includegraphics[width=\fw \textwidth]
		{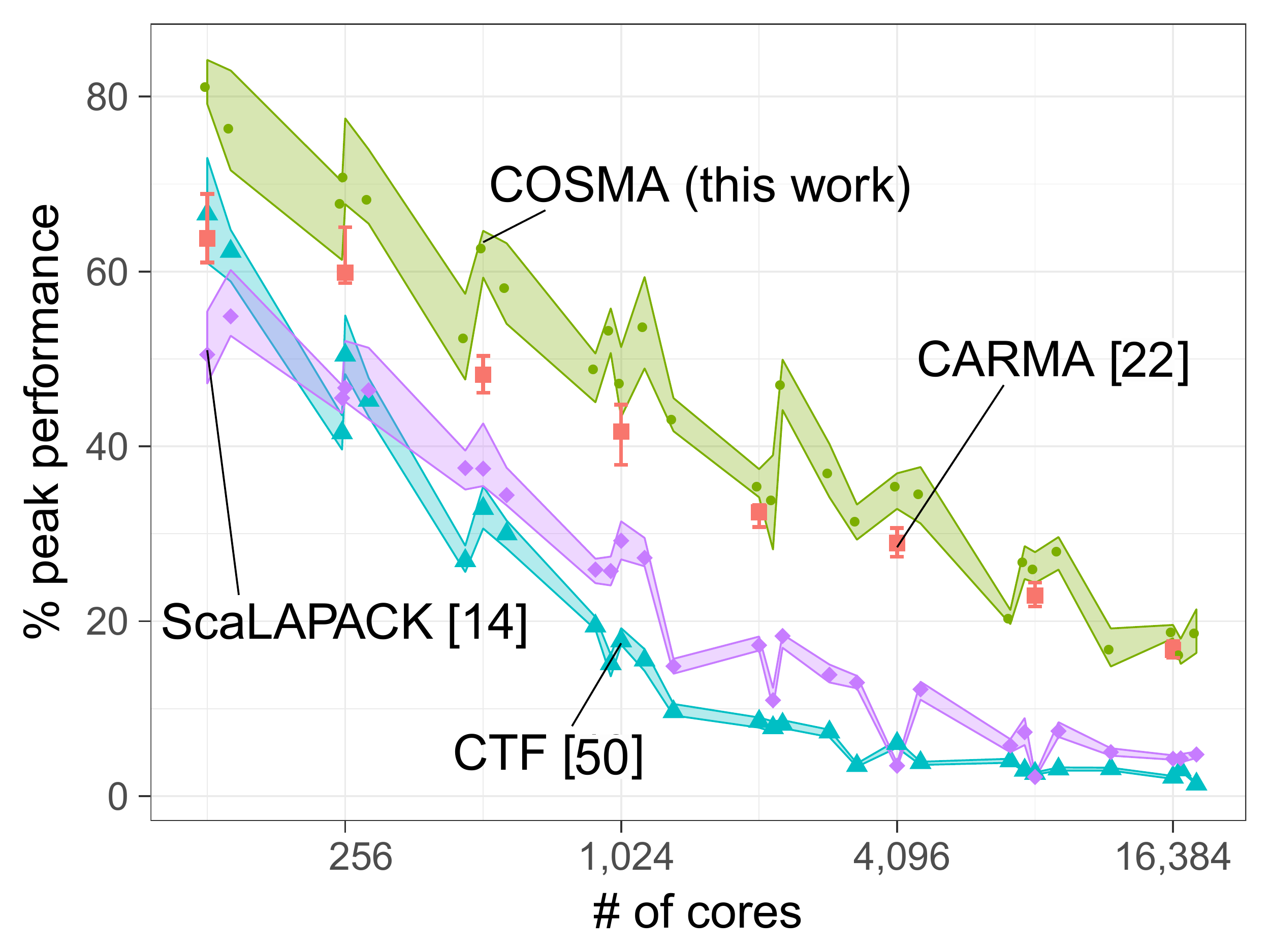}
		\label{fig:square_strong}}
	\hfill	
	\subfloat[Limited memory, $n = m = k = 
	\sqrt{\frac{pS}{3}}$]{\includegraphics[width=\fw \textwidth]	
		{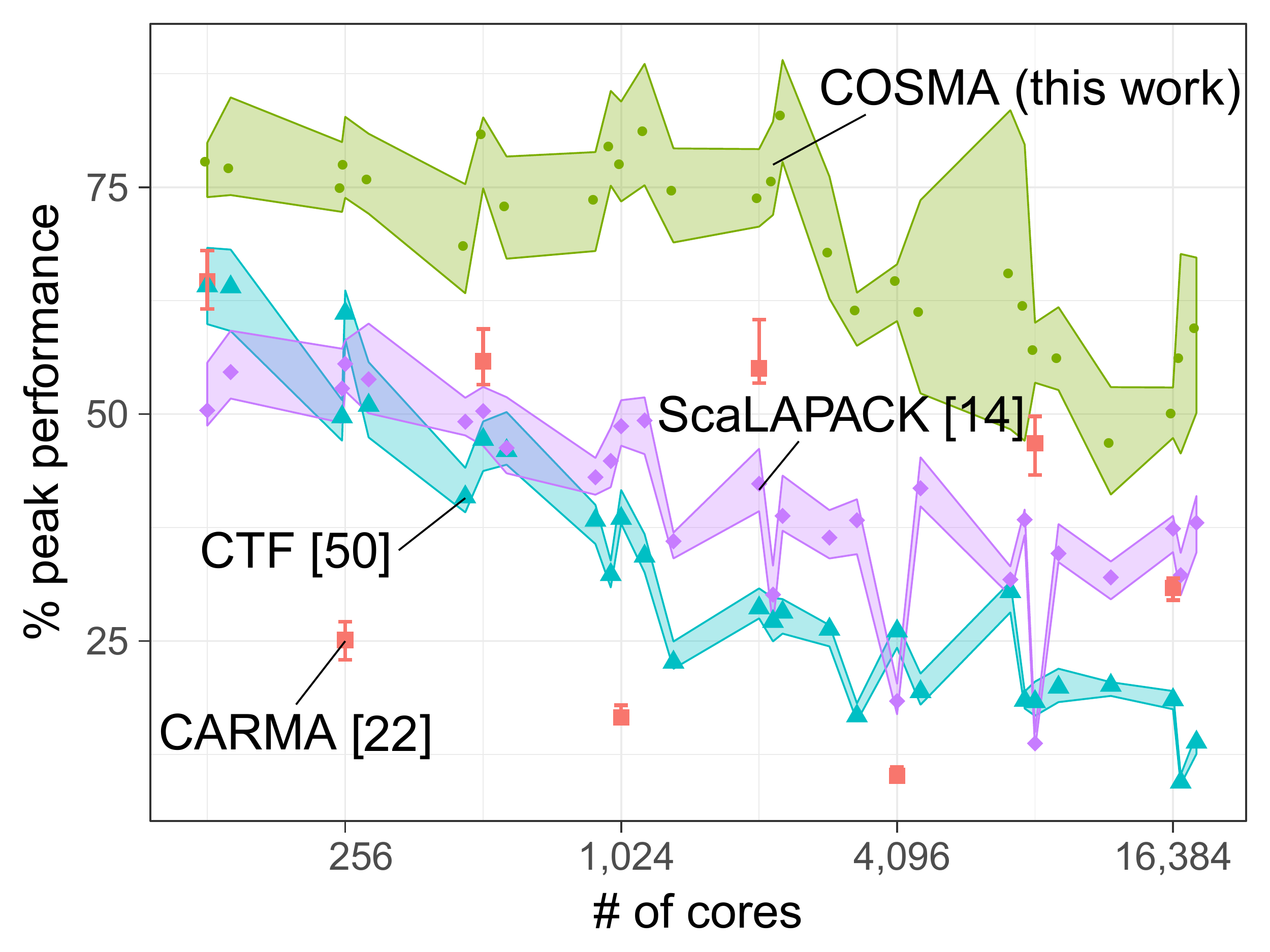}	
		\label{fig:square_weakp1}}
	%
	\hfill
	\subfloat[Extra memory, $m = n = k = 
	\sqrt{\frac{p^{2/3}S}{3}}$]{\includegraphics[width=\fw 
		\textwidth]			
		{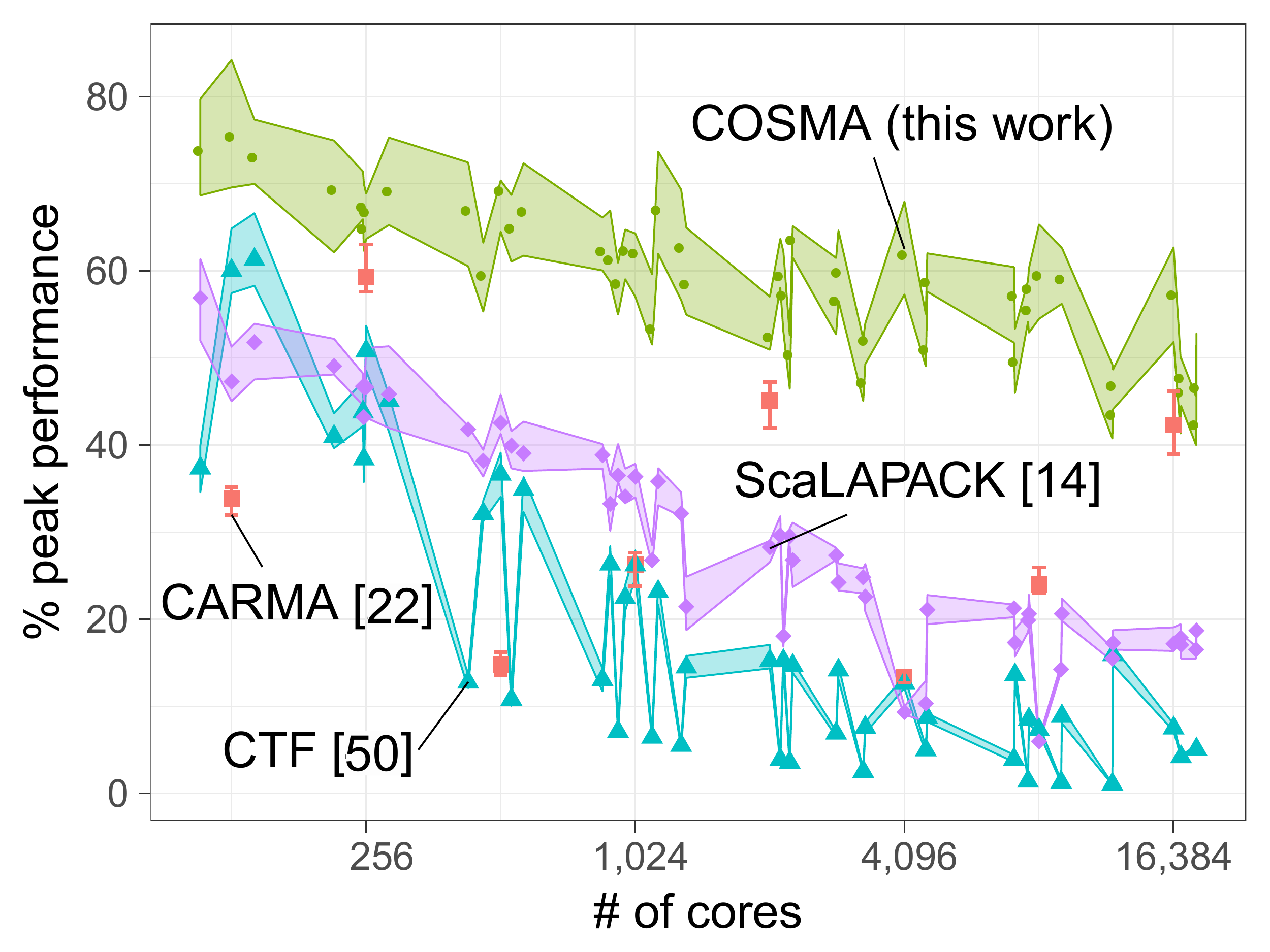}
		\label{fig:square_weakp2}}
	\vspace{-1em}
	\caption{
		\textmd{{Achieved \% of peak performance by COSMA, CTF, 
				ScaLAPACK 
				and 
				CARMA for square matrices, strong and weak scaling. We show 
				median 
				and 95\% confidence intervals.} }
	}
	\label{fig:performancePlotsSquare}
	
\end{figure*}

\begin{figure*}
	\centering
	\subfloat[Strong scaling, $n=m=k= $ 16,384 
	]{\includegraphics[width=\fw \textwidth]		
		{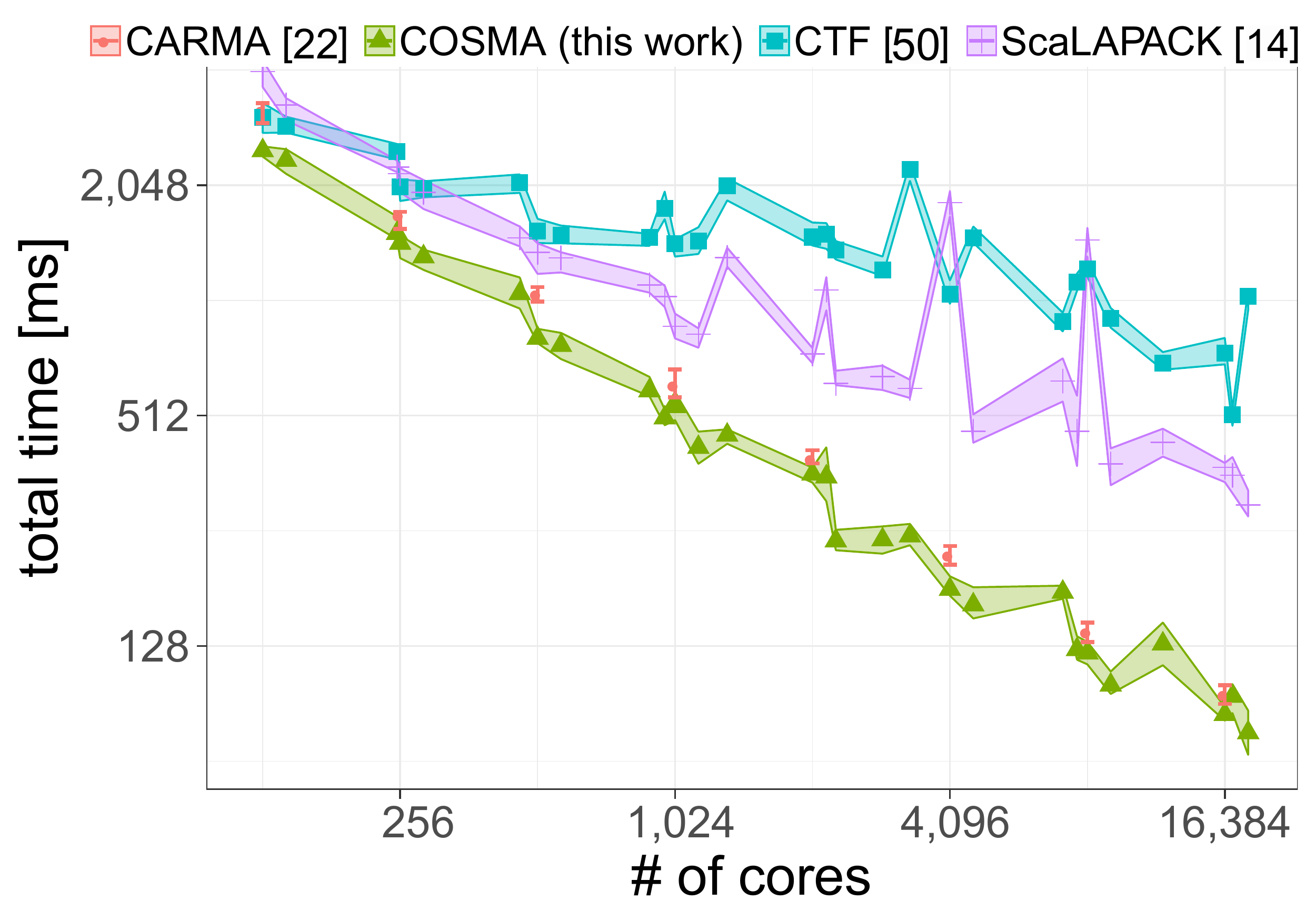}\label{fig:square_strong_time}}
	\hfill
	\subfloat[Limited memory, $n = m = k = 
	\sqrt{\frac{pS}{3}}$]{\includegraphics[width=\fw \textwidth]	
		{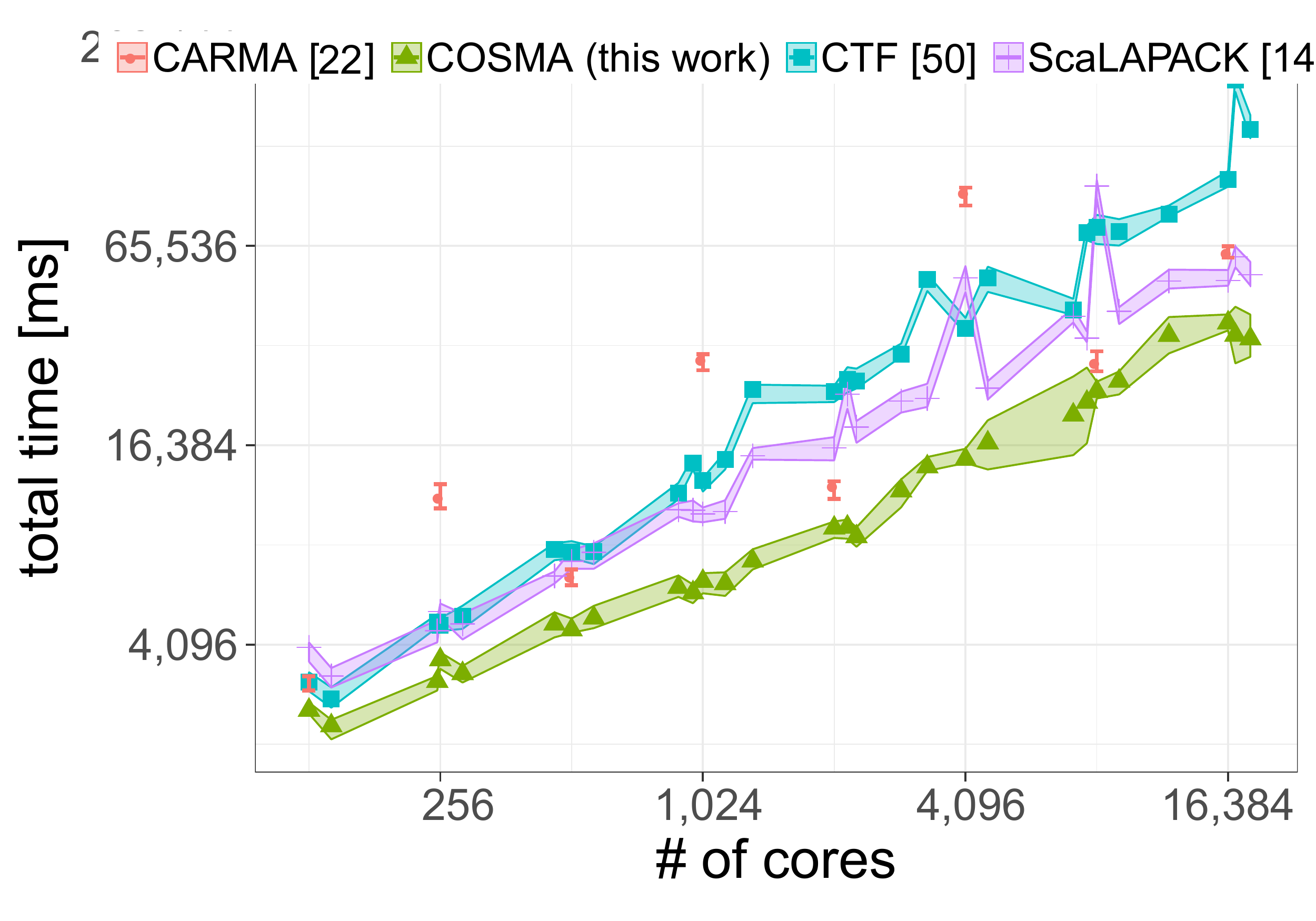}\label{fig:square_weakp1_time}}
	%
	\hfill
	\subfloat[Extra memory, $m = n = k = 
	\sqrt{(p^{2/3}S)/3}$]{\includegraphics[width=\fw \textwidth]		
		{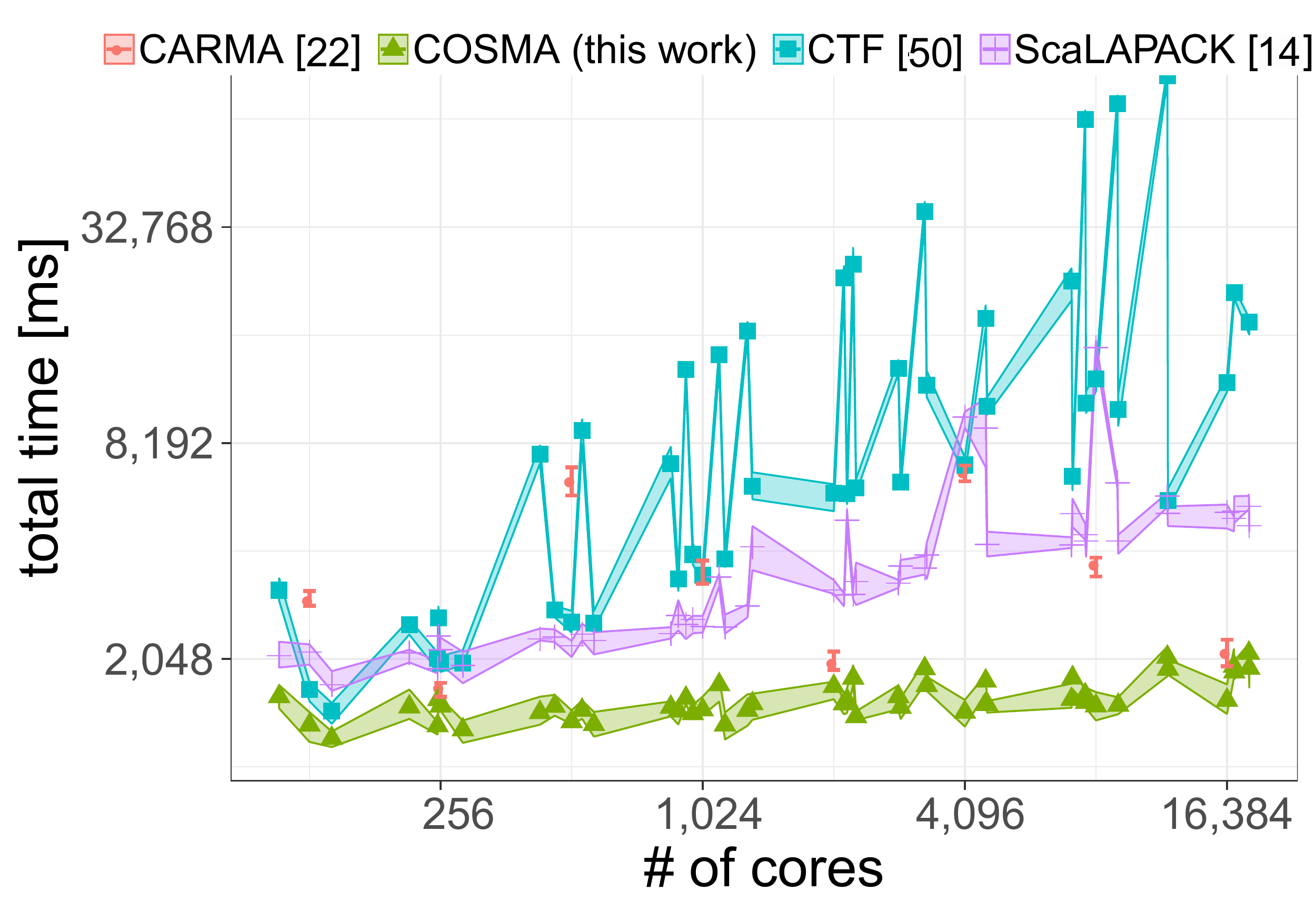}\label{fig:square_weakp2_time}}
	\vspace{-1em}
	\caption{
		\textmd{{Total runtime of COSMA, CTF, ScaLAPACK and 
				CARMA for square matrices, strong and weak scaling. We show 
				median 
				and 95\% confidence intervals.} }
	}
	\vspace{1.5em}
	\label{fig:performancePlotsSquareTime}
\end{figure*}

\begin{figure*}[t]
	\centering
	\subfloat[Strong scaling, $n=m= $17,408, $k= $ 3,735,552
	]{\includegraphics[width=\fw \textwidth]	
		{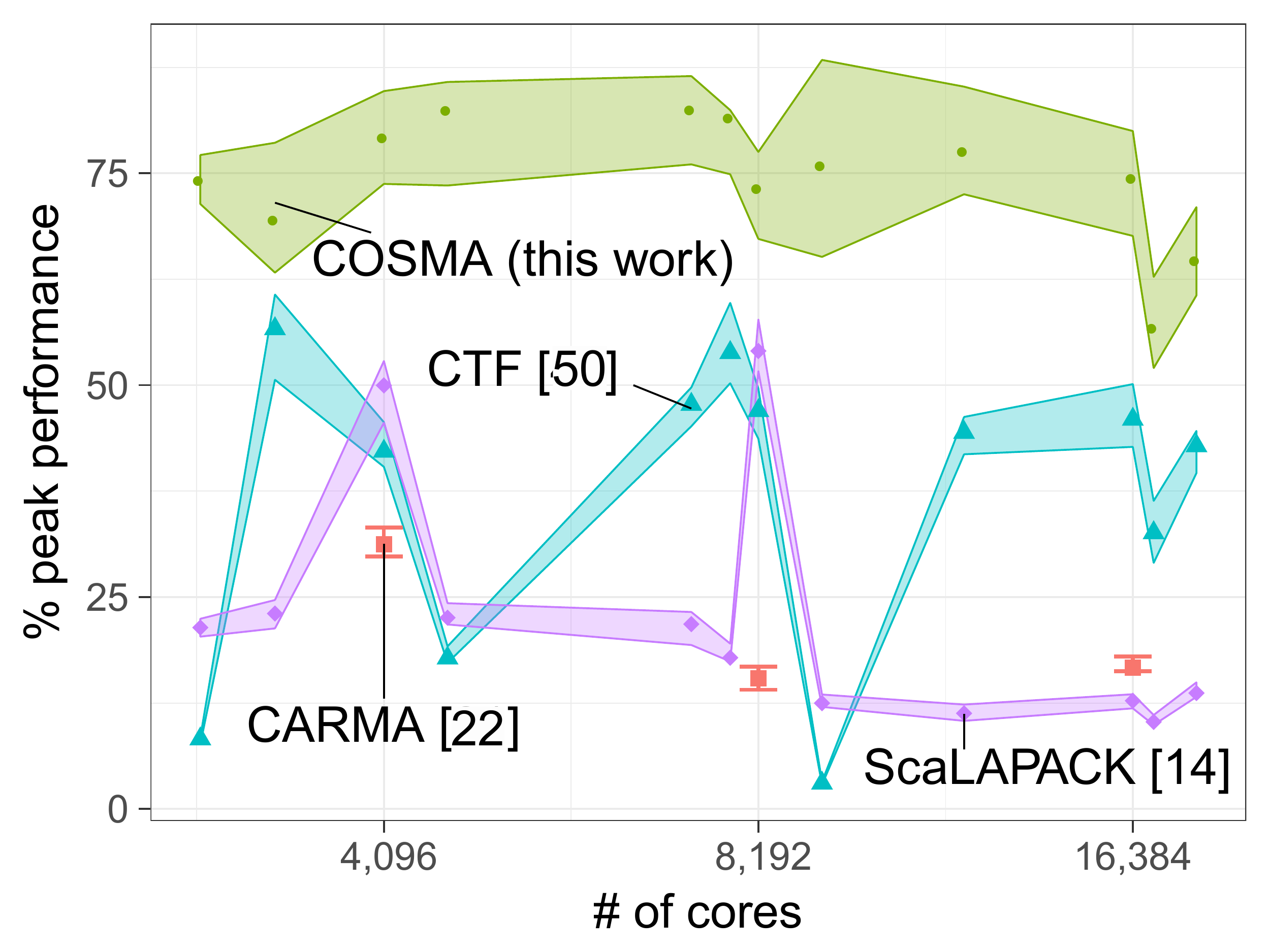}\label{fig:tall_strong}}
	\hfill
	\subfloat[Limited memory, $m=n=979 p^\frac{1}{3}$, 
	$k=$1.184$p^\frac{2}{3}$]{\includegraphics[width=\fw \textwidth]	
		{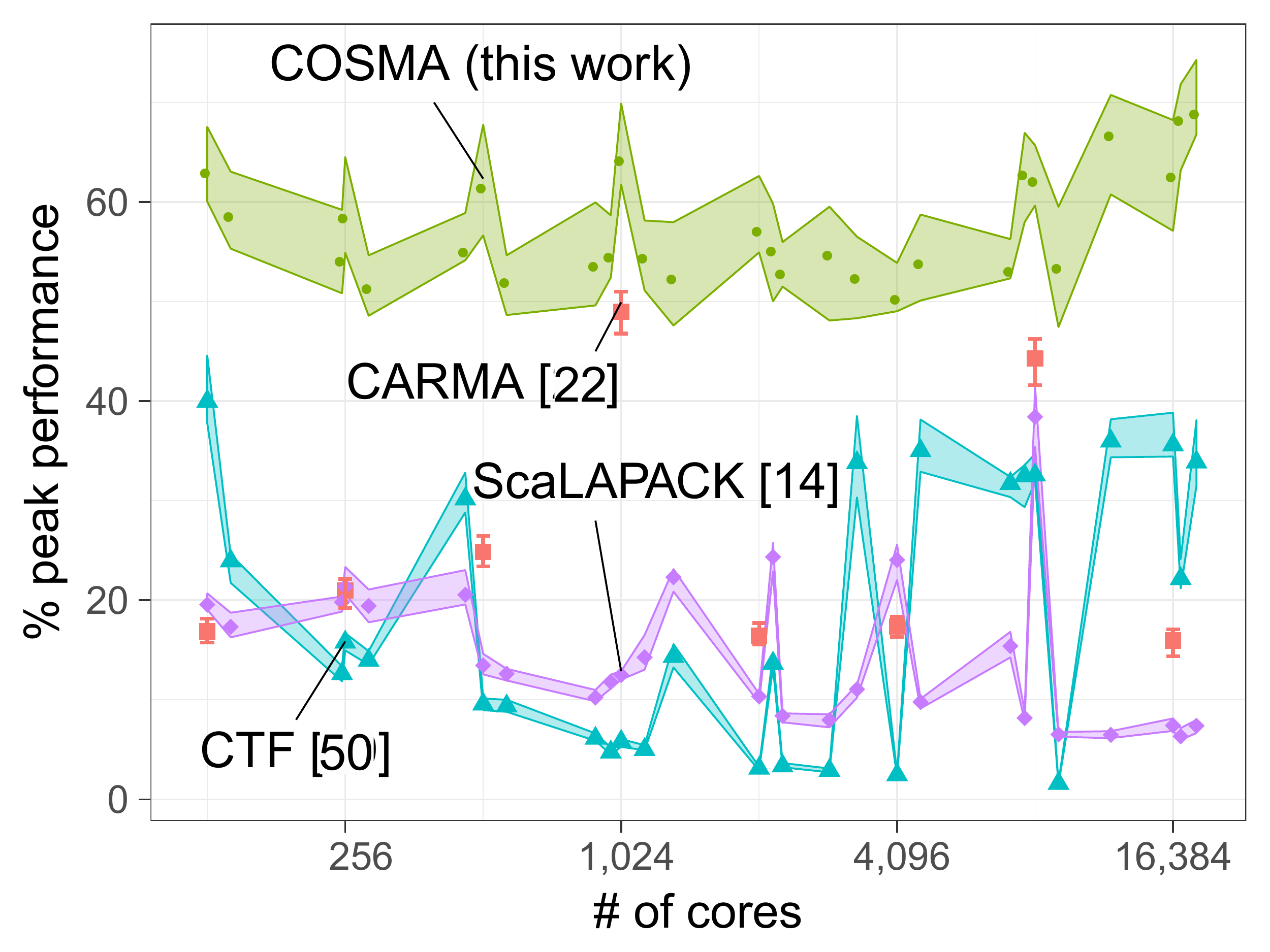}\label{fig:tall_strong}}
	%
	\hfill
	\subfloat[Extra memory,$m=n=979 p^\frac{2}{9}$, 
	$k=$1.184$p^\frac{4}{9}$]{\includegraphics[width=\fw \textwidth]	
		{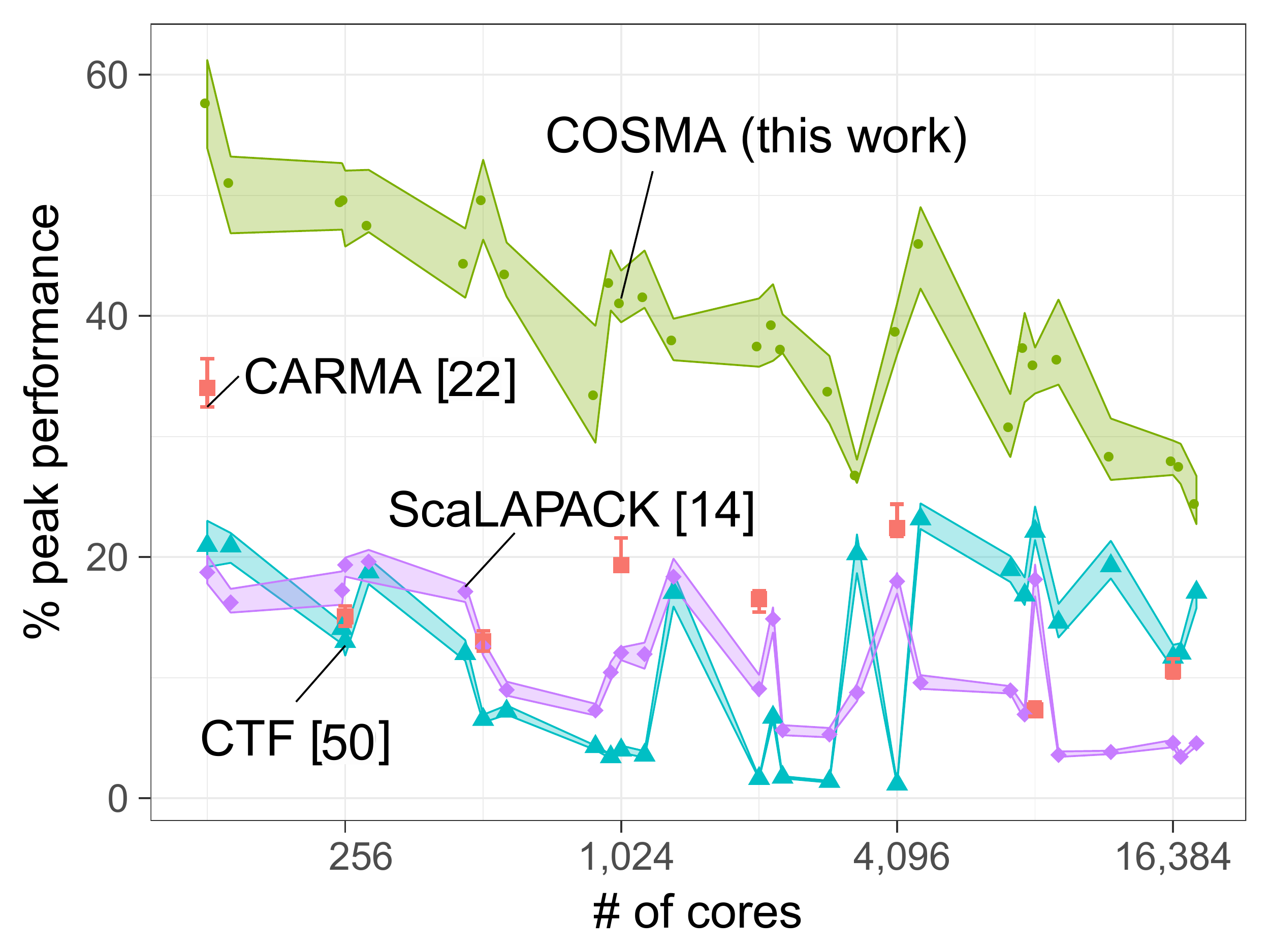}\label{fig:tall_strong}}
	\vspace{-1em}
	\caption{
		\textmd{{Achieved \% of peak performance by COSMA, CTF, 
				ScaLAPACK and 
				CARMA for ``largeK'' matrices. We show 
				median and 95\% confidence intervals.} }
	}
	\vspace{-1.5em}
	\label{fig:performancePlotsLargeK}
	
\end{figure*}

\begin{figure*}
	\centering
	\subfloat[Strong scaling, $n=m= $17,408, $k= $ 3,735,552
	]{\includegraphics[width=\fw \textwidth]	
		{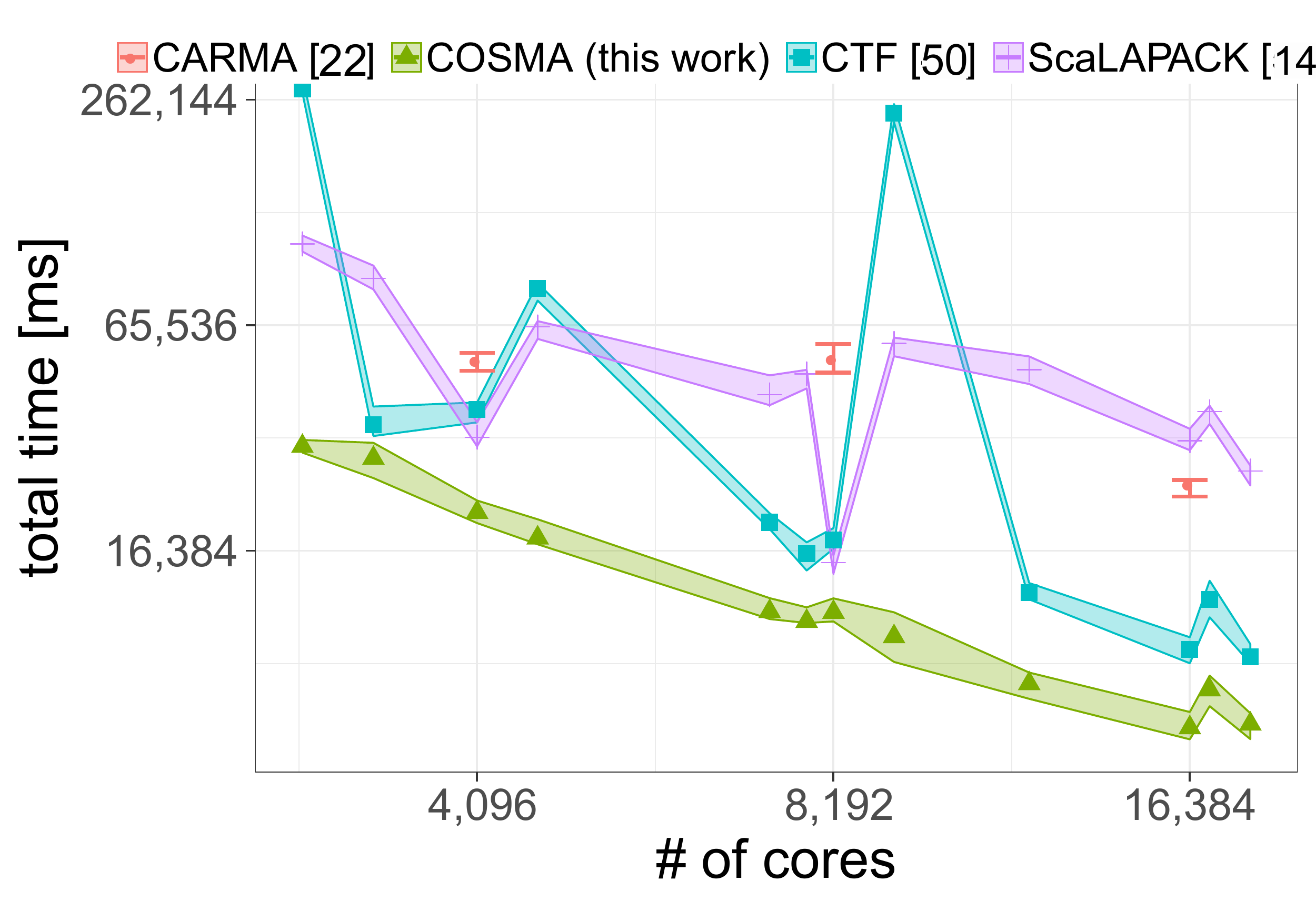}\label{fig:tall_strong_time}}
	\hfill
	\subfloat[Limited memory, $m=n=979 p^\frac{1}{3}$, 
	$k=$1.184$p^\frac{2}{3}$]{\includegraphics[width=\fw \textwidth]		
		{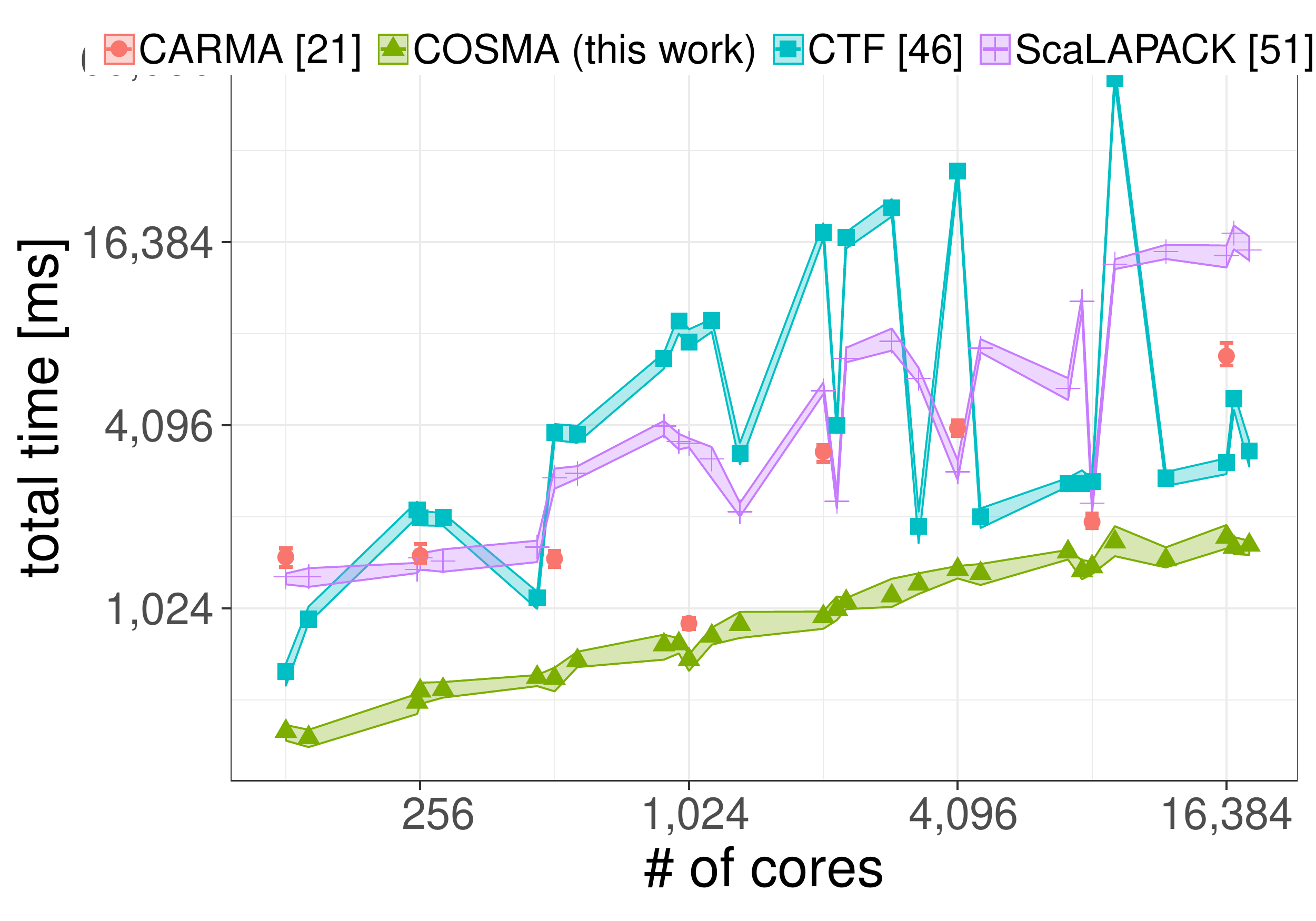}\label{fig:tall_weakp1_time}}
	%
	\hfill
	\subfloat[Extra memory, $m=n=979 p^\frac{2}{9}$, 
	$k=$1.184$p^\frac{4}{9}$]{\includegraphics[width=\fw \textwidth]		
		{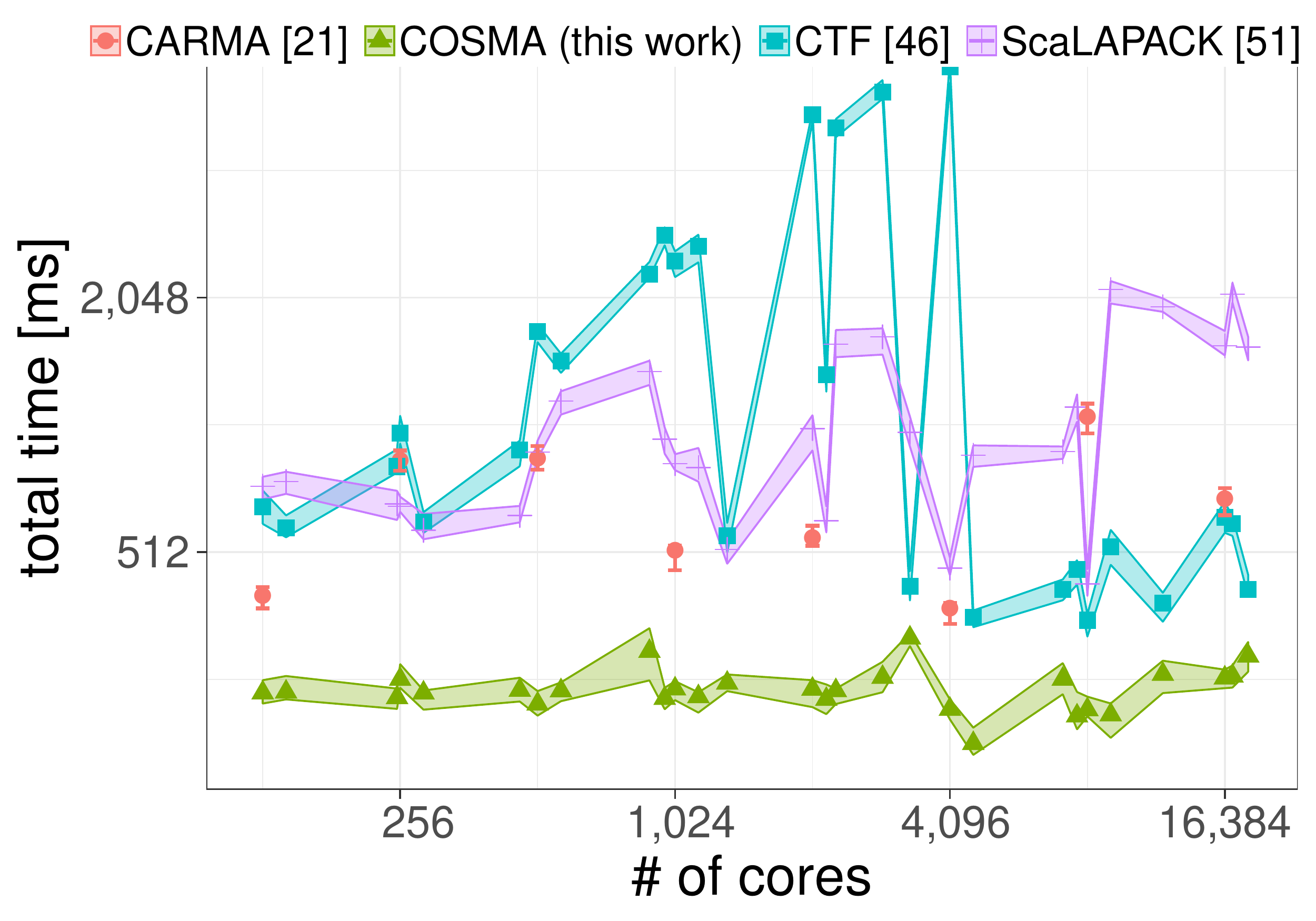}\label{fig:tall_weakp2_time}}
		\vspace{-1em}
	\caption{
		\textmd{Total runtime of COSMA, CTF, ScaLAPACK and 
			CARMA for ``largeK'' matrices, strong and weak scaling. We show 
			median and 95\% confidence intervals.} 
	}
		\vspace{-1.5em}
	\label{fig:performancePlotsLargeKTime}
	
\end{figure*}

\begin{figure}
	\vspace{-1em}
	\subfloat{\includegraphics[width=1.01\columnwidth]
		{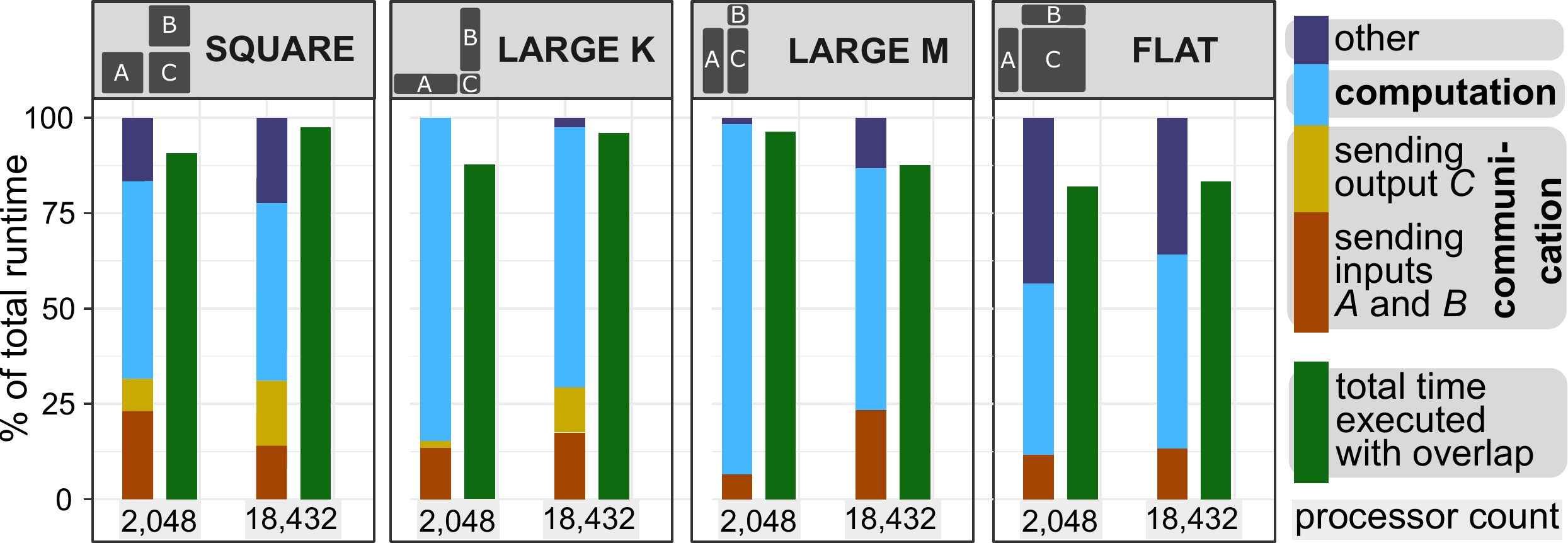}}
	\vspace{-1em}
	\caption{
		\textmd{{Time distribution of COSMA 
				communication and computation
				kernels for strong scaling executed on the smallest and the 
				largest 
				core counts for each of the matrix shapes. Left bar: no 
				communication--computation overlap. Right bar: overlap enabled. 
				}} 
	}
	\vspace{-1.5em}
	\label{fig:performancePlotsBreakdown}
\end{figure}

\begin{figure*}[t!]
	\vspace{-1.0em}
	\subfloat{\includegraphics[width=\fww \textwidth]
		{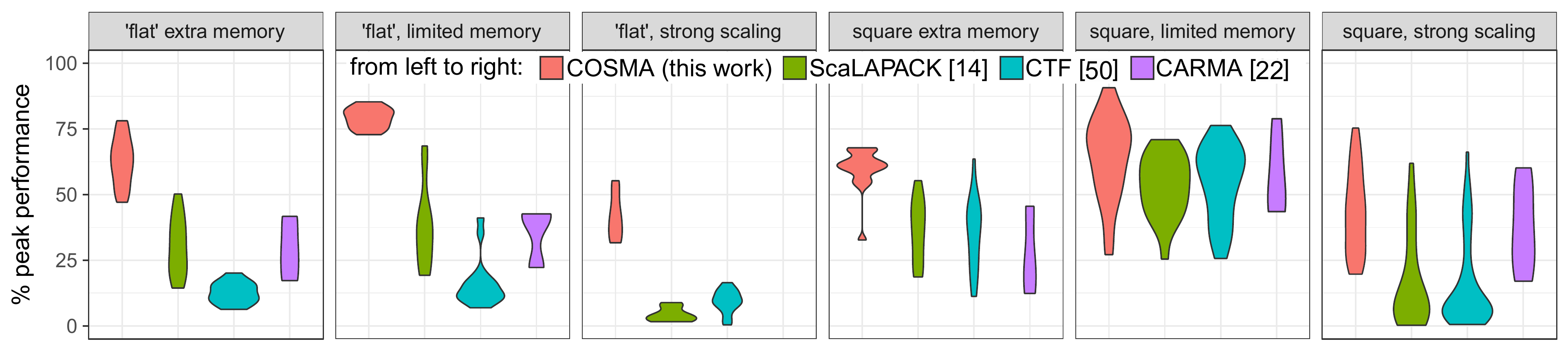}}
	\vspace{-1.0em}
	\caption{
		\textmd{{Distribution of achieved \% of peak performance of 
				the 
				algorithms 
				across all number of cores for ``flat'' and square matrices.}} 
	}
	\vspace{-2.5em}
	\label{fig:performancePlotsDistr}
\end{figure*}

\begin{figure*}[t!]
	\subfloat{\includegraphics[width=\fww \textwidth]
		{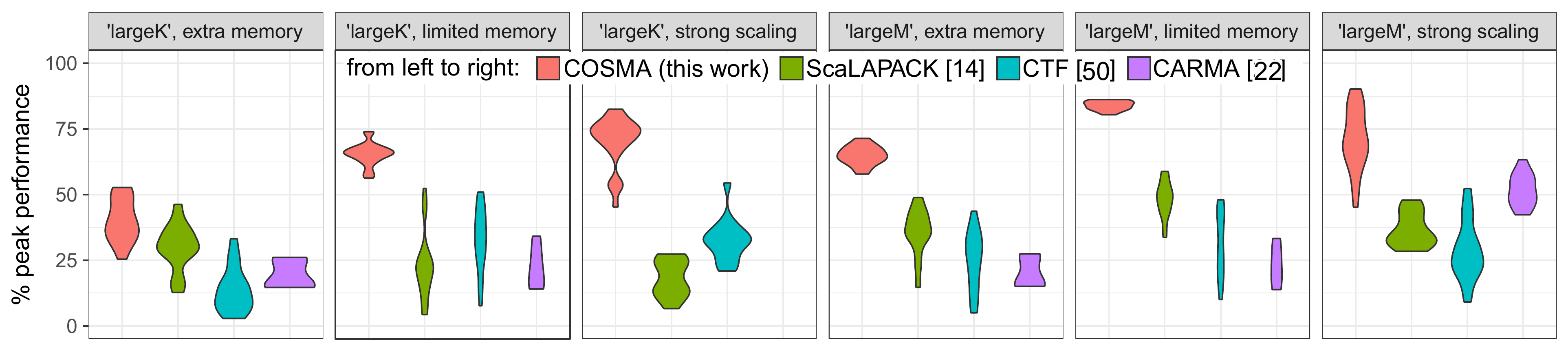}}
	\vspace{-1.0em}
	\caption{
		\textmd{{Distribution of achieved \% of peak performance of 
				the 
				algorithms 
				across all number of cores for tall-and-skinny matrices.} } 
	}
	\vspace{-1.0em}
	\label{fig:performancePlotsDistr2}
\end{figure*}
	
	\vspace{-1em}
	\section{Related work}
	
	Works on data movement minimization may be divided into two 
	categories: applicable across 
	memory hierarchy (vertical, also called I/O minimization), or between 
	parallel 
	processors (horizontal, also called communication minimization). Even 
	though 
	they are ``two sides of the same coin'', in literature they are often 
	treated 
	as 
	separate topics. In our work we combine them: analyze
	trade--offs between 
	communication optimal (distributed memory) and I/O optimal schedule 
	(shared memory).
	
	\subsection{General I/O Lower Bounds}
	Hong and 
	Kung~\cite{redblue} analyzed the I/O complexity for general CDAGs in their 
	the 
	red-blue pebble game, on which we base our 
	work.
	As a special case, they derived an asymptotic bound 
	$\Omega\left({n^3}/{\sqrt{S}}\right)$ for MMM.
	 Elango et al.~\cite{redbluewhite} extended this work to the 
	red-blue-white 
	game and Liu and Terman~\cite{redblueHard_} proved that it is also P-SPACE 
	complete.
	Irony et al.~\cite{IronyMMM} extended the MMM
	lower bound result to a parallel machine with $p$ processors,
	each having a fast private memory of size~$S$, proving the
	$\frac{n^3}{2\sqrt{2}p\sqrt{S}} - S$ 
	lower bound on the communication volume 
	per 
	processor.
	 Chan~\cite{justApebbleGame} studied different variants of pebble 
	games in the context of memory space and parallel time. Aggarwal and 
	Vitter~\cite{externalMem}
	introduced a two-memory machine that models a blocked access and latency in 
	an
	external storage. Arge et al.~\cite{parallelExMem} extended this model to 
	a 
	parallel machine. Solomonik et al.~\cite{edgarTradeoff} combined the 
	communication, synchronization, and computation in their general cost model 
	and 
	applied 
	it to several linear algebra algorithms. Smith and van de 
	Geijn~\cite{tightMMM} derived a sequential lower bound $2mnk/\sqrt{S} -2S$ 
	for MMM. They showed that the leading factor $2mnk/\sqrt{S}$ is tight. We 
	improve this result by 1) improving an additive factor of $2S$, but more 
	importantly 2) generalizing the bound to a parallel machine. Our work uses 
	a 
	simplified model, not taking into account the memory block size, as in the 
	external 
	memory model, nor the cost of computation. We motivate it by assuming that 
	the 
	block size is significantly smaller than the input size, the data is layout 
	contiguously in the memory, and that the computation is evenly distributed 
	among processors.
	
	\subsection{Shared Memory Optimizations}
	I/O optimization for linear algebra includes such techniques as loop tiling 
	and skewing~\cite{tiling}, interchanging and reversal~\cite{tiling2}. For 
	programs with multiple loop nests, Kennedy and McKinley~\cite{loopFusion} 
	showed various techniques for loop fusion and proved that in general this 
	problem is NP-hard. Later, 
	Darte~\cite{loopFusionComplexity} identified cases when this problem has 
	polynomial complexity.
	
	Toledo~\cite{IOsurvey} in his survey on Out-Of-Core (OOC) algorithms 
	analyzed 
	various I/O minimizing techniques for dense and sparse matrices.
	Mohanty~\cite{MohantyThesis} in his thesis optimized several OOC algorithms.
	Irony et 
	al.~\cite{IronyMMM} proved the I/O lower bound of classical MMM on a 
	parallel 
	machine. Ballard et al.~\cite{strassenBounds} proved analogous results for 
	Strassen's algorithm. This analysis was extended by Scott et 
	al.~\cite{generalStrassenBounds} to a general class of Strassen-like 
	algorithms.
	
	Although we consider only dense matrices, there is an extensive literature 
	on 
	sparse matrix I/O optimizations. Bender et al.~\cite{SpMVIO} extended 
	Aggarwal's external memory model~\cite{externalMem} and showed I/O 
	complexity 
	of the sparse matrix-vector (SpMV) multiplication.
	\linebreak
	Greiner~\cite{SpEverything} extended those results and provided  I/O 
	complexities of other sparse computations.
	
	\vspace{-0.5em}
	\subsection{Distributed Memory Optimizations}
	Distributed algorithms for dense matrix multiplication date back to the 
	work 
	of 
	Cannon~\cite{Cannon}, which has been analyzed and extended many 
	times~\cite{MManalysis}~\cite{generalCannon}. In the presence of extra 
	memory, 
	Aggarwal et al.~\cite{summa3d} included parallelization in the third 
	dimension. 
	Solomonik and Demmel~\cite{25d} extended this scheme with their 2.5D 
	decomposition to arbitrary range of 
	the available memory, effectively interpolating between Cannon's 2D and 
	Aggarwal's 
	3D scheme. A recursive, memory-oblivious 
	MMM algorithm was introduced by 
	Blumofe 
	et al.~\cite{recursiveMM} and extended to rectangular matrices by Frigo et 
	al.~\cite{recursiveRectangularMM}. Demmel el al.~\cite{CARMA} introduced   
	CARMA algorithm which achieves the asymptotic complexity 
	for all matrix and 
	memory sizes. We compare COSMA with these algorithms, showing that we 
	achieve 
	better results both in terms of communication complexity and the actual 
	runtime
	performance. 
	Lazzaro et al.~\cite{lazzaroSpMM} used the 2.5D technique for sparse 
	matrices, both for square and rectangular grids. Koanantakool et 
	al.~\cite{sparse15D} observed that for sparse-dense MMM, 1.5D decomposition 
	performs less communication than 2D and 2.5D schemes, as it distributes 
	only 
	the sparse matrix.
	
		\vspace{-0.5em}
	\section{Conclusions}
	In this work we present a new method (Lemma~\ref{lma:reuse}) for assessing 
	tight I/O lower bounds of algorithms using their CDAG representation and 
	the 
	red-blue pebble 
	game 
	abstraction. As a use case, we prove a tight bound for MMM, both for a 
	sequential (Theorem~\ref{thm:seqlowbounds}) and parallel 
	(Theorem~\ref{thm:parSchedule}) execution. Furthermore, our proofs are 
	constructive: our COSMA algorithm is near I/O optimal (up 
	to the factor of 
	$\frac{\sqrt{S}}{\sqrt{S+1}-1}$, which is less than 0.04\% from the lower 
	bound for 10MB of fast 
	memory) for any combination of 
	matrix 
	dimensions, number of processors and memory sizes. This is in contrast with 
	the 
	current state-of-the-art algorithms, which are communication-inefficient in 
	some scenarios. 
	
	To further increase the performance, we introduce a series of 
	optimizations, 
	both on an algorithmic level (processor grid 
	optimization~(\cref{sec:decompArbitrary}) and blocked data 
	layout~(\cref{sec:datalayout})) and hardware-related (enhanced 
	communication 
	pattern~(\cref{sec:commPattern}), 
	communication--computation 
	overlap~(\cref{sec:compOverlap}), one-sided~(\cref{sec:rdma}) 
	communication). The 
	experiments 
	confirm the superiority of 
	COSMA over the other analyzed algorithms - our algorithm significantly 
	reduces communication in all tested scenarios, supporting our theoretical 
	analysis. Most importantly, our work is of practical importance, 
	being maintained as an open-source implementation and
	achieving a 
	time-to-solution speedup of up to 
	12.8x
	times 
	compared to highly optimized state-of-the-art libraries. 
	
	The important feature of our method is that it does not require any manual 
	parameter tuning and is generalizable to other machine models (e.g., 
	multiple levels of memory) and linear algebra kernels (e.g., LU or Cholesky 
	decompositions), both for dense and sparse matrices. We believe that the 
	``bottom-up'' approach will lead to developing more efficient distributed 
	algorithms in the future.

		\vspace{0.5em}
\noindent
\macb{Acknowledgements}

\noindent	
We thank Yishai Oltchik and Niels Gleinig for invaluable  
help with the theoretical part of the paper, and Simon Pintarelli for 
advice and support with the implementation. We also thank CSCS for the compute 
hours needed to conduct all experiments. This project has received funding from 
the European Research Council (ERC) under the European Union’s Horizon2020 
programme (grant agreement DAPP, No.678880), and additional funding from the 
Platform for Advanced Scientific Computing (PASC).

\FloatBarrier
\pagebreak
	
\let\oldbibliography\thebibliography
\renewcommand{\thebibliography}[1]{
  \oldbibliography{#1}
      \setlength{\itemsep}{0pt}
}

\bibliographystyle{ACM-Reference-Format}

\bibliography{mmm-arxiv}	
	
\newpage
\appendix
\section{Change log}
\textbf{10.12.2019}
\begin{itemize}
	\item Added DOI (10.1145/3295500.3356181) of the SC'19 paper version 
	\item Section \emph{4.2.1 Data Reuse}, last paragraph (Section 4 in the 
	SC'19 paper): $W_{B,i}$ corrected to $W_{R,i}$ in the definition of $R(S)$. 
	\item Section \emph{4.2.2 Reuse Based Lemma}, Lemma 2: 
	$q \ge (X - R(S) - T(S)) \cdot (h - 1)$ corrected to $q \ge (X - R(S) + 
	T(S)) \cdot (h - 1)$
	\item Section \emph{4.2.2 Reuse Based Lemma}, Lemma 3 (Section 4, Lemma 1 
	in the SC'19 paper): ``$T(S)$ is the minimum store size'' corrected to 
	``$T(S)$ is the minimum I/O size''
	\item Section 6.2 Parallelization Strategies for MMM, \linebreak 
	\textbf{Schedule} 
	$\mathcal{P}_{ijk}$:  processor grid $\mathcal{G}_{ijk} = 
	\big[\frac{m}{\sqrt{S}}, \frac{n}{\sqrt{S}}, 
	\frac{k}{pS}\big]$ corrected to $\big[\frac{m}{\sqrt{S}}, 
	\frac{n}{\sqrt{S}}, 
	\frac{pS}{mn}\big]$
\end{itemize}
\end{document}